\documentclass[12pt]{article}

\usepackage{tikz}
\usetikzlibrary{calc}
\usepackage[utf8]{inputenc}
\usepackage{textcomp}
\usepackage{newunicodechar}
\newunicodechar{−}{-}
\definecolor{Inflow}{HTML}{2F4C6B}
\definecolor{Outflow}{HTML}{8C4040}
\definecolor{InflowFill}{HTML}{D7E0E8}
\definecolor{OutflowFill}{HTML}{E9DBD8}
\definecolor{GridLine}{HTML}{C9C9C9}
\definecolor{BaseFill}{HTML}{F7F7F8}
\tikzset{
  pcell/.style   = {draw=Inflow!55,line width=0.25pt,fill=InflowFill,text=Inflow,
                    minimum width=3mm,minimum height=3.6mm,inner sep=0pt,font=\scriptsize\bfseries},
  mcell/.style   = {draw=Outflow!55,line width=0.25pt,fill=OutflowFill,text=Outflow,
                    minimum width=3mm,minimum height=3.6mm,inner sep=0pt,font=\scriptsize\bfseries},
  rowlab/.style  = {anchor=east,font=\scriptsize,xshift=-0.8mm},
  glab/.style    = {anchor=center,font=\scriptsize},
  tlab/.style    = {anchor=west,rotate=-90,font=\tiny,text=black!60},
  caption/.style = {anchor=north,font=\normalsize},
}

\usepackage{tikz}
\usepackage{float}
\usepackage{booktabs}
\usepackage{threeparttable}
\usepackage{multirow}
\usepackage{bigdelim}
\usepackage{makecell}
\usepackage{booktabs}
\usepackage{multirow}
\usepackage{bigdelim}
\usepackage{makecell}
\usepackage{adjustbox}
\usepackage{caption}
\usepackage{nicematrix}
\newcommand{\vbeta}[1]{\rotatebox[origin=c]{90}{$\scriptstyle #1$}}
\usepackage[english]{babel}	    
\usepackage{amsmath, amsfonts, amssymb, amsthm, mathtools, bm, bbm} 
\usepackage{icomma}                            
\usepackage{dsfont}

\usepackage[usenames, dvipsnames, svgnames, table, rgb]{xcolor}
\usepackage{colortbl}
\usepackage{chngcntr}
\usepackage{apptools}

\definecolor{BlueBottle}{RGB}{0, 130, 255}
\usepackage{algorithm}
\usepackage{algorithmic}

\usepackage{thmtools}
\usepackage{hyperref}

\newenvironment{continuance}[1]
  {\par\bigskip\noindent\textbf{Example #1 (continued)}\itshape}
  {\par}

\usepackage{graphicx}                               
\graphicspath{{images/}{images2/}}               
\usepackage{wrapfig}                                      
\usepackage{subcaption}

\usepackage{array, tabularx, tabulary, booktabs} 
\usepackage{longtable}                                     
\usepackage{multirow}

\theoremstyle{plain}

\newtheorem{assumption}{Assumption}[section]
\newtheorem{proposition}{Proposition}
\newtheorem{lemma}{Lemma}
\newtheorem{definition}{Definition}[section]

\theoremstyle{definition}                

\newtheorem{example}{Example}

\newtheorem{remark}{Remark}

\theoremstyle{remark}

\counterwithin{lemma}{section}
\counterwithin{theorem}{section}
\counterwithin{equation}{section}

\newcommand{\E}{\mathbb{E}}
\newcommand{\Var}{\mathbb{V}ar}
\newcommand{\Cov}{\mathbb{C}ov}

\newcommand{\G}{\mathbb{G}}
\makeatletter
\renewcommand*\env@matrix[1][\arraystretch]{%
	\edef\arraystretch{#1}%
	\hskip -\arraycolsep
	\let\@ifnextchar\new@ifnextchar
	\array{*\c@MaxMatrixCols c}}
\makeatother

\usepackage{geometry}                  	
\usepackage{setspace}  
\usepackage{hyperref}
\hypersetup{
    colorlinks=true,
    linkcolor=BlueBottle,
    citecolor=BlueBottle,
    urlcolor=BlueBottle
}

\usepackage{natbib}
\usepackage{har2nat}

\usepackage{multicol}                   

\usepackage{pgf, tikz}                           
\usepackage{pgfplots}
\usepackage{pgfplotstable}
\usetikzlibrary{arrows}
\usetikzlibrary[patterns]
\usetikzlibrary{arrows.meta}
\usepackage{pdfpages}

\usepackage{caption}
\usepackage{subcaption}
\usepackage{cleveref}
\usepackage{graphicx}

\usepackage{lscape}
\newcommand{\ci}{\perp\!\!\!\perp}
\newcommand{\supp}{\text{supp}}

\usetikzlibrary{arrows.meta}

	\title{\Large Adaptive Estimation of Aggregated Values of Conditional Linear Programs   }

\author{
Gevorg Khandamiryan\thanks{Email: gevorgkh@berkeley.edu. University of California, Berkeley.} \quad 
Vira Semenova\thanks{Email: vsemenova@berkeley.edu.  First version: March 2023, arXiv ID: 2303.00982.   We are grateful to Victor Chernozhukov, Anna Mikusheva, Bryan Graham, Michael Jansson, Patrick Kline,  Demian Pouzo for their guidance and encouragement. Helpful comments were provided by  Denis Chetverikov,  Bulat Gafarov, Dalia Ghanem, Nail Kashaev,   Ying-Ying Lee, Lihua Lei, Rosa Matzkin,  Andres Santos, Christopher Walters.}
}

\begin{document}
  \maketitle
\date{}

\begin{abstract}
We develop a covariate-assisted approach to partially identified parameters that are solutions to an under-identified system of linear equations with known coefficients. 
Examples include bounds on treatment effects, models of unemployment with state dependence, choice-theoretic models of IV, and random utility models. The boundary (i.e., support function) of the proposed identified set is represented as an average of intersections of regression functions, aggregated over the covariate distribution. We show that the boundary is a regular parameter, propose asymptotic theory, and demonstrate using an empirical application to Jobs First.
\end{abstract}

\noindent\texttt{Keywords:}  Duality, linear programming, support function, partial identification,  intersection bounds, cross-fitting, asymptotic linear representation, stochastic programming
\noindent\texttt{JEL Numbers:} C14, C31, C54

\section{Introduction and Motivation}

Linear programming problems are ubiquitous in economics. They arise
in heterogeneous treatment analysis, multivalued treatments and
instruments \citep{HeckmanPinto, SalanieLee}, sample selection
\citep{HorowitzManski}, discrete choice \citep{TebaldiECMA2023},
state dependence \citep{Torgovitsky}, and random utility models
\citep{KitamuraStoye}, among many others. In many such settings, the
number of identifying restrictions is smaller than the number of
parameters, so point identification fails and worst-case bounds are
the natural target of inference.

Inference on bounds is challenging for two reasons. First,
closed-form expressions for the bounds are rarely available,
motivating case-by-case derivations in individual applications, as
in \citet{KT}'s analysis of Jobs First and
\citet{kamat2021identifying}'s analysis of Head Start. Second, the
linear program may admit multiple optimal solutions --- ``flat
faces'' of the identified set
\citep{Shapiro1991, Dumbgen2003, HsiehShiShum} which create
non-differentiabilities \citep{HiranoPorter2012} that invalidate
standard asymptotic arguments. Both hurdles often discourage the use
of baseline covariates, even when covariates are available and would
plausibly tighten the bounds.

This paper develops a general framework for estimation and inference
on aggregate values of conditional linear programs. The right-hand
side of the linear system is allowed to vary with observed
characteristics and is treated as an unknown nuisance function to be
learned from data. The target parameter is the aggregated support
function, obtained by averaging the covariate-specific bounds over
the covariate distribution. We establish that, when at least one
continuously distributed covariate enters the right-hand side, this
aggregated parameter is regular and pathwise differentiable, with an
influence function that takes a closed form. The cross-fitted plug-in
estimator is root-N consistent and asymptotically normal, with
confidence intervals available through a Gaussian
multiplier-bootstrap procedure. As a leading special case, when the
signal is taken to be the doubly robust signal of \citet{Robins},
the influence function we derive coincides with the efficient
influence function of \citet{LuedtkeLaan} for the sharp upper bound
on the always-takers' share. The framework accommodates first-stage
regularized regression methods or other machine-learning estimators
through cross-fitting.

The paper makes three theoretical contributions. The first is a
closed-form identification result. We establish, via strong LP
duality applied pointwise and a Jensen-type aggregation, that the
aggregated support function equals the expectation, over the
covariate distribution, of the minimum inner product between the
conditional right-hand side and a finite collection of dual
vertices. This aggregated bound is weakly tighter than the bound
based on aggregate data alone. The representation connects the
support-function approach to partial identification with an
aggregated intersection-bound representation, a combination that is
new to the literature. The second contribution is a regular
asymptotic theory. We characterize the influence function of the
aggregated support function under a margin condition that is
generically satisfied when the conditional moments are continuously
distributed, and develop a cross-fitted plug-in estimator together
with a Gaussian multiplier-bootstrap inference procedure whose
uniform coverage is established over an explicit class of
distributions. The third contribution is a self-contained asymptotic
theory for cross-fitted envelope-regression estimators, developed in
the online supplement at the level of an abstract finite index set
and a vector-valued nuisance function. This theory --- comprising
an oracle expansion, a Gaussian approximation, and
multiplier-bootstrap validity --- strictly generalizes Theorem~1 of
\citet{LuedtkeLaan}, which addresses the two-vertex case, and yields
the asymptotic results of the present paper as a direct
specialization to the conditional linear program. 

We demonstrate the proposed method by revisiting the Jobs
First study of \citet{KT}, who report bounds on five
response-probability parameters that summarize how Connecticut's
1996 welfare reform affected women's labor-supply and welfare
participation decisions. Using twenty-eight baseline covariates and
an $\ell_1$-penalized first stage with woman-id cross-fitting, we
recover their bounds without relying on the closed-form derivations
of their Online Appendix~B, providing an alternative methodological
route to their main results. We then use the LP framework to address
a question that \citeauthor{KT} themselves posed:  whether the
above-FPL opt-in response reflects substantive labor-supply
adjustment or trivial earnings reductions of a few dollars from just
above the poverty line. Refining the earnings grid into nine bins
and applying the same identification and inference machinery
verbatim, we find that the largest opt-in lower bound occurs not at
the smallest reduction but for women who would have had to reduce
earnings by twenty to forty percent of the Federal Poverty Line ---
a substantively large adjustment. The opt-in response is therefore
inconsistent with trivial rounding. Beyond this specific finding,
the empirical exercise illustrates how the LP framework makes
outcome-grid refinements  inferentially and computationally tractable.

A central conceptual point distinguishes this paper from earlier
work on linear-program bounds. The existing literature on the
support function approach to partial identification
\citep{BM, BMM, CCMS} studies the support function at a
fixed, aggregate data vector. We instead study the \emph{aggregated}
support function, obtained by averaging the conditional support
function over the covariate distribution. None of the papers in this
strand considers this object. The distinction matters because
aggregation, viewed as an integration operation, smooths the
non-regularity that motivates the literature's specialized inference
procedures \citep[e.g.,][]{FangSantosShaikh, HsiehShiShum}: covariate
values at which the binding dual vertex is non-unique form a measure
zero set under a mild smoothness condition on the covariate
distribution and therefore do not contribute to the first-order
asymptotics. As a consequence, the aggregated support function is a
regular, pathwise differentiable parameter, and we do not need
further regularization to restore asymptotic normality.

The notion of ``flat faces'' that complicates inference in the
existing literature merits clarification in the present setting. In
its general use, the term refers to multiplicity of solutions to the
\emph{primal} linear program. By LP duality, primal-solution
multiplicity is equivalent to the multiplicity of Lagrange
multipliers --- the dual solutions --- which corresponds to a
violation of the Linear Independence Constraint Qualification. Our
setting is a special case of linear systems in which the
left-hand-side coefficient matrix is deterministic and does not
depend on the covariates. In this special case, the derivative of
the value function with respect to the conditional right-hand side
depends only on the dual solution and not on the primal solution.
Multiplicity of primal solutions is therefore immaterial for
inference; we need only exclude the possibility of multiple dual
solutions, which is the content of our key assumption.

\subsection{Literature review}

This paper contributes to the growing literature on bounds arising
from linear programming problems and affine moment inequalities
\citep{HonoreTamer, AndrewsRothPakes, FangSantosShaikh, DongHsiehShum,
HsiehShiShum, JLS}. \citet{KitamuraStoye} employ a dual approach to
test whether observed demand is consistent with a random utility
model. \citet{FangSantosShaikh} and \citet{AndrewsRothPakes} also
rely on duality for inference procedures, in both cases at the level
of the population LP rather than its conditional counterpart. The
closest paper in this line is \citet{HsiehShiShum}, who invoke
duality arguments for a broad class of linear and quadratic
programming problems and develop an inference method that
accommodates ties. By contrast, we study a conditional version of
the linear program in which the right-hand-side covariates are
continuously supported, which generically rules out ties on a set of
positive measure and renders the aggregated support function
regular.

The paper also contributes to the literature on debiased inference
with machine learning for trimmed and bounded functionals: the
covariate-assisted Lee-type procedures of \citet{SemSupp2}, the
nonparametric truncated-mean estimator of
\citet{olma2021nonparametric}, the least-squares approach to
heterogeneous effects of \citet{heiler2024}, the
intensive--extensive margin treatment evaluation of
\citet{heiler2024intensive}, and the continuous-treatment extensions
of \citet{lee2025}. These papers typically rely on closed-form
representations of the target parameter. Our framework instead
focuses on settings where covariates enter through a conditional
linear program with no closed-form solution, providing a general
LP-based route to regular and efficient inference in partially
identified problems with rich covariates.

A second strand of related work is the support-function approach to
partial identification, including \citet{BM, BMM, Gafarov, SemJoE}.
Following \citet{BM}, we adopt the random-set approach to
characterize the covariate-assisted identified set, with randomness
induced by the covariate distribution. \citet{Gafarov} proposes a
different strategy, introducing regularization to handle the
non-differentiability associated with flat faces of the identified
set. Our approach instead leverages continuously distributed
covariates, which generically eliminate flat faces and restore
regularity without the need for regularization. The two approaches
are complementary: regularization offers a way forward in
environments with limited covariates, while covariate assistance
provides a natural route to regular inference when richer covariates
are available.

The most closely related concurrent work is \citet{JLS}. Eight
months after the first version of the present paper was posted,
\citet{JLS} developed a model-agnostic covariate-assisted inference
framework for partially identified causal effects using optimal
transport duality, with a central role for weak duality. Both papers
exploit the same fundamental observation: conditioning on covariates
and averaging covariate-specific dual bounds yields a bound that is
weakly tighter than the bound obtained from population-average data.
The two frameworks are structurally distinct and complementary, and
they make different trade-offs between sharpness and robustness.
The present paper delivers semiparametrically efficient,
square-root-of-N inference for the sharp aggregated bound under a
margin condition and consistent first-stage estimation, with an
explicit closed-form influence function. \citet{JLS} relax the
consistency requirement on the first stage by exploiting weak
duality directly: any dual-feasible selector delivers a valid bound
in expectation, regardless of whether the first-stage estimator
converges to the true conditional moments. The cost of this
robustness is that the resulting bound is one-sided and need not be
sharp. The two approaches are thus complementary: the optimal
transport framework offers broader scope (continuous outcomes,
model-agnostic validity under misspecification), while the finite-LP
framework of the present paper provides a direct route to
semiparametrically efficient inference with explicit influence
functions in the large class of discrete-outcome economic models in
which a linear-system structure is available.

\section{Setup}
\label{sec:setup}

Consider a system of linear equations
\begin{align}
\label{eq:linearcons}
 A \bm\beta_0 = \bm{b}_0, \qquad  \bm\beta_0 \geq 0.
\end{align}
Here,  ``\(\bm\beta_0 \geq 0\)'' indicates that all coordinates of the parameter vector \(\bm\beta_0 \) are nonnegative.
The parameter \(\bm\beta_0 = \E [ \bm{\Pi} ]\) is a $d$-vector summarizing the unobserved heterogeneity by taking an expectation of an unobserved random vector \(\bm{\Pi} \in \mathbf{R}^d\).
 The $k$-vector \(\bm{b}_0 = \E [ \mathbf{B} ]\) is the expectation of a  vector \(\mathbf{B}\) that is an unknown yet estimable parameter.
The matrix \(A\) is a known, deterministic \(k\times d\) matrix. We focus on an empirically relevant case when  \(k<d\), which implies that the vector \(\bm\beta_0\) may not be point-identified.

The paper studies projections of the partially identified vector $\bm\beta_0$ onto various directions of economic interest. For example, an upper bound on the \(j\)th coordinate (\(j\in\{1,\dots,d\}\)) is obtained by solving
\begin{align}
\label{eq:lpbasic}
  \max_{\bm\beta_0}\; \bm e_j^\top \bm\beta_0
  \quad\text{subject to}\quad
  A\,\bm\beta_0 = \bm{b}_0,\;\;\bm\beta_0 \ge 0,
\end{align}
where $\bm e_j$ is the $j$th standard basis vector in $\mathbf{R}^d$. A corresponding lower bound can be obtained by replacing \(\max_{\bm\beta_0} \bm e_j^\top \bm\beta_0\) with \(-\max_{\bm\beta_0} -\bm e_j^\top \bm\beta_0\) under the same constraints. More generally, for a given direction $q$ on a unit sphere
$$\mathcal{S}^{d-1} = \{ q \in \mathbf{R}^d, \ \|q \| =1 \},$$
 the projection of the identified set onto $q$ is given by 
\begin{align}
\label{eq:lpbasicq}
 \max_{\bm\beta_0} q' \bm\beta_0  \quad\text{subject to}\quad
  A\,\bm\beta_0 = \bm{b}_0,\;\;\bm\beta_0 \ge 0.
\end{align}
In order to fix the ideas, we discuss several empirically relevant examples of the  linear system \eqref{eq:linearcons}. 
%

\subsection{Motivating Examples}
\label{sec:mot}

\begin{example}[Principal stratification]
\label{ex:hp1}
Let $X$ denote a vector of baseline covariates taking values in the set $\mathcal{X}$. Let $Z$ denote a randomly assigned instrument and  $Y \in \mathcal{Y} \subset \mathbf{R}^{\dim Y}$ be the vector of  endogenous variables, where $Z$ and $Y$ are assumed to have finite support.  Let $\mathbf{U} \in \mathcal{U}$ denote a complete characterization of unobserved heterogeneity (e.g., a collection of counterfactual outcomes) with finite support $\supp \mathbf{U} = \{ \mathbf{u}_1, \dots, \mathbf{u}_{N_U} \}$. Given $\mathbf{U}$ and $Z$, the endogenous variables $Y$ must be non-random.   The instrument $Z$ is assumed to be completely independent of $\mathbf{U}$ and $X$, that is
\begin{align}
\label{eq:indep}
(\mathbf{U}, X) \ci Z.
\end{align}
The target parameter \(\bm\beta_0\) is the $N_U$-dimensional vector of latent response-type probabilities, 
\(\Pr (\mathbf{U} = \mathbf{u})_{\mathbf{u} \in \supp \mathbf{U}}\). For any set $R \subseteq \supp Y$, \begin{align}
\Pr (Y \in R \mid Z=z) 
&= \sum_{\mathbf{u} \in \supp \mathbf{U}} \Pr( Y \in R \mid \mathbf{U} = \mathbf{u}, Z=z )\, \Pr (\mathbf{U} = \mathbf{u} \mid Z=z) \label{eq:propscore2} \\
&= \sum_{\mathbf{u} \in \supp \mathbf{U}} \Pr( Y \in R \mid \mathbf{U} = \mathbf{u}, Z=z )\, \Pr (\mathbf{U} = \mathbf{u}). \nonumber
\end{align}
Special cases of this Example include IV bounds in  \cite{BalkePearl1994,BalkePearl1997},  \citep*{LeeBound} bounds with discrete-valued outcomes.
\end{example}

The linear system in Example~\ref{ex:hp1} --- and, more generally, in
most examples in this paper --- encodes a conservation-of-mass
condition: each equation matches the probability of an observable
outcome value to a sum of latent response-type probabilities
consistent with that value. The number of equations therefore equals
the cardinality of the outcome support, and a finite system requires
$Y$ to have finite support. When $Y$ is continuously distributed, the RHS function  becomes an infinite-dimensional object and the conditional
LP is replaced by an infinite-dimensional moment problem requiring
optimal-transport techniques \citep{JLS}. Such extensions are
outside the scope of the present paper. In applied work, researchers
commonly discretize continuously distributed outcomes for
tractability, as in \citet{KT} and our
Section~\ref{sec:illustration}.

When both the treatment and the outcome are binary, Example \ref{ex:hp1} reduces to a model  studied by \citet{HSC,Manski}. We spell it out as a separate example. 

\begin{example}
\label{ex:hsc}
Let $Z=D=1$ be an indicator of exogenous binary treatment, and let $S$ be a binary outcome (here, $Y=S$).   The propensity score equation \eqref{eq:propscore2} reduces to
\begin{align}
\Pr (\mathbf{U} =(1, 1)) + \Pr (\mathbf{U} =(1, 0)) &= \Pr (S=1 \mid D=1) \label{eq:eq1} \\
\Pr (\mathbf{U} =(1, 1)) + \Pr (\mathbf{U} =(0, 1)) &=  \Pr (S=1 \mid D=0)  \label{eq:eq2}.
\end{align}
This system is a special case of \eqref{eq:linearcons} with 
\begin{align}
\label{eq:linearconshsc} 
A = \begin{pmatrix}
1 & 1 & 0 & 0 \\
1 & 0 & 1 & 0  \\
1 & 1 & 1 & 1 \\
\end{pmatrix}, \quad \bm{b}_0 = \begin{pmatrix} \Pr (S=1 \mid D=1) \\  \Pr (S=1 \mid D=0) \\ 1 \end{pmatrix}  = \begin{pmatrix} s(1) \\  s(0) \\ 1 \end{pmatrix}.
\end{align}
\end{example}

Example~\ref{ex:hsc} serves as the running example throughout the rest of the paper. Beyond this binary-outcome benchmark, the framework also accommodates a range of applied settings. 

\begin{example}[Choice-theoretic model of IV]
\label{ex:hp2}
Consider a model generated by a system of treatment and outcome equations as studied in \citet{HeckmanPinto}:
\begin{align}
D &= f_D(X, Z, V), \nonumber \\
Y &= f_Y (X, D, V, \epsilon_Y), \nonumber \\
X, V , \epsilon_Y, Z &\text{ are mutually independent}. \label{eq:indep2}
\end{align}
Here $Z$ is an exogenous instrument, $D$ is an endogenous treatment, and
\[
Y = \sum_{d \in \mathcal{D}} Y(d)\, \bm{1}\{D=d\}
\]
is the observed outcome.
The response vector $\mathbf{U} = (D(1), \dots, D(n_Z))$ encodes potential treatment assignments across instrument values $Z \in \mathcal{Z}$.
The target parameter is the average potential outcome for a given response type $\mathbf{u}$:
\begin{align}
\label{eq:ratiohp}
\E[ Y(d) \mid \mathbf{U} = \mathbf{u}]
= \frac{\sum_{y \in \supp Y} y \, \Pr(Y(d)=y, \mathbf{U}=\mathbf{u})}{\Pr(\mathbf{U}=\mathbf{u})}.
\end{align}

The independence assumption \eqref{eq:indep2} implies, for any $d \in \mathcal{D}$ and $z \in \mathcal{Z}$,
\begin{align}
\label{eq:propscore2hp}
\Pr(D=d \mid Z=z)
= \sum_{\mathbf{u} \in \supp(\mathbf{U})}
   \Pr(D=d \mid \mathbf{U}=\mathbf{u}, Z=z)\, \Pr(\mathbf{U}=\mathbf{u}),
\end{align}
which is parallel to \eqref{eq:propscore2}.
It also implies
\begin{align}
&\Pr(\bm{1}\{Y \in R\}\cdot \bm{1}\{D=d\} \mid Z=z) \nonumber \\
&\quad = \sum_{\mathbf{u} \in \supp(\mathbf{U})}
   \Pr(D=d \mid \mathbf{U}=\mathbf{u}, Z=z)\,
   \Pr(Y(d)\in R, \mathbf{U}=\mathbf{u}), \quad z \in \mathcal{Z}.
   \label{eq:propscore3hp}
\end{align}

As discussed in \citet{HeckmanPinto}, these equations can be written as linear programs.
If the denominator of \eqref{eq:ratiohp} is identified, then bounds on $\E[Y(d) \mid \mathbf{U}=\mathbf{u}]$ are equivalent to bounds on the numerator $\sum_{y \in \supp Y} y \Pr(Y(d)=y, \mathbf{U}=\mathbf{u})$, which is a special case of \eqref{eq:lpbasic} with $q=(y_1, \dots, y_{\#\supp Y})$ and $\bm{\beta}_0=(\Pr(Y(d)=y, \mathbf{U}=\mathbf{u}))_{y \in \supp Y}$.
\end{example}

\begin{example}[Survey Response]
\label{ex:response}
Consider the setup of Example~\ref{ex:hp2}. In addition to the instrument and treatment, suppose we observe a survey response variable
\[
S \;=\; \sum_{d \in \mathcal{D}} \sum_{z \in \mathcal{Z}} \bm{1}\{D=d, Z=z\}\, S(d, z),
\]
where $S(d, z)$ is a binary indicator of whether the potential outcome $Y(d)$ is observed. In contrast to job training, where the employment indicator  $S(d, z)=S(d)=\bm{1}\{Y(d)>0\}$ satisfies an exclusion restriction (e.g., \cite{ChenFlores}) by design, the indicator $S(d, z)$ is unrestricted.   The observed data consist of $(Z, X, D, S, S\cdot Y)$, where $Y$ is assumed to have finite support.

Define the unobserved heterogeneity vector $\mathbf{U}$ as the $N_Z (N_D+1)$-dimensional random vector collecting the counterfactual treatment assignments and response indicators:
\[
\mathbf{U}
= \bigl(D(z_1), \dots, D(z_{N_Z}), (S(d, z))_{d \in \mathcal{D},\,z \in \mathcal{Z}}\bigr)'.
\]
Then, analogously to \eqref{eq:propscore3hp}, we obtain
\begin{align}
\Pr(S=1, D=d \mid Z=z)
&=\sum_{\mathbf{u} \in \mathcal{U}}
   \Pr(S(d, z)=1, D=d \mid \mathbf{U}=\mathbf{u}, Z=z)\,
   \Pr(\mathbf{U}=\mathbf{u} \mid Z=z) \label{eq:propscore_app_resp} \\
&=\sum_{\mathbf{u} \in \mathcal{U}}
   \Pr(S(d, z)=1, D=d \mid \mathbf{U}=\mathbf{u}, Z=z)\,
   \Pr(\mathbf{U}=\mathbf{u}), \nonumber
\end{align}
where the last equality uses the independence $\mathbf{U} \ci Z$.  In this setup, it is most informative to focus on treatment effect parameters for the always-observed principal strata. For example, one may consider the average potential outcome for always-observed compliers:
\[
\E\!\big[ Y(d) \,\big|\,
(S(d, z)=1)_{d \in \mathcal{D},\,z \in \mathcal{Z}}, \, D(1)=1, \, D(0)=0\big].
\]
\end{example}

We include below a final core example drawn from the random utility literature. 

\begin{example}[Random utility models, \citep*{KitamuraStoye,KitamuraStoye2019}]
\label{ex:rum}
A random utility model (RUM, \citep*{KitamuraStoye,KitamuraStoye2019}) partially identifies counterfactual demand from a repeated cross-section of prices, budgets, and observed product choices. Suppose there are $K$ goods and $J$ budgets. Given a vector $p \in \mathbf{R}^K$, the budget set is $B(p) = \{y \in \mathbf{R}^K : p'y = 1\}$ and the consumption bundle satisfies $Y \in B(p)$. Let $A$ be a binary matrix whose columns represent possible non-stochastic demand systems, $\bm{\beta}_0 \ge 0$ a probability distribution over demand systems, and $\bm{b}_0$ a vector of observed factual shares $(\Pr(Y \in V \mid P = p_j))_{j \in J,\, V \in \mathcal{V}}$. \citet{KitamuraStoye2019} propose bounds on functions of counterfactual demand; our framework delivers debiased inference for such bounds when baseline covariates $X$ are available.
\end{example}

\subsection{Covariates}
\label{sec:2}

This section extends the linear program \eqref{eq:linearcons} to incorporate covariates. 
Specifically, suppose the system can be written as the conditional \textit{linear program}
\begin{align}
\label{eq:condineq}
A \bm\beta_0(x) = \bm{b}_0(x), \qquad x \in \mathcal{X},
\end{align}
where $\bm\beta_0(x) = \E[\bm{\Pi} \mid X=x]$ is the $d$-dimensional vector function summarizing unobserved heterogeneity, 
and $\bm{b}_0(x) = \E[\bm{B} \mid X=x]$ is the $k$-dimensional vector function that is estimable. 
This assumption is high level and must be verified on a case-by-case basis. 
Below, we revisit several motivating examples and show how the conditional linear restriction \eqref{eq:condineq} arises from first principles. 

\begin{continuance}{\ref{ex:hp1}}
The independence assumption \eqref{eq:indep} implies
\begin{align}
&\Pr (Y \in R \mid Z=z, X=x)  \\
&= \sum_{\mathbf{u} \in \supp \mathbf{U}}
  \Pr( Y \in R \mid \mathbf{U}=\mathbf{u}, Z=z, X=x )\,
  \Pr(\mathbf{U}=\mathbf{u} \mid Z=z, X=x) \nonumber \\
&= \sum_{\mathbf{u} \in \supp \mathbf{U}}
  \Pr( Y \in R \mid \mathbf{U}=\mathbf{u}, Z=z )\,
  \Pr(\mathbf{U}=\mathbf{u} \mid X=x). \nonumber
\end{align}
The independence assumption \eqref{eq:indep} implies that \eqref{eq:condineq} holds with $\bm{b}_0(x) = (s(1, x), s(0, x), 1)'$ where $s(d, x) = \Pr (S=1 \mid D=d, X=x)$ for $d \in \{1, 0\}$ and $A$ as in \eqref{eq:linearconshsc}. \end{continuance}

For any $q \in \mathcal{S}^{d-1}$, consider the conditional linear program
\begin{align}
\label{eq:standardform}
\sigma(q,x) := \max_{\bm\beta_0} \; q' \bm\beta_0 
\quad \text{s.t. } A \bm\beta_0 = \bm{b}_0(x), \; \bm\beta_0 \geq 0.
\end{align}
Averaging over the distribution of covariates gives
\begin{align}
\label{eq:targetlp}
\sigma(q) = \E[\sigma(q,X)],
\end{align}
which characterizes the boundary of the covariate-assisted \textit{identified set}. This set is
\begin{align}
\mathcal{B} 
= \bigcap_{q \in \mathcal{S}^{d-1}} \{ b \in \mathbf{R}^d : q'b \leq \sigma(q) \}, 
\label{eq:idset}
\end{align}
where $\sigma(q)$ is given by \eqref{eq:standardform}--\eqref{eq:targetlp}.  
Invoking the argument of \citet*{Artstein} (cf.\ \citet*{BM}, Definition 5, p.\ 771) yields the following Proposition:

\begin{proposition}[Covariate-assisted identified set]
\label{lem:identification}
The set $\mathcal{B}$ is convex and compact and has support function $\sigma(q)$ in \eqref{eq:targetlp}.
Equivalently, $\mathcal{B}$ consists of points of the form $\bm{\beta}_0 = \E[\bm{\beta}_0(X)]$,
where $\bm{\beta}_0(x)$ satisfies \eqref{eq:condineq}.
\end{proposition}

Proposition \ref{lem:identification} shows that $\mathcal{B}$ can be characterized in two equivalent ways:
(i) as the average of partially identified solutions to \eqref{eq:condineq} (random set approach, \citep*{BM,BMM2}),
or (ii) as the intersection of supporting hyperplanes defined by conditional projections of \eqref{eq:condineq} (support function approach, \citep*{BM,BMM,CCMS}).

\section{Theoretical Results}

Section \ref{sec:review} gives an overview of strong duality and presents a dual representation of the boundary, which  facilitates estimation and inference.
Section \ref{sec:mr2} derives an influence function for the boundary.

\subsection{Dual Identification}
\label{sec:review}

Given a direction $q \in \mathbf{R}^d$, the linear system \eqref{eq:standardform} can be written as a standard-form \textit{linear program} (LP): 
\begin{align*}
\sigma(q,x) &= \max \; q' \bm\beta_0 \\
\text{subject to } \quad A \bm\beta_0 &= \bm{b}_0(x), \\
\bm\beta_0 &\geq 0.
\end{align*}
The data enter the problem only through the expectation function $\bm{b}_0(x) = \E[ \bm{B} \mid X=x]$. The primal optimal value $\sigma(q, x)$ is assumed finite for every covariate value $x \in \mathcal{X}$.   The dual LP is
\begin{align*}
\sigma_D(q, x) &= \min \; \nu' \bm{b}_0(x)  \\
\text{subject to } \quad A' \nu - \lambda - q &= 0, \\
\lambda &\geq 0,
\end{align*}
where $\nu \in \mathbf{R}^k$ and $\lambda \in \mathbf{R}^d$ are the dual variables associated with the equality and inequality constraints. Eliminating the slack vector $\lambda$ via $\lambda = A'\nu - q \geq 0$ gives the inequality form of the dual feasible set
\begin{align}
\label{eq:dualset}
A' \nu \geq q,
\end{align}
which is a data-free, covariate-free convex polytope with a finite vertex set $\mathcal{T}=\mathcal{T}(q)$, that is, the set of extreme points of $\{\nu \in \mathbb{R}^k : A'\nu \geq q\}$.  Because most empirical questions fix a direction $q$, we suppress the argument and write $\mathcal{T}$ for $\mathcal{T}(q)$ whenever no ambiguity arises. For the dual variable as \textit{any} minimizer
\begin{align}
\label{eq:dualshadowprice1}
\nu_0(x) \in \arg \min_{\nu \in \mathcal{T}} \nu' \bm{b}_0(x),
\end{align}
$\sigma_D(q, x) = \nu_0(x)' \bm{b}_0(x)$ is the dual optimal value. For linear programs, the primal and dual optimal values coincide:
\begin{align}
\label{eq:strongduality}
\sigma(q,x) = \sigma_D(q,x) = \nu_0(x)' \bm{b}_0(x),
\end{align}
that is, strong duality holds.  Aggregating $\sigma(q, x)$ over the covariate distribution gives a dual representation for $\sigma(q)$:
\begin{align}
\label{eq:mainpsi2}
\sigma(q) = \E[ \sigma(q, X) ] = \E[ \nu_0(X)' \bm{b}_0(X)].
\end{align}

\begin{continuance}{\ref{ex:hsc}}
The upper Fréchet–Hoeffding bound on the always-takers’ share is
\[
\E_X \big[\min (s(0, X), s(1, X))\big] 
= \E \big[\nu_0(X)' \bm{b}_0(X)\big],
\]
where the expectation function is 
$\bm{b}_0(x) = (s(1, x), \; s(0, x), \; 1)'$.   The dual value $\nu_0(x)$ reduces to
\begin{align}
\nu_0(x) = 
\begin{cases} 
(0, 1, 0)' & s(1, x) > s(0, x), \\[6pt]
(1, 0, 0)' & s(1, x) < s(0, x), \\[6pt]
w (1, 0, 0)' + (1-w)(0, 1, 0)', & s(1, x) = s(0, x), \;\; w \in [0, 1].
\end{cases}
\end{align}
In the first two cases, where $s(1, x) \neq s(0, x)$, the binding vertex is unique. In the third case, when $s(1, x)=s(0, x)$, the binding vertex is not unique: the dual value $\nu_0(x)$ can take any value on the line segment $\{w (1, 0, 0)' + (1-w)(0, 1, 0)': w \in [0, 1]\}$.
\end{continuance}

\begin{proposition}
\label{cor:tightness}
Let $\mathcal{T}$ be the set of vertices of the dual feasible set  $\{\nu \in \mathbf{R}^k: A' \nu \geq q\}$ defined in \eqref{eq:dualset}.   The following statements hold:

(1) The basic (i.e., no-covariate) boundary is an intersection bound
\[
\bar{\sigma}(q) = \inf_{\nu \in \mathcal{T}} \nu' \E[\bm{b}_0(X)] 
= \inf_{\nu \in \mathcal{T}} \nu' \bm{b}_0,
\]

(2) The covariate-assisted boundary is an aggregated intersection bound
\begin{align}
\label{eq:mainpsi3}
\sigma(q) 
&= \E \big[ \inf_{\nu \in \mathcal{T}} \nu' \bm{b}_0(X)\big] 
= \E \big[\nu_0(X)' \bm{b}_0(X)\big] 
= \E \big[\nu_0(X)' \mathbf{B}\big].
\end{align}

(3) The covariate-assisted boundary is weakly tighter than the basic one:
\begin{align}
\label{eq:weakinequality}
\sigma(q)
= \E \big[\inf_{\nu \in \mathcal{T}} \nu' \bm{b}_0(X)\big]
\;\leq\; \inf_{\nu \in \mathcal{T}} \nu' \E[\bm{b}_0(X)]
= \bar{\sigma}(q),
\end{align}
and the covariate-assisted set $\mathcal{B}$ is weakly contained in the basic set $\bar{\mathcal{B}}$.
\end{proposition}

Proposition \ref{cor:tightness} characterizes the boundary of the identified set for $\bm{\beta}_0$ with and without covariates. It represents the boundary (i.e., support function) as (aggregated) intersection bounds for a general class of linear programs. As we demonstrate later on, this representation facilitates estimation and inference.

\subsection{Influence Function} 
\label{sec:mr2}

In this section, we derive the influence function for the support function. Our first assumption requires the dual variable to be unique a.s. in covariate space. This assumption rules out ``flat faces'' and establishes regularity of support function.

\begin{assumption}[Unique Dual Vertex]
\label{ass:boundary}
The dual minimization problem \eqref{eq:dualshadowprice1} has a unique minimizer $\nu_0(x)$ for each $x \in \mathcal{X}$. In other words, for any distinct $\nu_1, \nu_2 \in \mathbf{R}^k$, the probability that both achieve the same dual value as $\nu_0(x)$ is zero:
\begin{align}
\label{eq:zeromass}
\Pr\!\Big( X : \exists \nu_1 \neq \nu_2 \in \mathcal{T}
\ \text{s.t.}\ 
\nu_1' \bm{b}_0(X)  = \nu_2' \bm{b}_0(X) = \nu_0(X)'  \bm{b}_0(X) \Big) = 0.
\end{align}
\end{assumption}

Assumption \ref{ass:boundary} requires that the binding vertex in the set $\mathcal{T}$ is unique almost surely under $P_X$. If this condition holds, we can define the dual function mapping $\mathcal{X}$ into $\mathcal{T}$ as
\begin{align}
\label{eq:dualshadowprice}
\nu_0(x) 
= \sum_{\nu \in \mathcal{T}} \nu \, \bm{1}\!\left\{ \nu \in \arg\min_{\tilde\nu \in \mathcal{T}} \tilde\nu' \bm{b}_0(X) \right\}.
\end{align}
This assumption is plausible if the vector $\bm{b}_0(X)$ is continuously distributed, for example, if it has an a.s. bounded density.

The following Assumption \ref{ass:secmom} is a mild technical condition on the  random variable $\mathbf{B}$. 
For example, if the random variable $\mathbf{B}$ is bounded a.s., this assumption is trivially satisfied.

\begin{assumption}[Bounded Second Moment]
\label{ass:secmom}
The variance of the signal vector $\mathbf{B}$ is bounded in operator norm:
\begin{align}
\label{eq:boundedsecondmoment}
\sup_{x \in \mathcal{X}} \lambda_{\max}\!\left( \E_P[\mathbf{B}\mathbf{B}' \mid X=x] \right) \;\leq\; \bar B.
\end{align}
\end{assumption}

A statistical functional $\theta(P)$ is \emph{pathwise differentiable} at~$P_0$
if, for every regular parametric submodel $\{P_\varepsilon : \varepsilon \in
(-\delta, \delta)\}$ passing through~$P_0$ at $\varepsilon = 0$ with
score~$S(W)$, the map $\varepsilon \mapsto \theta(P_\varepsilon)$ is
differentiable at $\varepsilon = 0$ and the derivative can be represented as
\begin{align}
\label{eq:pathwise_def}
\left.\frac{d}{d\varepsilon}\right|_{\varepsilon=0} \theta(P_\varepsilon)
\;=\; \E_{P_0}\!\big[\phi(W)\, S(W)\big]
\end{align}
for some mean-zero, finite-variance function $\phi(W)$ that does not depend on the
choice of submodel.

Pathwise differentiability underwrites regular $\sqrt{N}$-inference:
a regular, asymptotically linear estimator $\widehat\theta$ with influence
function $\phi$ satisfies $\sqrt{N}(\widehat\theta - \theta) =
N^{-1/2}\sum_{i=1}^N \phi(W_i) + o_P(1)$ and is asymptotically normal with
variance $\Var(\phi)$. When pathwise differentiability fails --- as it does
for the pointwise support function at a fixed~$b$ when flat faces are present
\citep{HiranoPorter2012} --- regular $\sqrt{N}$-inference is generally
impossible without additional smoothing or regularization.

\begin{proposition}
\label{thm:ptw}
Suppose Assumptions~\ref{ass:boundary} and~\ref{ass:secmom} hold. Then $\sigma(q)$ is pathwise differentiable, and
\begin{align}
\label{eq:if}
\phi_q(W) \;=\; \nu_0(X)'\mathbf{B} - \sigma(q)
\end{align}
is an influence function for $\sigma(q)$. Consequently, any regular, asymptotically linear estimator $\widehat\sigma(q)$ with influence function $\phi_q$ --- the existence of which is established for the cross-fitted plug-in of Definition~\ref{def:dual} in Proposition~\ref{cor:lp} --- admits the representation
\begin{align}
\label{eq:alinear}
\sqrt{N}\big(\widehat\sigma(q) - \sigma(q)\big)
\;=\; \frac{1}{\sqrt{N}}\sum_{i=1}^N \phi_q(W_i) + o_P(1),
\end{align}
and is asymptotically normal with variance $\Var(\phi_q(W)) = \Var(\nu_0(X)'\mathbf{B})$, where the second equality uses the fact that $\sigma(q)$ is a deterministic constant for each fixed $q$.
\end{proposition}

Proposition~\ref{thm:ptw} shows that the influence function for the boundary of the aggregated identified set depends only on the dual variable $\nu_0(x)$ and on the chosen signal $\mathbf{B}$. The key finding is that the identity of the binding dual vertex $\nu_0(x)$ can be treated as known. No correction term is needed for estimating the $\arg\min$ in \eqref{eq:dualshadowprice}, paralleling the ``oracle property'' in \citet{LuedtkeLaan}.

\begin{remark}[\citet{LuedtkeLaan} as a special case]
In Example \ref{ex:hsc}, Assumption \ref{ass:boundary} reduces to
\[
\Pr\!\big(s(1, X) = s(0, X)\big) = 0.
\]
If this condition holds, the semiparametric efficiency bound for $\E[\min(s(1, X), s(0, X))]$ is well defined \citep{LuedtkeLaan}; otherwise, regular estimators may not exist \citep{HiranoPorter2012}. Let $\pi(X) = \Pr(D=1 \mid X)$ be the propensity score. Taking $\mathbf{B}$ as the doubly robust signal of \citet{Robins}, the dual signal $\nu_0(X)'\mathbf{B}$ entering the influence function $\phi_q(W) = \nu_0(X)'\mathbf{B} - \sigma(q)$ of Proposition~\ref{thm:ptw} reduces to
\begin{align}
\label{eq:dr}
\nu_0(X)'\mathbf{B}
\;&=\; \bm{1}\{s(1, X) > s(0, X)\}\,\cdot\,\left[\, s(0, X) \;+\; \frac{1-D}{1-\pi(X)}\,\bigl(S - s(0, X)\bigr)\,\right] \nonumber\\
&\quad +\; \bm{1}\{s(1, X) < s(0, X)\}\,\cdot\,\left[\, s(1, X) \;+\; \frac{D}{\pi(X)}\,\bigl(S - s(1, X)\bigr)\,\right].
\end{align}
The right-hand side coincides, up to centering, with the efficient influence function for $\E[\min(s(0, X), s(1, X))]$ established in \citet[Theorem 1]{LuedtkeLaan}.
\end{remark}

\begin{remark}[\citet{Robins} as a special case]
In Example \ref{ex:hsc}, if covariates fail to detect the sign change, i.e.
\[
s(1, x) > s(0, x) \quad \text{for all } x \in \mathcal{X},
\]
then $\nu_0(x) = (0, 1, 0)'$ for all $x$, and \eqref{eq:dr} reduces to the efficient influence function of \citet{Robins} for $s(0) = \E[s(0, X)] = \Pr(S=1 \mid D=0)$.
\end{remark}

\begin{remark}[\cite{levis2023covariateassisted} as a special case]
\cite{levis2023covariateassisted} studies the classical Balke--Pearl bounds \citep{BalkePearl1997}, which can be written as a special case of Example \ref{ex:hp1} with binary outcomes. In their formulation, the vector of latent shares is 
\(\bm{\beta}_0 = (\beta_{10}, \beta_{01}, \beta_{11}, \beta_{00})\), denoting the joint distribution of potential outcomes \((Y(1), Y(0))\). The revealed-preference consistency conditions with the observed marginals can be written compactly as
\[
A \bm{\beta}_0 = \bm{b}_0
\]
with
\[
A =
\begin{bmatrix}
1 & 0 & 1 & 0 \\
0 & 1 & 0 & 1 \\
1 & 1 & 0 & 0 \\
0 & 0 & 1 & 1
\end{bmatrix}, \quad
\bm{b}_0 =
\begin{bmatrix}
P(Y=1 \mid Z=0) \\
1-P(Y=1 \mid Z=0) \\
P(Y=1 \mid Z=1) \\
1-P(Y=1 \mid Z=1)
\end{bmatrix}.
\]
\end{remark}

Riesz Representation is often used to construct moment functions obeying orthogonality conditions, see, e.g., \cite{Newey1994}. 
The novelty of this work is to give an example of Riesz representer that is not available in a closed form but instead is given as a solution to linear program.

\begin{remark}[Riesz Representation of Support Function]
\label{lem:imp}
Let $P_{\theta}$ be a parametric submodel of $P_W$, and let $b_{\theta}(x)$ be the expectation function. Then
\begin{align}
\label{eq:adjustmentterm}
\partial_{\theta} \sigma_{\theta}(q) 
= \E \big[ \nu_0(X)' \partial_{\theta} b_{\theta}(X) \big],
\end{align}
so $\nu_0(x)$ is the Riesz representer of $\bm{b}_0(x)$. Adding the correction term and invoking strong duality gives
\[
\sigma(q, X) + \nu_0(X)' (\mathbf{B} - \bm{b}_0(X)) 
= \nu_0(X)' \mathbf{B}.
\]
\end{remark}
Remark \ref{lem:imp} shows that $\nu_0(x)$ is the Riesz representer of the nuisance vector-function $\bm{b}_0(X)$. The last element of representation \eqref{eq:adjustmentterm} provides the adjustment term needed for the construction of the influence function. Adding it to the original (primal) moment $\sigma(q) := \E[\sigma(q, X)]$ yields the influence function
\begin{align}
\psi_q(W) 
= \sigma(q, X) + \nu_0(X)' (\mathbf{B} - \bm{b}_0(X)) - \sigma(q) 
= \nu_0(X)' \mathbf{B} - \sigma(q).
\label{eq:influencefunction}
\end{align}

\section{Illustration: The Jobs First Welfare-to-Work Experiment}
\label{sec:illustration}

In this section, we develop the Jobs First linear program end-to-end, from
the MDRC randomized trial through the coefficient matrix~$A$ to the
finer earnings grid that the empirical analysis of
Section~\ref{sec:emp} exploits. The treatment follows \citet{KT}'s
data construction and discretization verbatim, and adds the mechanical
refinement of~$A$ that is new here. Numerical results appear in
Section~\ref{sec:emp}; the estimation and inference theory used below
is developed in Section~\ref{sec:overview}.

Jobs First was a welfare-to-work assistance program introduced in
Connecticut in the late 1990s as an alternative to the federal Aid to
Families with Dependent Children (AFDC) program. In 1996, the
Manpower Demonstration Research Corporation (MDRC) conducted a
randomized trial in which eligible female applicants were randomly
assigned either to Jobs First (treatment group; we abbreviate as
``JF'' throughout) or left eligible for AFDC (control group). JF
replaced AFDC's unlimited eligibility rules with a $21$-month time
limit and a more generous earnings disregard: recipients could keep
their welfare check while working until their earnings exceeded the
Federal Poverty Line (FPL). The program thus provided stronger work
incentives and enforced stricter participation limits.

\paragraph{Data description.}
The data structure is a special case of Example~\ref{ex:hp1}. The
exogenous treatment indicator is $D_i = 1$ for women assigned to JF
and $D_i = 0$ for women assigned to AFDC; since compliance with
random assignment is perfect in the MDRC trial, we set $Z_i = D_i$
throughout. For each woman $i$, let $E_i$ denote quarterly earnings
and $W_i$ an indicator of welfare receipt. Following \citet{KT}, we
form a discretized earnings outcome $Y^{\text{earn}}_i$ and a binary
welfare-participation outcome $Y^{\text{wel}}_i$ by
\begin{equation}
\label{eq:yearn}
Y^{\text{earn}}_i \;=\; \begin{cases}
0 & \text{if } E_i = 0, \\
1 & \text{if } 0 < E_i \le \text{FPL}, \\
2 & \text{if } E_i > \text{FPL},
\end{cases}
\qquad\qquad
Y^{\text{wel}}_i \;=\; \begin{cases}
p & \text{if } W_i = 1, \\
n & \text{if } W_i = 0.
\end{cases}
\end{equation}
The combined outcome $Y_i = (Y^{\text{earn}}_i, Y^{\text{wel}}_i)$
takes six values in $\{0p,\,1p,\,2p,\,0n,\,1n,\,2n\}$. Three
reporting conventions, inherited from \citet{KT}, govern how these
values map to latent response types:
\begin{enumerate}
\item $2p \equiv 2u$: anyone reporting above-FPL earnings while on
welfare is, by definition, under-reporting.
\item Under AFDC, $1p$ pools two latent types, which are observationally
indistinguishable, $1r$ (truthful) and
$1u$ (under-reporting below FPL).
\item Under JF, the earnings disregard removes the incentive to
under-report below FPL, so $1p = 1r$.
\end{enumerate}
The baseline covariate vector $X_i$ contains $28$ variables: age,
education level, number of children, family and marital status, and a
quarterly history of employment, earnings, AFDC participation, and
food-stamp receipt over the eight quarters prior to random
assignment. The continuously distributed earnings and AFDC-receipt
histories are central to the validity of the margin condition
(Assumption~\ref{ass:mamain}) in this application: the conditional
shares $\bm b_0(x)$ inherit non-degeneracy from these continuous
components, so near-ties between competing dual vertices occur with
probability zero (cf.\ Remark~\ref{rm:margin_verif}). Following
\citet{KT}, we focus on the sample of $N = 4{,}641$ women whose
child-count variable is not missing.

\paragraph{Linear-programming problem.}
We represent this problem as a special case of the linear
system~\eqref{eq:linearcons}. Absent restrictions there are
$7\times 6 = 42$ latent response margins, corresponding to every
pairing of an AFDC state with a JF state. Revealed-preference
arguments spelled out in detail in
\citet[Table~1 and Online Appendix~B]{KT} imply that only $10$ of
these margins are feasible, and one of them, $\beta_{1u,1r}$, is
degenerate at one. The number of free response probabilities
therefore reduces to $9$, which we collect into the vector
\[
\bm{\beta}_0 \equiv \big(
\beta_{0n,1r},\,
\beta_{0r,0n},\,
\beta_{2n,1r},\,
\beta_{0r,2n},\,
\beta_{0r,1r},\,
\beta_{0r,1n},\,
\beta_{1n,1r},\,
\beta_{0r,2u},\,
\beta_{2u,1r}
\big)'.
\]
These latent shares represent the fractions of women whose
counterfactual AFDC and JF states are linked by revealed preference:
$\beta_{s^a, s^j} = P(\text{AFDC state} = s^a,\ \text{JF state} = s^j)$
is the joint probability that a randomly selected woman would occupy
AFDC state $s^a$ and JF state $s^j$. The empirical moments
$\bm{b}_0$ are differences in observed shares under JF and AFDC,
\[
\bm{b}_0 \equiv \big(
p^{\,\text{JF}}_{0n} - p^{\,\text{AFDC}}_{0n},\,
p^{\,\text{JF}}_{1n} - p^{\,\text{AFDC}}_{1n},\,
p^{\,\text{JF}}_{2n} - p^{\,\text{AFDC}}_{2n},\,
p^{\,\text{JF}}_{0p} - p^{\,\text{AFDC}}_{0p},\,
p^{\,\text{JF}}_{2p} - p^{\,\text{AFDC}}_{2p}
\big)',
\]
where $p^{\,\text{AFDC}}_{s}$ and $p^{\,\text{JF}}_{s}$ denote
observed shares in state $s$ under AFDC and JF, respectively. Rather
than working with joint probabilities $\beta_{s^a, s^j}$ directly, we
reparametrize them as conditional probabilities
\[
\pi_{s^a,s^j} \;=\; P(\text{JF state} = s^j \mid \text{AFDC state} = s^a)
\;=\; \beta_{s^a,s^j} / p^{\,\text{AFDC}}_{s^a},
\]
where $p^{\,\text{AFDC}}_{s^a}$ is the IPW-adjusted probability of
state $s^a$ in the AFDC control group. This reparametrization
ensures that the coefficient matrix $A$ is non-stochastic, and
restrictions can be written as
\[
A \bm{\beta}_0 = \bm{b}_0,
\]
with the $5 \times 9$ coefficient matrix
\begin{equation}
\label{eq:Acoarse}
A \;=\; \begin{pmatrix}
-1 & 1 & 0 & 0 & 0 & 0 & 0 & 0 & 0 \\
0 & 0 & 0 & 0 & 0 & 1 & -1 & 0 & 0 \\
0 & 0 & -1 & 1 & 0 & 0 & 0 & 1 & 0 \\
0 & -1 & 0 & -1 & -1 & -1 & 0 & -1 & 0 \\
0 & 0 & 0 & 0 & 0 & 0 & 0 & 1 & -1
\end{pmatrix}.
\end{equation}
Each of the five rows enforces a conservation-of-mass condition: the
observed difference in the share of women in a given observable state
across the two regimes must be accounted for by the latent ``flows''
permitted by revealed preference. Unlike \citet{KT}, who solve this
system by deriving closed-form expressions of the conditional
probabilities (their Online Appendix~B), we work directly with the
linear-programming representation.

The function to be maximized reduces to interpretable expressions of
structural parameters. For example, to bound the transition
probability $\pi_{2n, 1r} = P(\text{JF state} = 1r \mid \text{AFDC state}
= 2n)$, the fraction of women who reduce their labor supply and
opt into welfare in response to Jobs First, we choose
$q = e_3 = (0, 0, 1, 0, 0, 0, 0, 0, 0)'$ in~\ref{eq:lpbasicq}, the third standard basis vector,
and obtain an upper bound on the joint probability $\beta_{2n, 1r}$;
the lower bound corresponds to $q = -e_3$. The combinatorial bound
$|\mathcal{T}(q)| \le \binom{d}{k}$ of Remark~\ref{rm:largeLP} is
$\binom{9}{5} = 126$, but the dual vertex set is small in practice:
for $q = \pm e_3$, $\mathcal{T}(q)$ contains only a handful of
vertices, well within the regime where Section~\ref{sec:overview}'s
fixed-$|\mathcal{T}|$ asymptotic theory applies.

\paragraph{Discretizing outcomes.}
\citet{KT} discretize earnings into three bins - zero earnings,
positive earnings at or below FPL, and earnings above FPL as shown in~\ref{eq:yearn}. They derive bounds for conditional probabilities analytically using closed-form
expressions (see their Online
Appendix~B).

In discussing their empirical results,
\citet[Section~X.A, p.~1009]{KT} themselves pose the following
question:
\begin{quote}\small
``[The] finding of a significant opt-in response could hypothetically
reflect trivial earnings reductions from \$1 above the poverty line
to exactly the poverty line.''
\end{quote}
Answering this question requires looking inside the coarse above-FPL
bin: if the measured opt-in response is concentrated among women who
reduced earnings by only a few dollars from just above FPL to the
FPL threshold, the response is consistent with a trivial rounding;
if instead the response is spread across larger reductions, it is
consistent with substantive labor-supply adjustments. This question
is beyond the scope of existing discretization and requires finer partition of the outcome bins.

Dividing the above-FPL category into two sub-bins proved to be analytically more involved than the baseline with three income bins as \citet{KT} show in Online Appendix, Section 6. The linear programming framework makes this and further refinements easy to implement. Each new sub-bin introduces additional transition parameters (columns) and potentially additional conservation-of-mass equations (rows), so that the coefficient matrix $A$ is replaced by a block-expanded matrix $\widetilde{A}$. One might argue that under this finer granularity regime the construction of the non-stochastic matrix $\widetilde A$ is challenging. However, the refined partition is induced from the main (coarse) one and directly inherits the revealed preference arguments. These arguments underlying each row depend only on the welfare participation and reporting status of the woman, and do not depend on the specific level of earnings within a bin. We will elaborate with a setup and give an example construction below.

\paragraph{Block expansion of $A$.}

Let $\mathcal{S} = \{s_1, \ldots, s_k\}$ denote the original partition of the outcome support into coarse bins resulting in matrix $A$. A \emph{refinement} $\widetilde{\mathcal{S}}$ is any partition that subdivides each $s \in \mathcal{S}$ into one or more disjoint sub-bins; we write $\tilde s \prec s$ to mean that $\tilde s \in \widetilde{\mathcal{S}}$ is contained in $s$. The refinement $\widetilde{\mathcal{S}}$ induces a refined linear system 
\begin{align}
\label{eq:linearcons-refined}
\widetilde A\, \tilde{\bm\beta}_0 = \tilde{\bm b}_0, \qquad \tilde{\bm\beta}_0 \geq 0,
\end{align}
where $\tilde{\bm b}_0$ collects the expectations of indicators for each refined observable state, $\widetilde A$ summarizes the flows into and out of each state at a finer partition level $\widetilde{\mathcal{S}}$ and $\tilde{\bm\beta}_0$ is the vector collecting corresponding finer transition probability parameters.

Revealed preference restricts transitions at the level of latent states, which do not depend on how the outcome is partitioned. Consequently, if the coarse model rules out the transition $s^a \to s^j$, then the refined model rules out every transition $\tilde s^a \to \tilde s^j$ with $\tilde s^a \prec s^a$ and $\tilde s^j \prec s^j$. The matrix $\widetilde A$ therefore inherits a block structure from $A$: each entry $A_{ij}$ expands into a block whose rows now should reflect the sub-bins of $s^a_i$ (if any) and columns should account for new transition parameters. In Section~\ref{sec:emp}, we report main results on a finer partition case with 9 income bins. For a coarse state $s \in \mathcal{S}$, let $m(s) = |\{\tilde s \in \widetilde{\mathcal{S}} : \tilde s \prec s\}|$ denote the number of sub-bins that refine it. Our main granular specification sets $m(s) = 1$ for zero earnings (which admits no further subdivision), $m(s) = 5$ for earnings below FPL, and $m(s) = 3$ for earnings above $\mathrm{FPL}$.

Index the rows of $A$ by coarse observable states $s_r$ and its columns by the pairs $(s^a, s^j)$ that label coarse $\beta$ parameters. The refined matrix $\widetilde A$ then decomposes into blocks $\widetilde A_{r, (s^a, s^j)}$ of size $m(s_r) \times m(s^a)\, m(s^j)$, one for each entry of $A$. Two cases determine the content of each block:
\begin{enumerate}
    \item If $A_{r, (s^a, s^j)} = 0$, the block is zero: no fine-level transition contributes to the flow equation for any sub-bin of $s_r$.
    \item If $A_{r, (s^a, s^j)} \neq 0$, the block is sparse, with entries equal to $A_{r, (s^a, s^j)}$ in positions dictated by fine-level flow conservation. When $m(s_r) = 1$, every fine transition $(\tilde s^a, \tilde s^j) \prec (s^a, s^j)$ enters the single equation, so the block is the row vector $A_{r, (s^a, s^j)} \cdot \mathbf{1}'_{m(s^a)\, m(s^j)}$. When $m(s_r) > 1$, each sub-bin $\tilde s_r \prec s_r$ has its own flow equation, and the coarse coefficient appears only in positions $(\tilde s_r, (\tilde s^a, \tilde s^j))$ for which $\tilde s_r \in \{\tilde s^a, \tilde s^j\}$, i.e. on sub-diagonals selected by the fine-level conservation accounting.
\end{enumerate}

To illustrate the block expansion of $A$ concretely, consider the
transition parameter~$\beta_{2n,\,1r}$, where state~$2n$
(above-FPL, not on welfare) is the source and state~$1r$
(below-FPL, on welfare) is the destination. In the main
granular specification of Section~\ref{sec:emp}, $2n$ is refined into three sub-bins
$\widetilde{\mathcal{S}}_{2n} = \{b_{6}n,\, b_{7}n,\, b_{8}n\}$ and
$1r$ is refined into five sub-bins
$\widetilde{\mathcal{S}}_{1r} = \{b_{1}r,\, b_{2}r,\, b_{3}r,\, b_{4}r,\, b_{5}r\}$.
Here and below, subscripts $1$--$5$ index the five below-FPL
sub-bins and subscripts $6$--$8$ index the three above-FPL
sub-bins.  We choose this transition to illustrate the general case
in which both source and destination states are divided into sub-bins.

Under the coarse partition, the conservation-of-mass constraint for
state~$2n$ reads
\begin{equation}
  \label{eq:coarse_row1}
  -\beta_{2n,\,1r} + \beta_{0r,\,2n} \;=\; p^{j}_{2n} - p^{a}_{2n}.
\end{equation}
The revealed-preference restriction encoded in this row permits
exactly one outflow from~$2n$, to $1r$: a woman earning above the poverty line
while not on welfare under AFDC may reduce her earnings below the
line and take up assistance under~JF.
The only permitted inflow to~$2n$ is from state~$0r$: a woman
on welfare who exits and earns above the poverty line.  Following
\citet{KT}, we drop the conservation constraint for the reference
state~$1r$; the $1r$ partition therefore affects only the column
structure of~$\widetilde{A}$, not its row structure\footnote{Women earning in range 1 on welfare face an incentive to underreport under AFDC, and hence the observable state $1r$ pools truthful reporters and underreporters ($1r$ and $1u$). A conservation-of-mass row for this group would require separating these latent types, which are empirically indistinguishable}.

Refining the partition introduces $3 \times 5 = 15$ new transition
parameters $\{\beta_{b_{k}n,\,b_{i}r}\}$ for ${k\,\in\,\{6, 7, 8\},\;\, 
i\,\in\,\{1, \dots, 5\}}$, one for each pair of a source sub-bin
of~$2n$ and a destination sub-bin of~$1r$.  Because the revealed
preference argument underlying equation~\eqref{eq:coarse_row1} ---
that a woman earning above the poverty line may reduce earnings to
qualify for assistance --- operates at the level of the \emph{coarse}
states $2n$ and $1r$ and is independent of how finely earnings are
classified within each bin, every sub-bin transition $b_{k}n \to
b_{i}r$ is admissible for all $k \in \{6, 7, 8\}$ and
$i \in \{1, \dots, 5\}$.  Consequently, the refined conservation
constraint for each source sub-bin $b_{k}n$ with $k \in \{6,7,8\}$
takes the form
\begin{equation}
  \label{eq:fine_row1}
  -\sum_{i=1}^{5}\beta_{b_{k}n,\,b_{i}r}
  + \beta_{0r,\,b_{k}n}
  \;=\;
  p^{j}_{b_{k}n} - p^{a}_{b_{k}n}.
\end{equation}
Equation~\eqref{eq:fine_row1} is the direct granular analogue of
equation~\eqref{eq:coarse_row1}: the single outflow term
$-\beta_{2n,\,1r}$ expands into the sum $-\sum_{i=1}^{5}
\beta_{b_{k}n,\,b_{i}r}$ over all five destination sub-bins of~$1r$,
while the inflow term and the right-hand side update to reflect the
finer observable states.

\vspace{1.5em}

\begin{figure}[htbp]

\centering
\caption{\textsc{Expansion of the boxed coarse entries into granular blocks}}
\label{fig:block}
\renewcommand{\arraystretch}{1.25}

\resizebox{!}{0.16\textheight}{$
\begin{array}{c}

\begin{bNiceMatrix}[first-row, first-col, margin]
& & & \beta_{2n,1r} &\beta_{0r,2n} & & & & & \\

\Block{5-1}{2n}
& -1 & +1 & 0 & 0 & 0 & 0 & 0 & 0 & 0 \\

& 0 & 0 & 0 & 0 & 0 & +1 & -1 & 0 & 0 \\

& 0 & 0
& \Block[draw=blue, line-width=1pt]{1-1}{-1}
& \Block[draw=red, line-width=1pt]{1-1}{+1}
& 0 & 0 & 0 & 0 & 0 \\

& 0 & -1 & 0 & -1 & -1 & -1 & 0 & -1 & 0 \\

& 0 & 0 & 0 & 0 & 0 & 0 & 0 & +1 & -1
\end{bNiceMatrix}

\\[1.5em]
\Downarrow
\\[1.5em]

\begin{bNiceMatrix}[first-row, first-col, margin]
&
\cdots
& \vbeta{\beta_{b_6n,b_1r}}
& \vbeta{\beta_{b_6n,b_2r}}
& \vbeta{\beta_{b_6n,b_3r}}
& \vbeta{\beta_{b_6n,b_4r}}
& \vbeta{\beta_{b_6n,b_5r}}
& \vbeta{\beta_{b_7n,b_1r}}
& \vbeta{\beta_{b_7n,b_2r}}
& \vbeta{\beta_{b_7n,b_3r}}
& \vbeta{\beta_{b_7n,b_4r}}
& \vbeta{\beta_{b_7n,b_5r}}
& \vbeta{\beta_{b_8n,b_1r}}
& \vbeta{\beta_{b_8n,b_2r}}
& \vbeta{\beta_{b_8n,b_3r}}
& \vbeta{\beta_{b_8n,b_4r}}
& \vbeta{\beta_{b_8n,b_5r}}
& \cdots
& \vbeta{\beta_{0r,b_6n}}
& \vbeta{\beta_{0r,b_7n}}
& \vbeta{\beta_{0r,b_8n}}
& \cdots
\\

\vdots
& \ddots
& & & & & & & & & & & & & & & & \cdots & & & & \ddots
\\

b_6n
& \cdots
& \Block[draw=blue, line-width=1pt]{3-15}{} -1
& -1 & -1 & -1 & -1
&    &    &    &    &
&    &    &    &    &
& \cdots
& \Block[draw=red, line-width=1pt]{3-3}{} +1
&    &
& \cdots
\\

b_7n
& \cdots
&    &    &    &    &
& -1 & -1 & -1 & -1 & -1
&    &    &    &    &
& \cdots
&    & +1 &
& \cdots
\\

b_8n
& \cdots
&    &    &    &    &
&    &    &    &    &
& -1 & -1 & -1 & -1 & -1
& \cdots
&    &    & +1
& \cdots
\\

\vdots
& \ddots
& & & & & & & & & & & & & & & & \cdots & & & & \ddots
\end{bNiceMatrix}

\end{array}
$}

\vspace{1.75em}

\noindent
\parbox{\linewidth}{%
\footnotesize
\emph{Notes:} Refining the aggregate state $2n$ into
    $\{b_{6}n,\, b_{7}n,\, b_{8}n\}$ and $1r$ into
    $\{b_{1}r, \ldots, b_{5}r\}$ splits the coarse ``$2n$'' row
    into three rows and replaces each coefficient by structured
    blocks, $3 \times 15$ and $3 \times 3$ for $\beta_{2n,1r} = -1$ and $\beta_{0r,2n} = +1$, respectively. 
    Signs are preserved throughout.
}

\end{figure}

Figure~\ref{fig:block} illustrates how to apply this block-expanding procedure going from the original $5 \times 9$ matrix. The researcher is free to adopt whichever granular specification they want: however, the granularity of the columns cannot be coarser than those of the rows, since the conservation-of-mass equation of the row would not be possible to pin down using coarser flow parameters over the columns. As long as this criterion is satisfied, the number of columns (i.e. the granularity of the included transition parameters) can be freely modified. Granular refinement of the destination state need not be applied uniformly across parameters. For $\pi_{2n,1r}$ one may resolve the destination at the fine grid, while for $\pi_{1n,1r}$ the coarse grouping suffices; the choice is dictated by which transition probability is the object of interest. Figure~\ref{fig:design-matrices} shows sequentially more granular design specifications: Panel (b) shows an example of $13 \times 25$ matrix, where only 2n, 1n and 2p latent types are granularized, with a choice of $\beta_{2n, 1r}$, $\beta_{0r, 2n}$, $\beta_{0r, 1n}$, $\beta_{1n, 1r}$, $\beta_{0r, 2u}$ and $\beta_{2u, 1r}$ being split into $\beta_{b_kn, 1r}$, $\beta_{0r, b_kn}$, $\beta_{0r, b_in}$, $\beta_{b_in, 1r}$, $\beta_{0r, b_ku}$ and $\beta_{b_ku, 1r}$ with $k=6,7,8$ and $i=1,2,3,4, 5$. Panel (c) shows further splitting into $13 \times 33$ matrix, and Panel (d) corresponds to \eqref{eq:fine_row1} when both source and destination bins are split (group $G_3$).

\paragraph{Welfare Bounds.}
The conditional linear program framework bounds any linear functional of $\bm\beta_0(x)$, not only individual transition probabilities. Bounding the aggregate welfare gain induced by the Jobs~First reform is a leading example, corresponding to a particular choice of the objective vector~$q$ in~\ref{eq:lpbasicq},
with the constraint set $A\bm\beta_0(x) = \bm b_0(x)$,
$\bm\beta_0(x) \geq 0$ unchanged.

For a woman with covariates~$x$, her monthly disposable
income in state~$s$ under regime $t \in \{a, j\}$ is $
  I^{t}(s,\, x) \;\coloneqq\; y_s + G^{t}(y_s,\, x)$,
where $y_s$ denotes the representative earnings in
bin~$s$ and $G^{t}(\cdot,\, x)$ is the transfer schedule as in
\citet{KT} under regime~$t$. Let $\mathcal{M}$ denote the set of all transitions admissible by revealed preference restrictions. Abstracting from other forms of income and assistance, such as food stamps and the EITC, for each admissible transition $(s^a, s^j) \in \mathcal{M}$, let the welfare gain be the induced change in monthly disposable income $
  \Delta_{(s^a,\, s^j)}(x)
  \;\coloneqq\;
  I^{j}(s^j,\, x) - I^{a}(s^a,\, x).
$
The aggregate expected welfare gain attributable
to the transitions in~$\mathcal{M^*} \subseteq \mathcal{M}$ is
\begin{equation}
  \label{eq:welfare-aggregate}
  \omega_{\mathcal{M^*}}
  \;\coloneqq\;
  \mathbb{E}\!\left[
    \sum_{(s^a,\,s^j)\,\in\,\mathcal{M^*}}
    \Delta_{(s^a,\,s^j)}(X)\;
    \beta_{0,\,(s^a,\,s^j)}(X)
  \right].
\end{equation}
Because $q$ is fixed in our framework, we evaluate the welfare gain at a representative household, abstracting from heterogeneity in covariates in~\ref{eq:welfare-aggregate}. This approximation discards within-bin variation in disposable income, and its cost diminishes under designs with finer partitions: as bins become more homogeneous, the representative household evaluation more closely approximates the covariate-specific gains it replaces. The welfare weights are then fixed, and the functional is recovered by an objective vector~$q$ that does not depend on~$x$, indexed conformably with~$\bm\beta_0$ and with coordinate
\[
  q_{(s^a,\,s^j)}
  \;=\;
  \Delta_{(s^a,\,s^j)}\;
  \mathbf{1}\!\left\{(s^a, s^j) \in \mathcal{M^*}\right\}.
\]
corresponding to welfare gain value of transition $(s^a, s^j)$.

\paragraph{From unconditional to conditional probabilities.}
The LP~\eqref{eq:Acoarse} is solved over the vector of joint
probabilities $\bm{\beta}_0$, but the parameters of economic interest
are the conditional transition probabilities $\pi_{s^a, s^j}$
introduced above, which is what we report in
Section~\ref{sec:emp}. The two are linked by
$\pi_{s^a, s^j} = \beta_{s^a, s^j} / p^{\,\text{AFDC}}_{s^a}$, where
$p^{\,\text{AFDC}}_{s^a} = P(\text{AFDC state} = s^a)$ is a
separately and consistently estimable scalar that is bounded away
from zero in the data. A bound on $q' \bm{\beta}_0$ is therefore
converted to a bound on the corresponding conditional probability by
dividing through by $\widehat p^{\,\text{AFDC}}_{s^a}$.

This division has two consequences for the asymptotic theory of
Section~\ref{sec:overview}. First, the influence function for the
conditional probability $\pi_{s^a, s^j}$ is obtained from the
influence function $\varphi_q(W)$ of $q' \bm{\beta}_0$ by the delta method:
\[
\varphi_{\pi_{s^a, s^j}}(W) \;=\;
\frac{1}{p^{\,\text{AFDC}}_{s^a}}
\!\left[\, \varphi_q(W) \;-\; \pi_{s^a, s^j}
\big(\mathbf{1}\{\text{AFDC state} = s^a\} - p^{\,\text{AFDC}}_{s^a}\big)
\right].
\]
The first term is the contribution of estimating $q' \bm{\beta}_0$;
the second term accounts for the estimated denominator
$\widehat p^{\,\text{AFDC}}_{s^a}$. Both contributions are
$\sqrt{N}$-consistent and asymptotically Gaussian under the
assumptions of Section~\ref{sec:overview}, since
$p^{\,\text{AFDC}}_{s^a}$ is a sample mean of an indicator and is
bounded away from zero in the data. Second, the cross-fitted plug-in estimator
$\widehat{\sigma}(q) = N^{-1}\sum_{i=1}^N \widehat{\nu}(X_i)' \mathbf{B}_i$
of Definition~\ref{def:dual} is converted to an estimator of the
bound on $\pi_{s^a, s^j}$ by dividing by
$\widehat p^{\,\text{AFDC}}_{s^a}$. The reported bounds and confidence intervals in
Section~\ref{sec:emp} reflect this rescaling.

\section{Estimation and Inference}
\label{sec:overview}

This section develops the estimation and inference theory for the
covariate-assisted boundary $\sigma(q)$. Section~\ref{sec:estimator}
introduces the cross-fitted dual plug-in estimator and a multiplier-bootstrap procedure for constructing confidence intervals.
Section~\ref{sec:largesample} states the assumptions and the main
asymptotic result (Proposition~\ref{cor:lp}), establishing
$\sqrt{N}$-consistency, asymptotic normality, and uniform bootstrap
coverage. Section~\ref{sec:discussion} discusses the assumptions, the
role of the margin condition, and robust alternatives that remain
valid when the margin condition fails.

\subsection{The Estimator}
\label{sec:estimator}

In this Section, we introduce the dual estimator of the boundary assuming the conditions of Proposition \ref{thm:ptw} hold. Let $\mathbf{B}$ be an observable function of data obeying $\bm{b}_0 = \E[\mathbf{B}]$, and let $(W_i)_{i=1}^N$ be an i.i.d.\ sample of size $N$. The first step is to construct the fitted values for the expectation function $\bm{b}_0(x)$. The second step is to construct an estimate of the boundary $\sigma(q)$.  

\begin{definition}[Primal and Dual Cross-Fitted Values]
\mbox{}
\label{def:crossfit} 	
\begin{enumerate} 
\item For a random sample of size $N$, denote a $K$-fold random partition of the sample indices $[N]=\{1, 2, \dots, N\}$ by $(J_k)_{k=1}^K$, where $K$ is the number of partitions and the sample size of each fold is $n=N/K$. For each $k \in [K] = \{1, 2, \dots, K\}$ define $J_k^c = \{1, 2, \dots, N\} \setminus J_k$.

\item For each $k \in [K]$, construct an estimator $\widehat{\bm{b}}_k = \widehat{\bm{b}}(W_{i \in J_k^c})$ of the nuisance parameter $\bm{b}_0$ using only the data $\{W_j: j \in J_k^c\}$. For any observation $i \in J_k$, define the primal fitted value as $\widehat{\bm{b}}_i = \widehat{\bm{b}}_k(X_i)$ and the dual plug-in fitted value
\begin{align}
\label{eq:dualfitted}
\widehat{\nu}(X_i) = \sum_{\nu \in \mathcal{T}} \nu \, \bm{1}\!\left\{ \nu \in \arg \min_{\tilde\nu \in \mathcal{T}} \tilde\nu' \widehat{\bm{b}}_i \right\}.
\end{align}
\end{enumerate}
\end{definition}

\begin{definition}[Dual Plug-in Estimator]
\label{def:dual}
Let $(\widehat{\nu}'(X_i))_{i=1}^N$ be the dual cross-fitted values. Define
\begin{align}
\label{eq:psihatlp}
\widehat{\sigma}(q) := N^{-1} \sum_{i=1}^N \widehat{\nu}'(X_i)\, \mathbf{B}_i.
\end{align}
\end{definition}

The plug-in $\widehat\nu(X_i)$ in~\eqref{eq:dualfitted} is one of several
reasonable estimators of $\nu_0(x)$; alternatives based on direct
classification of the binding vertex or on smoothing of the inner minimization are natural extensions.

\begin{definition}[Multiplier Bootstrap]
\label{def:bb}
Let $(e_i)_{i=1}^N$ be i.i.d.\ exponential random variables $e_i \sim \text{Exp}(1)$ independent of the data. Define the bootstrap analog of $\widehat{\sigma}(q)$ as
\begin{align}
\label{eq:psihatlpboot}
\widetilde{\sigma}(q) := N^{-1} \sum_{i=1}^N \widehat{\nu}'(X_i)\, \mathbf{B}_i e_i.
\end{align}
\end{definition}
The multiplier bootstrap  is computationally efficient: the $K$-fold cross-fitted nuisance $\widehat\nu(X_i)$ and the signal $\mathbf{B}_i$ are held fixed across bootstrap replications, and each replication reduces to a reweighted average.

Under sparsity or smoothness conditions on $\bm{b}_0(x)$, the dual estimator of the support function enjoys the following properties pointwise in $q \in \mathcal{S}^{d-1}$:
\begin{enumerate}

\item For each fixed $q \in \mathcal{S}^{d-1}$, the estimator is consistent at the parametric rate:
\begin{align}
\label{eq:urate}
\big| \widehat{\sigma}(q) - \sigma(q) \big|
= O_P(N^{-1/2}) = o_P(1).
\end{align}

\item For each fixed $q \in \mathcal{S}^{d-1}$, the estimator $\widehat{\sigma}(q)$ is asymptotically Gaussian:
\begin{align}
\label{eq:limit}
S_N(q) := \sqrt{N}\,\big(\widehat{\sigma}(q) - \sigma(q)\big)
= \G_N(q) + o_P(1).
\end{align}

\item The estimator $\widehat{\sigma}(q)$ can be used to construct pointwise confidence intervals. Taking a $(1-\tau)$-pointwise confidence region (CR) as
\begin{align*}
[\underline{i}(q), \bar{i}(q)]
:= \big[-\widehat{\sigma}(-q) - N^{-1/2} \widehat{c}_{1-\tau/2}(-q), \;\;
\widehat{\sigma}(q) + N^{-1/2} \widehat{c}_{1-\tau/2}(q)\big],
\end{align*}
where the critical value $\widehat{c}_{1-\tau/2}(q)$ is the $(1-\tau/2)$ quantile of the absolute bootstrap statistic $|\widetilde{S}_N(q)|$, with
\begin{align*}
\widetilde{S}_N(q) := \sqrt{N}\,\big(\widetilde{\sigma}(q) - \widehat{\sigma}(q)\big).
\end{align*}
The lower endpoint subtracts the critical value from the estimated lower-bound boundary $-\widehat{\sigma}(-q)$, while the upper endpoint adds the critical value to the estimated upper-bound boundary $\widehat{\sigma}(q)$, giving a symmetric Wald-type interval for $q'\bm{\beta}_0 \in [-\sigma(-q), \sigma(q)]$.
That is, for each fixed $q \in \mathcal{S}^{d-1}$ and each $P \in \mathcal{P}$,
\begin{align}
\label{eq:pcoverage}
\liminf_{N \to \infty}
\Pr_{P}\!\big(q' \bm{\beta}_0 \in [\underline{i}(q), \bar{i}(q)]\big)
\;\geq\; 1 - \tau.
\end{align}
\end{enumerate}

\subsection{Asymptotic Theory}
\label{sec:largesample}

In this Section, we describe the assumptions and state the asymptotic results for the proposed estimation and inferential procedures.  

\begin{assumption}[First-Stage Rate]
\label{ass:firststage}
There exists a sequence $\phi_N = o(1)$ and a sequence of sets $\{B_N, N \geq 1\}$ such that the first-stage estimates $\widehat{\bm{b}}(x)$ of the true function $\bm{b}_0(x)$ belong to $B_N$ with probability at least $1 - \phi_N$. The sets $B_N$ shrink at the following rate:
\begin{align}
\sup_{b \in B_N}\;\sup_{x \in \mathcal{X}} \| \bm{b}(x) - \bm{b}_0(x) \| = o(N^{-1/4}).
\end{align}
\end{assumption}

\noindent
Assumption \ref{ass:firststage} requires that the estimates $\widehat{\bm{b}}(x)$ of the vector function $\bm{b}_0(x)$ converge in $\ell_{\infty}$-norm at a sufficiently fast rate. A mean-square version of this assumption frequently occurs in the semiparametric literature (see, e.g., \citet{Newey1994}, \citet{chernozhukov2016double}). Examples of estimators achieving $\ell_{\infty}$ rates include $\ell_1$-regularized estimators in \citet*{Program} under sparsity conditions on the linear or logistic approximations of the coordinates of $\bm{b}_0(x)$.

\begin{remark}[Verification of Assumption \ref{ass:firststage} for Example \ref{ex:hp1}]
\label{rm:firststage_verif}
Each coordinate of $\bm{b}_0(x) = (b_0^{(1)}(x), \ldots, b_0^{(k)}(x))'$ in Example~\ref{ex:hp1} is a conditional probability,
\[
b_0^{(j)}(x) \;=\; \Pr\!\big(Y \in R_j \,\big|\, Z = z_j,\, X = x\big), \qquad j = 1, \ldots, k,
\]
and is therefore bounded in $[0,1]$. Suppose each coordinate admits a sparse logistic-regression approximation on a $p$-dimensional dictionary $B(x) = (B_1(x),\ldots,B_p(x))'$ (possibly with $p \gg N$):
\[
b_0^{(j)}(x) \;=\; \Lambda\!\big(\gamma_j'\,B(x) + r_j(x)\big), \qquad \Lambda(t) = \frac{1}{1 + e^{-t}},
\]
where $\gamma_j \in \mathbb{R}^p$ is at most $s_\gamma$-sparse and the approximation error obeys, uniformly in $j$,
\[
\sup_{j \le k}\;\Big(\frac{1}{N}\sum_{i=1}^N r_j^2(X_i)\Big)^{1/2} \;\lesssim_P\; \sqrt{\frac{s_\gamma^2\,\log p}{N}}.
\]
Estimate each coordinate by $\ell_1$-penalized logistic regression with a single, uniform-in-$j$ penalty level chosen as in \citet*{Program}; equivalently, run $k$ parallel logistic-Lasso regressions sharing the same regularizer, scaled so that the union bound over the $k$ score processes is absorbed into a $\log(pk)$ factor. As shown in \citet*{Program},  uniformly in $j$,
\[
\sup_{x \in \mathcal{X}}\;\big|\widehat{b}_0^{(j)}(x) - b_0^{(j)}(x)\big| \;=\; O_P\!\left(\sqrt{\frac{s_\gamma^2\,\log(pk)}{N}}\right).
\]
For a vector with fixed dimension $k$, the Euclidean norm and the worst-coordinate sup-norm are equivalent up to a $\sqrt{k}$ factor:
\[
\sup_{x \in \mathcal{X}} \|\widehat{\bm{b}}(x) - \bm{b}_0(x)\|_2 \;\le\; \sqrt{k}\,\sup_{x \in \mathcal{X}}\;\max_{j \le k}\,\big|\widehat{b}_0^{(j)}(x) - b_0^{(j)}(x)\big|,
\]
so the same rate carries over. Define $B_N$ as the set of vector-valued logistic models whose coefficients lie within an $\ell_1$-ball of radius $r_N = c\sqrt{s_\gamma^2 \log(pk)/N}$ around $(\gamma_1,\ldots,\gamma_k)$; with probability at least $1 - \phi_N$, $\widehat{\bm{b}} \in B_N$. Hence, the sparsity-and-design condition
\begin{align}
\label{eq:fs_sufficient}
s_\gamma^2\,\log(pk) \;=\; o \big(N^{1/2}\big)
\end{align}
is sufficient for $\sup_{x}\|\widehat{\bm{b}}(x) - \bm{b}_0(x)\|_2 = o(N^{-1/4})$, which verifies Assumption~\ref{ass:firststage}. The condition~\eqref{eq:fs_sufficient} is the natural multivariate generalization of the standard high-dimensional logistic-Lasso requirement; for fixed $k$ it reduces to the familiar $s_\gamma^2 \log p = o(N^{1/2})$.
\end{remark}

\begin{remark}[Bounded image, unbounded $\mathcal{X}$]
\label{rm:unbounded_X}
No support condition on $X$ is required for the verification above. Each coordinate satisfies $b_0^{(j)}(x) \in [0,1]$ regardless of the support of $X$, and logistic-Lasso fitted values lie in $(0,1)$, so the sup-norm rate is a property of the function class rather than of the geometry of $\mathcal{X}$. For nonparametric first stages on unbounded $\mathcal{X}$, the same conclusion follows after a $\sqrt{\log N}$-truncation onto a growing compact set $\mathcal{X}_N = \{x : \|x\| \le C\sqrt{\log N}\}$: under sub-Gaussian tails on $X$, $\Pr(X \notin \mathcal{X}_N) = o(N^{-1/2})$, so observations outside $\mathcal{X}_N$ contribute negligibly to the first-order asymptotics, and it suffices to verify the $\ell_\infty$ rate uniformly over $\mathcal{X}_N$.
\end{remark}

\begin{assumption}[Smooth Distribution of Covariates]
\label{ass:mamain}
The covariate distribution of $\bm{b}_0(X)$ satisfies the smoothness condition
\begin{align}
\label{eq:bound}
\Pr \left( 0 < \inf_{\nu_1 \neq \nu_2, \quad \nu_1, \nu_2 \in \mathcal{T}} 
\dfrac{| (\nu_1 - \nu_2)' \bm{b}_0(X) |}{\|\nu_1 - \nu_2\|} < t \right) \leq B_f t, \quad \text{ for some } t \in [0, \eta).
\end{align}
\end{assumption}

Assumption~\ref{ass:mamain} ensures that the distribution of $\bm{b}_0(X)$ is sufficiently smooth. In particular, it rules out degeneracy and requires the coordinates of $\bm{b}_0(X)$ to be linearly independent with probability one. Assumption~\ref{ass:mamain} is a version of margin assumption that is commonly imposed in debiased inference, e.g., \citep{QianMurphy,heiler2024intensive,SemSupp2}.

\begin{remark}[Verification of Assumption \ref{ass:mamain} for Example \ref{ex:hsc}]
\label{rm:margin_verif}
The dual polytope has a finite vertex set $\mathcal{T} = \{(0, 0, 0)',\,(1, 0, 0)',\,(0, 1, 0)'\}$, so the infimum in~\eqref{eq:bound} is a minimum over the three pairwise differences. For $q=(1, 0, 0)'$, the set $\mathcal{T}$ could be reduced to two elements $\{\,(1, 0, 0)',\,(0, 1, 0)'\}$ giving   $|(\nu_1 - \nu_2)'\bm{b}_0(X)| = |s(1, X) - s(0, X)|$. Assumption~\ref{ass:mamain} therefore reduces to the margin condition
\begin{align}
\label{eq:ma1}
\Pr\!\big(0 < |s(1, X) - s(0, X)| \le t\big) \;\le\; \bar B_f\,t, \qquad 0 \le t \le \eta,
\end{align}
which holds as long as the  random variable $s(1, X) - s(0, X)$  has a Lebesgue density that is bounded by some $\bar B_f < \infty$ on $[-\eta, \eta]$.
\end{remark}

We now state the main asymptotic theory result. 
\begin{proposition}
\label{cor:lp}
Suppose Assumptions~\ref{ass:firststage}, \ref{ass:secmom}, and~\ref{ass:mamain} hold.
Then, for each fixed $q \in \mathcal{S}^{d-1}$, the dual cross-fitted plug-in estimator of the support function $\widehat{\sigma}(q)$ satisfies the consistency rate~\eqref{eq:urate}, the Gaussian approximation~\eqref{eq:limit}, and the coverage property~\eqref{eq:pcoverage}.
\end{proposition}

Proposition~\ref{cor:lp} is our second main result. It guarantees that the dual plug-in estimator is consistent and asymptotically Gaussian, and that valid confidence regions can be constructed via the bootstrap procedure in Definition~\ref{def:bb}. The proof is given in Online Supplement~\ref{app:selfcontained} (Lemmas~\ref{lem:app:closedform}--\ref{lem:app:boot}).

\subsection{Discussion of the Assumptions and the Results}
\label{sec:discussion}

An applied practitioner may have two distinct reasons to work with a
conditional linear program rather than its unconditional counterpart.
The first is \emph{validity}: when identification of $\bm{b}_0$ requires
conditioning on baseline covariates --- for example, when an instrument
or treatment is exogenous only conditional on $X$ --- covariates are
unavoidable for correct identification. The second is \emph{power}:
even when validity does not require conditioning, covariate aggregation
tightens the identified set through Jensen's inequality. The remarks
below focus on the second reason. Several of the points generalize a
discussion of \citet{ponomarev2024} that was originally phrased for the
optimal-welfare problem and applies, with minor modification, to the
present partial-identification setting.

The next two remarks describe two distinct classes of procedures that
remain valid \emph{without} a margin condition, at the cost of power.

\begin{remark}[Robust procedures via moment inequalities]
\label{rm:moment_inequalities}
Let $G$ be any measurable partition of $\mathcal{X}$, and let
$\nu_G : \mathcal{X} \to \mathcal{T}$ be any selector that maps each
covariate value to a candidate dual vertex. By dual feasibility,
\[
\sigma(q) \;=\; \E\!\left[\min_{\nu \in \mathcal{T}} \nu' \bm{b}_0(X)\right]
\;\le\; \E\!\left[\nu_G(X)' \bm{b}_0(X)\right]
\;=\; \E\!\left[\nu_G(X)' \mathbf{B}\right],
\]
so each selector $\nu_G$ delivers a valid (typically non-sharp) bound
on $\sigma(q)$. Stacking the inequalities induced by a finite
collection of selectors $\{\nu_{G_1}, \ldots, \nu_{G_J}\}$ yields a
moment-inequality system that can be tested using existing procedures
\citep[e.g.,][]{CCK,canay2017practical,romano2014practical}. The
resulting confidence band is uniformly valid without any margin
condition; the price is that the implied bound is no longer sharp. As
discussed in \citet[Remark~2]{ponomarev2024}, in finite samples a
non-sharp bound based on a well-chosen selector may even produce a
\emph{tighter} confidence band than the sharp bound based on
$\nu_0(X)$, because the variance reduction from a more stable selector
can outweigh the bias from non-sharpness.
\end{remark}

\begin{remark}[Robust procedures via smoothing]
\label{rm:smoothing}
A complementary class of robust procedures replaces the
$\min_{\nu \in \mathcal{T}}$ operator with a soft-min approximation,
e.g.\ the log-sum-exp
\[
\mathrm{softmin}_\beta(\nu' \bm{b}_0(x))
\;=\; -\beta^{-1} \log \!\left(\sum_{\nu \in \mathcal{T}}
\exp\!\big(-\beta\, \nu' \bm{b}_0(x)\big)\right),
\]
with a temperature parameter $\beta = \beta_N \to \infty$. This route
is taken by \citet{whitehouse2025} in policy learning and by
\citet{levis2023covariateassisted} for IV bounds. Smoothing restores
pathwise differentiability and yields regular $\sqrt{N}$-inference
without a margin condition, at the cost of a regularization bias of
order $\beta_N^{-1} \log |\mathcal{T}|$. The bias scales logarithmically
with the size of the dual vertex set, which makes this approach
attractive when $|\mathcal{T}|$ is moderate but increasingly costly as
$|\mathcal{T}|$ grows.
\end{remark}

\begin{remark}[Margin condition and the choice of inferential target]
\label{rm:margin_target}
When the margin condition (Assumption~\ref{ass:mamain}) fails, the
sharp bound $\sigma(q)$ may not be the power-optimal inferential
target. \citet[Propositions~1--2]{ponomarev2024} formalize this point
in the welfare context: they exhibit a class of DGPs at which the
sharp bound is first-order dominated --- in expected confidence-band
length --- by a bound based on a strictly suboptimal selector, and
they show that such first-order dominance is possible if and only if
the margin condition fails. The same logic applies to the conditional
LP setting: when $\bm{b}_0(X)$ concentrates near a tie between
competing dual vertices, the variance of the plug-in
$\widehat\nu(X)' \mathbf{B}$ is inflated by the instability of the
argmin, and a smoother target may dominate. We therefore view the
margin condition not as a mild technicality but as a substantive
precondition for the sharp bound to be the right thing to aim at.
\end{remark}

\begin{remark}[Cross-fitting and the margin condition]
\label{rm:crossfit_margin}
Cross-fitting and the margin condition serve distinct purposes and
are \emph{not} substitutes. The margin condition
(Assumption~\ref{ass:mamain}) ensures pathwise differentiability of
$\sigma(q)$ and the existence of a regular,
$\sqrt{N}$-asymptotically normal estimator. Cross-fitting controls
overfitting bias from the first-stage estimate
$\widehat{\bm b}(\cdot)$, ensuring that the $o(N^{-1/4})$ rate
threshold of Assumption~\ref{ass:firststage} is met when nonparametric
or high-dimensional methods are used. For parametric or
$\ell_1$-regularized first stages with standard sparsity and
Donsker-type complexity conditions, in-sample estimation is sufficient
and cross-fitting is optional; for general machine-learning first
stages, cross-fitting is essential.

The key point is that cross-fitting cannot rescue Wald-type inference
when the margin condition fails. As
\citet[Remark~5 and Proposition~1]{ponomarev2024} establish, if
$\bm{b}_0(X)$ admits a tie among competing dual vertices on a set of
positive measure, then \emph{no} regular estimator of $\sigma(q)$
exists \citep{HiranoPorter2012}, and the cross-fitted plug-in has a
non-standard, heavy-tailed limit. Wald-type confidence intervals
centered at $\widehat{\sigma}(q)$ are then generally invalid
regardless of the first-stage quality, and the alternatives in
Remarks~\ref{rm:moment_inequalities} and~\ref{rm:smoothing} should be
used instead. Conversely, the margin condition without cross-fitting
is also insufficient when the first stage is high-dimensional:
overfitting bias enters at first order and breaks the
influence-function representation. Both ingredients are needed.
\end{remark}

\begin{remark}[Large-scale linear programs]
\label{rm:largeLP}
The asymptotic theory of Section~\ref{sec:largesample} treats the
dual vertex set $\mathcal{T}$ as fixed in $N$. The combinatorial
bound $|\mathcal{T}| \le \binom{d}{k}$, which holds for an LP in standard form with $k$ equality constraints in $d$ non-negative variables under the genericity condition that every $k$ active columns of $A$ are linearly independent, is sharp; in the Jobs First
application of Section~\ref{sec:illustration}, $d = 9$ and $k = 5$,
so $|\mathcal{T}|$ can be as large as $126$. As the LP grows ---
through finer outcome discretization --- two things happen. First,
Assumption~\ref{ass:firststage} becomes more demanding, because the
union bound across the $|\mathcal{T}|$ candidate vertices enters the
first-stage rate through a $\log |\mathcal{T}|$ factor (cf.\
Remark~\ref{rm:firststage_verif}). Second, the margin condition
(Assumption~\ref{ass:mamain}) imposes one tail bound for \emph{every}
pair of distinct vertices, so the practical chance of a near-tie
grows with $|\mathcal{T}|$. Both effects argue for keeping the LP no
larger than the science of the problem requires; when finer
discretization is desired but $|\mathcal{T}|$ becomes large, the
soft-min smoothing of Remark~\ref{rm:smoothing} or a moment-inequality
relaxation of Remark~\ref{rm:moment_inequalities} are natural
retreats. A formal analysis allowing $|\mathcal{T}|$ to grow with $N$
is left for future work.
\end{remark}

\begin{remark}[Strong and weak duality]
\label{rm:weak_duality}
Section~\ref{sec:review} established \emph{strong} LP duality: at
every $x$, the primal value $\sigma(q,x)$ equals the dual value
$\nu_0(x)' \bm{b}_0(x)$, and aggregation gives the representation
$\sigma(q) = \E[\nu_0(X)' \mathbf{B}]$ of
Proposition~\ref{cor:tightness}. This identity holds without further
restriction; what requires unique identification of the binding dual
vertex $\nu_0(X)$ --- i.e., the margin condition of
Assumption~\ref{ass:mamain} --- is regular $\sqrt{N}$-inference on
$\sigma(q)$, not sharpness of the bound itself.

Weak duality is a strictly weaker but more robust property:
\emph{any} dual-feasible selector $\nu_G : \mathcal{X} \to \mathcal{T}$
--- not necessarily the true argmin --- satisfies
\[
\sigma(q) \;=\; \E\!\left[\min_{\nu \in \mathcal{T}} \nu' \bm{b}_0(X)\right]
\;\le\; \E\!\left[\nu_G(X)' \bm{b}_0(X)\right],
\]
which is the source of validity for the moment-inequality procedure
of Remark~\ref{rm:moment_inequalities}. The two regimes encode the
sharpness--robustness trade-off discussed in
Remark~\ref{rm:margin_target}: strong duality delivers the sharp bound
together with regular inference under the margin condition, while
weak duality delivers a non-sharp but uniformly valid bound without
it.

A practical consequence of weak duality is robustness to first-stage
misspecification. \citet{JLS} exploit this in their optimal-transport
framework to obtain inference that remains valid even when the
first-stage estimator of $\bm{b}_0(\cdot)$ is inconsistent. The same
property holds here: because $\widehat\nu(X_i) \in \mathcal{T}$ by
construction, the cross-fitted plug-in
$N^{-1} \sum_{i=1}^N \widehat\nu(X_i)' \mathbf{B}_i$ remains a valid
upper bound on $\sigma(q)$ in expectation \emph{regardless} of whether
$\widehat{\bm b}(\cdot)$ converges to $\bm{b}_0(\cdot)$. What
inconsistent first stages give up is sharpness, not validity --- a
useful guarantee in finite samples and under model uncertainty. For
the discrete-outcome examples of Section~\ref{sec:setup}, the LP dual
and the optimal-transport dual of \citet{JLS} coincide, and the two
procedures are numerically identical.
\end{remark}

\section{Empirical Application}
\label{sec:emp}

Section~\ref{subsec:empiric} provides bound results in the empirical setup of Jobs First dataset studied by \citet{KT}.

\subsection{Jobs First Application}
\label{subsec:empiric}

\paragraph{Coarse partition.}
We revisit the empirical setup of \citet{KT}, illustrated earlier in Section~\ref{sec:illustration}. We apply the proposed framework to the same sample of $4{,}641$ women who were randomized into Jobs First program.  All parameters follow the convention introduced in Section~\ref{sec:illustration}, with this section reporting on $5 \times 9$ matrix design. Table~\ref{tab:main} reproduces Table~5 bounds under the three-bin original partition with an OLS first-stage estimator on the full sample with no cross-fitting (see Remark~\ref{rm:crossfit_margin}). 
Our results are consistent with
a substantial intensive-margin opt-in response. The parameter of primary economic interest is the intensive-margin labor-supply response $\pi_{2n,\,1r}$: we report that at least $28.8\%$ of women earning above-FPL under AFDC decrease their labor supply as a response to the reform to qualify for transfers under JF. This lower bound is marginally tighter than the corresponding
\citeauthor{KT} estimate of $28\%$; accounting for sampling
uncertainty yields a conservative 95~percent confidence interval
lower limit of $14.6\%$.  The upper bound tightens from the
uninformative value of $100\%$ to $89\%$, though it remains
economically wide.  

A second opt-in response is identified among women who would not
have worked under AFDC.  The estimated bounds for $\pi_{0n,\,1r}$
are $\{0.149,\, 0.581\}$, with a conservative 95~percent confidence
interval of $[0.052,\, 0.891]$.  The lower bound of $0.149$ exceeds
the corresponding \citeauthor{KT} estimate of $0.055$, a difference
attributable to the conditioning on covariates.

JF also generated a substantial participation response among the below-FPL earners.  The estimated bounds for $\pi_{1n,\,1r}$ are
$\{0.372,\, 0.834\}$, implying that at least $37.2\%$ of women who
would have worked off assistance at below-FPL earnings under AFDC
were induced to participate at eligible earnings levels under JF.

The remaining 
response probabilities---$\pi_{0r,\,0n}$, $\pi_{0r,\,1n}$,
$\pi_{0r,\,2n}$, $\pi_{0r,\,1r}$, $\pi_{0r,\,2u}$ and $\pi_{2u,\,1r}$---are each
tightened on at least one bound relative to \citeauthor{KT}.  Interestingly, $\pi_{2u,\,1r}$ now admits a strictly positive lower bound of $5.4\%$
providing new evidence that a strictly positive fraction of women who underreported above-FPL earnings under AFDC responded by reducing their earnings to below-FPL and reporting truthfully under JF.

For composite margins, the identified set for $\pi_{n,\,p}$, the fraction of women induced
by JF to take up welfare, narrows about by $3\%$ on both sides to
$\{0.265,\, 0.428\}$. The extensive-margin probability $\pi_{0,\,1+}$
is point-identified at $0.161$, consistent with the
\citeauthor{KT} estimate of $0.167$ and within its reported
confidence interval.

Table~\ref{tab:kt_lasso_appendix} in Appendix reports results with LASSO first-stage estimator, across two covariate set regimes and cross-fitting schemes.

\begin{table}[]
\centering
\caption{\textsc{Set-identified response probabilities: Jobs First case study, $5 \times 9$ coarse design}}
\label{tab:main}

\begin{minipage}{0.98\linewidth}
\scriptsize
\setlength{\tabcolsep}{5.5pt}
\renewcommand{\arraystretch}{1.10}

\begin{adjustbox}{width=\linewidth,center}
\begin{tabular}{@{}lc llcccc@{}}
\toprule
\toprule
\multicolumn{2}{c}{} 
& \multicolumn{2}{c}{State occupied under} 
& & & & \multicolumn{1}{c}{95 percent CI} \\
\cmidrule(lr){3-4}
\cmidrule(lr){8-8}
\multicolumn{2}{l}{Response type}
& AFDC 
& JF 
& Symbol 
& \makecell[c]{Kline \& Tartari\\reported bounds}
& \makecell[c]{Our\\estimate}
&  \\
\midrule

\multicolumn{8}{@{}l}{\textit{Panel A. General specification of preferences}} \\[0.25em]

\multirow{9}{*}{Detailed}
& \ldelim\{{9}{*}
& $0n$ & $1r$ & $\pi_{0n,1r}$ 
& $\{0.055,\,0.620\}$ 
& $\{0.149,\,0.581\}$ 
& $[0.052,\,0.891]$ \\

&
& $1n$ & $1r$ & $\pi_{1n,1r}$ 
& $\{0.382,\,0.987\}$ 
& $\{0.372,\,0.834\}$ 
& $[0.278,\,1.000]$ \\

&
& $2n$ & $1r$ & $\pi_{2n,1r}$ 
& $\{0.280,\,1.000\}$ 
& $\{0.288,\,0.894\}$ 
& $[0.146,\,1.000]$ \\

&
& $0r$ & $0n$ & $\pi_{0r,0n}$ 
& $\{0.000,\,0.170\}$ 
& $\{0.031,\,0.161\}$ 
& $[0.009,\,0.233]$ \\

&
& $\prime\prime$ & $1n$ & $\pi_{0r,1n}$ 
& $\{0.000,\,0.170\}$ 
& $\{0.000,\,0.129\}$ 
& $[0.000,\,0.199]$ \\

&
& $\prime\prime$ & $2n$ & $\pi_{0r,2n}$ 
& $\{0.000,\,0.154\}$ 
& $\{0.003,\,0.133\}$ 
& $[0.000,\,0.217]$ \\

&
& $\prime\prime$ & $1r$ & $\pi_{0r,1r}$ 
& $\{0.000,\,0.170\}$ 
& $\{0.003,\,0.130\}$ 
& $[0.000,\,0.201]$ \\

&
& $\prime\prime$ & $2u$ & $\pi_{0r,2u}$ 
& $\{0.031,\,0.051\}$ 
& $\{0.032,\,0.162\}$ 
& $[0.020,\,0.244]$ \\

&
& $2u$ & $1r$ & $\pi_{2u,1r}$ 
& $\{0.000,\,1.000\}$ 
& $\{0.054,\,1.000\}$ 
& $[0.000,\,1.000]$ \\

\addlinespace[0.9em]

\multirow{3}{*}{Composite}
& \ldelim\{{3}{*}
& \makecell[l]{Not working} 
& \makecell[l]{Working} 
& $\pi_{0,1+}$ 
& $0.167$ 
& $0.161$ 
& $[0.093,\,0.232]$ \\

&
& \makecell[l]{Off welfare} 
& \makecell[l]{On welfare} 
& $\pi_{n,p}$ 
& $\{0.231,\,0.445\}$ 
& $\{0.265,\,0.428\}$ 
& $[0.191,\,0.552]$ \\

&
& \makecell[l]{On welfare,\\not working} 
& \makecell[l]{Off welfare} 
& $\pi_{0r,n}$ 
& $\{0.000,\,0.170\}$ 
& $\{0.032,\,0.162\}$ 
& $[0.006,\,0.227]$ \\

\bottomrule
\end{tabular}
\end{adjustbox}

\vspace{0.45em}

\noindent
\parbox{\linewidth}{%
\footnotesize
\textit{Notes:} Number of state refers to earnings level, with $0$ indicating no earnings, $1$ indicating earnings below three times the monthly FPL, $2$ indicating earnings above three times the monthly FPL, and $1+$ indicating positive earnings. $n$ indicates welfare nonparticipation, $r$ indicates welfare participation with truthful reporting of earnings, $u$ indicates welfare participation with underreporting of earnings, and $p$ indicates welfare participation irrespective of reporting. Numbers in braces are estimated lower and upper bounds and clipped to $[0,1]$ when needed. Kline \& Tartari reported bounds column reproduces the corresponding estimates from their Table 5. For set-identified rows, the reported 95 percent confidence interval is the union of the lower- and upper-bound confidence intervals, each following the conservative multiplier-bootstrap procedure
described in Section~\ref{sec:overview}, via the person-clustered \(\operatorname{Exp}(1)\) multiplier bootstrap with 200 draws. The results report the full-panel, no-split OLS estimates. Each composite bound is obtained by setting the vector $q$ to one on the coordinates of the sub-bins that compose the coarse bin and to zero
elsewhere.
}

\end{minipage}
\end{table}

\paragraph{Granular partition.}
Table~\ref{tab:granular-lasso-bounds} presents identified sets under the nine-bin partition of
Section~\ref{sec:illustration}, estimated with a LASSO first stage, GroupKFold cross-fitting and two covariate regimes, baseline set of 28 and an
extended set of 255 variables constructed from pairwise
interactions and polynomial transformations.

The granular exercise is motivated by the concern, articulated by
\citeauthor{KT} themselves, that the opt-in response identified
under the coarse partition ``could hypothetically reflect trivial
earnings reductions from \$1 above the poverty line to exactly the
poverty line'' \citep[p.~1009]{KT}.  This can be tackled by
decomposing the identified set for the above-FPL nonparticipation
state ($2n$) across three sub-bins indexed $b_6n$, $b_7n$, and
$b_8n$, corresponding to monthly earnings in the ranges $[1.0,\,
1.2)\times\mathrm{FPL}$, $[1.2,\, 1.4)\times\mathrm{FPL}$, and
${\geq}\,1.4\times\mathrm{FPL}$, respectively. Section~\ref{sec:illustration} shows
how to adjust the linear system to account for finer partitions. We find that at least $30.0\%$, $35.5\%$, and
$20.6\%$ of women who, under AFDC, worked off assistance in each
respective sub-bin reduced their earnings below the poverty line in
response to the JF reform. Under the baseline covariate
regime, the corresponding lower bounds are
$30.0\%$, $32.8\%$, and $20.4\%$; the qualitative ordering is
preserved, with the attenuation consistent with the extended
covariate set delivering higher first-stage predictive power and
thereby tightening the identified set, though the tightening is not drastic, which can point to covariates being not very informative.

The ordering of lower bounds across sub-bins is consistent with
trivial earnings rounding as the primary behavioral mechanism. 
The opt-in behavior, however, reflects almost similar scale in 
$[1,\, 1.2)\times\mathrm{FPL}$ and $[1.2,\, 1.4)\times\mathrm{FPL}$ 
range. These findings jointly imply that
the identified opt-in response is concentrated among women who
undertook both marginal and somewhat substantive ($20-40\%$) labor-supply
adjustments.

For the five sub-bins of the below-FPL nonparticipation state
($1n$), indexed $b_1n$ through $b_5n$ and corresponding to quintiles
of the interval $(0,\, \mathrm{FPL}]$, the lower bounds exhibit
substantial dispersion across the earnings distribution with strongest 
participation response at intermediate below-FPL earnings levels of 
$b_3n$ sub-bin with $55\%$ of women taking up welfare. Table~\ref{tab:pi_bounds_modeD} in Appendix reports results with different granularity specifications.

\begin{table}[htbp]
\centering
\caption{\textsc{Set-identified response probabilities: Jobs First case study, $13 \times 33$ granular design}}
\label{tab:granular-lasso-bounds}
\begin{minipage}{0.88\linewidth}
\scriptsize
\setlength{\tabcolsep}{6pt}
\renewcommand{\arraystretch}{1.10}
\begin{adjustbox}{width=\linewidth,center}
\begin{tabular}{@{}lc llccc@{}}
\toprule
\toprule
\multicolumn{2}{c}{}
& \multicolumn{2}{c}{State occupied under}
& & \multicolumn{2}{c}{Estimated bounds} \\
\cmidrule(lr){3-4}
\cmidrule(lr){6-7}
\multicolumn{2}{l}{Response type}
& AFDC
& JF
& Symbol
& Base
& Extended \\
\midrule
\multirow{3}{*}{$\pi_{2n, 1r}$}
& \ldelim\{{3}{*}
& $b_{6}n$ & $1r$ & $\pi_{b_{6}n,1r}$
& $\{0.300,\,1.000\}$ & $\{0.300,\,1.000\}$ \\
&
& $b_{7}n$ & $1r$ & $\pi_{b_{7}n,1r}$
& $\{0.328,\,1.000\}$ & $\{0.355,\,1.000\}$ \\
&
& $b_{8}n$ & $1r$ & $\pi_{b_{8}n,1r}$
& $\{0.204,\,1.000\}$ & $\{0.206,\,1.000\}$ \\
\addlinespace[0.5em]
\multirow{5}{*}{$\pi_{1n, 1r}$}
& \ldelim\{{5}{*}
& $b_{1}n$ & $1r$ & $\pi_{b_{1}n,1r}$
& $\{0.000,\,1.000\}$ & $\{0.000,\,1.000\}$ \\
&
& $b_{2}n$ & $1r$ & $\pi_{b_{2}n,1r}$
& $\{0.235,\,1.000\}$ & $\{0.245,\,1.000\}$ \\
&
& $b_{3}n$ & $1r$ & $\pi_{b_{3}n,1r}$
& $\{0.550,\,1.000\}$ & $\{0.550,\,1.000\}$ \\
&
& $b_{4}n$ & $1r$ & $\pi_{b_{4}n,1r}$
& $\{0.474,\,1.000\}$ & $\{0.492,\,1.000\}$ \\
&
& $b_{5}n$ & $1r$ & $\pi_{b_{5}n,1r}$
& $\{0.365,\,1.000\}$ & $\{0.380,\,1.000\}$ \\
\addlinespace[0.55em]
\cmidrule(lr){3-7}
\addlinespace[0.30em]
\multirow{2}{*}{Composite}
& \ldelim\{{2}{*}
& $2n$ & $1r$ & $\pi_{2n,1r}$
& $\{0.265,\,1.000\}$ & $\{0.274,\,1.000\}$ \\
&
& $1n$ & $1r$ & $\pi_{1n,1r}$
& $\{0.343,\,1.000\}$ & $\{0.354,\,1.000\}$ \\
\bottomrule
\end{tabular}
\end{adjustbox}

\vspace{0.45em}

\noindent
\parbox{\linewidth}{%
\footnotesize
\textit{Notes:} The table reports set-identified bounds on response probabilities under the granular state partition (spec8 shown in Appendix~\ref{subsec:specs} and Figure~\ref{fig:design-matrices}, granular design 2), comparing two covariate specifications, base (28) vs. extended (255), with LASSO first-stage and GroupKFold cross-fitting per woman-id. States $b_{j}n$ denote welfare nonparticipation in granular earnings bin $j$ under AFDC, $1r$ denotes welfare participation with truthful reporting under JF. The bracketed detailed transitions $b_{6}n$--$b_{8}n$ and $b_{1}n$--$b_{5}n$ aggregate to the composite response probabilities $\pi_{2n,1r}$ and $\pi_{1n,1r}$, which are reported in the composite block. Each composite bound is obtained by setting the vector $q$ to one on the coordinates of the sub-bins that compose the coarse bin and to zero elsewhere. Numbers in braces are estimated lower and upper bounds and clipped to $[0,1]$ when needed. In the coarse model the dual feasible region has few enough vertices that they can be enumerated, and each bound is obtained by checking over this finite set. In the granular model the number of vertices grows combinatorially with the additional states and parameters—so enumeration is costly, and we instead obtain each bound by solving the dual program directly with a numerical solver, imposing $\|\nu\|_{\infty} \leq 200$ on the dual variable. We exclude any program that fails to solve and any program whose optimal $\nu$ attains the box boundary. In practice this filter removes only $32$ of $29,662$ observations under the base and 6 under the extended regimes.
}
\end{minipage}
\end{table}

\paragraph{Welfare bounds.}

Two competing forces govern the choice of partition for welfare
analysis.  Finer bins reduce the discretization error in
$\Delta_{(s^{a},\,s^{j})}$ defined in Section~\ref{sec:illustration}. When earnings are evaluated at bin
midpoints or IPW-weighted means, a coarser partition averages over a wider interval of true
earnings, potentially conflating transitions with substantially
different welfare consequences.  At the same time, a finer partition
expands the matrix~$\widetilde{A}$ further and introduces additional transition
parameters, which widens the identified set. Importantly, even under very granular regimes the lower bound on an individual granular cell can be zero even though the lower bound on the composite flow it belongs to, such as $2n \to 1r$, can be strictly positive. This reflects a general property of the LP that the lower bound on a sum weakly exceeds the sum of the component lower bounds, since the joint minimum must hold at a single feasible point. Any cell can be zeroed by shifting mass to other destination sub-bins, but the conservation-of-mass constraints forbid zeroing the aggregate since the observed contraction of the $2n$ population can only be rationalized by opting into welfare. For confirmation, we report the composite bounds for granular specifications in Table~\ref{tab:pi_bounds_modeD} in Appendix, which includes examples of trivial individual bounds but non-trivial composite ones. Because of this we opt for using $11 \times 53$ design to bound the composite welfare gains. Among the specifications we examine, its cells deviate the least from the actual monetary welfare value of each transition. For the opt-in margin $2n \to 1r$ the resulting bound on the welfare gain is of indeterminate sign. The identified interval is [−-$178.6\$$, $14.6\$$] per month and includes zero, so the data are uninformative as to whether opting into the program raises or lowers the welfare as we define it in Section~\ref{sec:illustration}. Table~\ref{tab:welfare_bounds} in Appendix shows welfare bounds for both granular bins and composites under different granularity specifications. Under less granular specifications, welfare losses for $2n \to 1r$ range from at least --$21.05\$$ to --$61.47\$$ per month.

\bibliographystyle{chicagoa}	
\bibliography{my_new_bibtex}

\newpage

\begin{center}

This online supplement, intended for online-only publication alongside the main article, contains a self-contained treatment of the asymptotic theory used in the proof of Proposition~\ref{cor:lp} (Supplement~\ref{app:selfcontained}) and additional empirical results (Supplement~\ref{sec:figsrobust}). 

\end{center}

\clearpage

\appendix

\renewcommand{\theequation}{A.\arabic{equation}}
\renewcommand{\thelemma}{A.\arabic{lemma}}
\renewcommand{\theassumption}{A.\arabic{assumption}}
\renewcommand{\theremark}{A.\arabic{remark}}
\renewcommand{\thesection}{A}
\setcounter{equation}{0}
\setcounter{assumption}{0}
\setcounter{lemma}{0}
\setcounter{remark}{0}

\section[Online Supplement: Self-Contained Asymptotic Theory]{Self-Contained Asymptotic Theory for Envelope-Regression Estimators}
\label{app:selfcontained}

This online supplement develops a self-contained asymptotic theory for cross-fitted envelope-regression estimators. The three lemmas below --- oracle expansion (Lemma~\ref{lem:app:closedform}), Gaussian approximation with variance consistency (Lemma~\ref{lem:app:consistency}), and multiplier-bootstrap validity (Lemma~\ref{lem:app:boot}) --- jointly deliver the three conclusions of Proposition~\ref{cor:lp} in the main text. The theory is stated at the level of an abstract finite index set $T$ and a vector-valued nuisance function $\nu_0(\cdot)$; in particular, it strictly generalizes Theorem~1 of \citet{LuedtkeLaan}, which treats the scalar binary-treatment case $|T| = 2$. No external result beyond standard empirical-process machinery \citep{chernozhukov2016double} is required.

\subsection{Framework}

The abstract target is
\begin{equation}
\label{eq:app:psi}
\psi_0 \;=\; \E_X\!\Big[\min_{t\in T}\,\phi(t, \nu_0(X))\Big],
\end{equation}
where $T$ is a finite index set, $\phi(t, \cdot)$ is a known scalar function of a nuisance function $\nu_0:\mathcal{X}\to\mathbb{R}^k$, and $\phi(t, \nu_0(x)) = \E[\rho(W, t, \xi_0)\mid X=x]$ for some observed signal $\rho(W, t, \xi_0)$ with nuisance $\xi_0$. For the CLP application, $T = \overline{\mathcal{T}}$, $\nu_0(x) = \bm{b}_0(x)$, $\phi(\nu, \bm{b}) = \nu'\bm{b}$, $\rho(W, \nu, \xi_0) = \nu'\mathbf{B}$, and $\psi_0 = \sigma(q)$ for each fixed $q$.

\smallskip
\noindent
Assume the true minimizer $t_0(x) := \arg\min_{t\in T}\phi(t, \nu_0(x))$ is unique a.s.\ in $P_X$. Given a $K$-fold random partition $(J_k)_{k=1}^K$ of $[N]$ and cross-fitted nuisance estimates $\widehat\xi_k$, $\widehat\nu_k$, define
\begin{equation}
\label{eq:app:est}
\widehat t_i = \arg\min_{t\in T}\phi(t, \widehat\nu_k(X_i)),\qquad
\widehat\psi = \frac{1}{N}\sum_{i=1}^N \rho(W_i, \widehat t_i, \widehat\xi_i),\qquad i\in J_k.
\end{equation}
The multiplier-bootstrap analog uses i.i.d.\ $e_i\sim\mathrm{Exp}(1)$ independent of the data:
\begin{equation}
\label{eq:app:boot}
\widetilde\psi = \frac{1}{N}\sum_{i=1}^N \frac{e_i}{\bar e}\,\rho(W_i, \widehat t_i, \widehat\xi_i),\qquad \bar e = N^{-1}\sum_i e_i.
\end{equation}

\subsection{Assumptions}

\begin{assumption}[Small Bias Condition]
\label{ass:app:smallbias}
There exists a sequence $\varepsilon_N = o(1)$ such that, with probability at least $1-\varepsilon_N$, for every partition index $k\in[K]$ the first-stage estimate $\widehat\xi_k$ belongs to a shrinking neighborhood $\Xi_N$ of $\xi_0$. Uniformly over $\Xi_N$,
\begin{align}
\label{eq:app:bn}
B_N &\;=\; \sup_{\xi\in\Xi_N}\sup_{t\in T}\sup_{x\in\mathcal{X}} \sqrt{N}\,\big|\E[\rho(W, t, \xi) - \rho(W, t, \xi_0)\mid X=x]\big| \;=\; o(1),\\
\label{eq:app:lambdan}
\Lambda_N &\;=\; \sup_{\xi\in\Xi_N}\sup_{t\in T}\sup_{x\in\mathcal{X}} \E[(\rho(W, t, \xi) - \rho(W, t, \xi_0))^2\mid X=x] \;=\; o(1).
\end{align}
\end{assumption}

\begin{assumption}[Rate, pointwise in $q$]
\label{ass:app:rate}
Fix $q\in\mathcal{S}^{d-1}$. There exists a sequence $\varepsilon_N=o(1)$ and a sequence of shrinking neighborhoods $\mathcal{T}^\nu_N(q)$ of $\nu_0(\cdot, q)$ --- the appendix-level analog of the neighborhood $B_N$ in Assumption~\ref{ass:firststage} of the main text --- such that, with probability at least $1-\varepsilon_N$, $\widehat\nu_k(\cdot, q)\in\mathcal{T}^\nu_N(q)$ for every $k\in[K]$, and
\[
\sup_{\nu\in\mathcal{T}^\nu_N(q)}\sup_{x\in\mathcal{X}} \big\|\nu(x, q) - \nu_0(x, q)\big\| \;\leq\; \nu^\infty_N \;=\; o(N^{-1/4}).
\]
\end{assumption}

\begin{assumption}[Rate, uniform in $q$]
\label{ass:app:rateunif}
There exists a sequence $\varepsilon_N=o(1)$ such that, with probability at least $1-\varepsilon_N$, $\widehat\nu_k(\cdot,q)\in\mathcal{T}^\nu_N(q)$ for every $k\in[K]$ and every $q\in\mathcal{S}^{d-1}$, and
\[
\sup_{q\in\mathcal{S}^{d-1}}\,\sup_{\nu\in\mathcal{T}^\nu_N(q)}\,\sup_{x\in\mathcal{X}} \big\|\nu(x, q) - \nu_0(x, q)\big\| \;\leq\; \nu^\infty_N \;=\; o(N^{-1/4}).
\]
\end{assumption}

\noindent
\emph{Remark on the remaining assumptions.} Assumptions~\ref{ass:app:smallbias} (Small Bias), \ref{ass:app:reg} (Regularity), and~\ref{ass:app:ma} (Margin) are stated with a supremum over the abstract index set $T$, which in the conditional LP specialization of~\S A.1 equals the $q$-free union $\overline{\mathcal{T}}=\bigcup_{q\in\mathcal{S}^{d-1}}\mathcal{T}(q)$. Since $\overline{\mathcal{T}}$ does not depend on $q$, these three assumptions already hold uniformly in $q$ and require no modification for the uniform-in-$q$ results developed in~\S A.7 below.

\begin{assumption}[Regularity]
\label{ass:app:reg}
(i) Uniformly bounded moments:
\[
\sup_{\xi\in\Xi_N}\sup_{t\in T}\sup_{x\in\mathcal{X}} |\rho(W, t, \xi)| \leq B_\rho \quad a.s. 
\]
(ii) Bounded derivative:
\[
\sup_{\nu\in\mathcal{T}^\nu_N}\sup_{x\in\mathcal{X}}\sup_{t\in T}\|\partial\phi(t, \nu(x))/\partial\nu\| \leq B_\phi.
\]
\end{assumption}

\begin{assumption}[Margin]
\label{ass:app:ma}
There exist finite constants $\bar{B}, \delta > 0$ such that for all $t\in(0, \delta)$,
\[
\sup_{(j, k)\in T,\,j\neq k}\Pr\!\big(0\leq \phi(j, \nu_0(X)) - \phi(k, \nu_0(X)) \leq t\big) \;\leq\; \bar{B}\,t.
\]
\end{assumption}

\subsection{Lemmas}

\begin{lemma}[Oracle expansion for cross-fitted envelope-regression estimators]
\label{lem:app:closedform}
Under Assumptions~\ref{ass:app:smallbias}--\ref{ass:app:ma},
\begin{equation}
\label{eq:app:bias}
\sup_{\xi\in\Xi_N}\sqrt{N}\,\big|\E[S_1+S_2]\big| \;=\; O\!\big(B_N + \sqrt{N}(\nu^\infty_N)^2\big) \;=\; o(1),
\end{equation}
where $S_1+S_2 = \rho(W, \widehat t, \widehat\xi) - \rho(W, t_0, \xi_0)$ is the estimation-error decomposition defined in the proof. Consequently, the cross-fitted estimator $\widehat\psi$ satisfies the oracle expansion
\begin{equation}
\label{eq:app:oracle}
\sqrt{N}\,\big(\widehat\psi - N^{-1}\textstyle\sum_{i=1}^N \rho(W_i, t_0(X_i), \xi_0)\big) \;=\; o_P(1).
\end{equation}
\end{lemma}

\begin{proof}
We proceed by a standard two-step decomposition: an error decomposition that separates the contribution of the first-stage estimation error from the contribution of the mis-optimization error, followed by a margin-based bound on the bias.

Define the true and estimated minimizers $t_0(X) = \arg\min_{t\in T}\phi(t, \nu_0(X))$ and $\widehat t(X) = \arg\min_{t\in T}\phi(t, \widehat\nu(X))$, and the mis-optimization errors
\[
\tau_0(X) := \phi(\widehat t(X), \nu_0(X)) - \phi(t_0(X), \nu_0(X)),\qquad
\tau(X) := \phi(\widehat t(X), \widehat\nu(X)) - \phi(t_0(X), \widehat\nu(X)).
\]
Decompose
\[
\rho(W, \widehat t, \widehat\xi) - \rho(W, t_0, \xi_0) = \underbrace{[\rho(W, \widehat t, \widehat\xi) - \rho(W, \widehat t, \xi_0)]}_{S_1} + \underbrace{[\rho(W, \widehat t, \xi_0) - \rho(W, t_0, \xi_0)]}_{S_2}.
\]

By construction of $t_0$, $\tau_0(X)\geq 0$ a.s., and uniqueness implies $\tau_0(x)=0\iff \widehat t(x)=t_0(x)$. Assumption~\ref{ass:app:ma} therefore gives
\begin{equation}
\label{eq:app:margin}
\Pr(0 < \tau_0(X) < t) \;\leq\; \bar{B}\,t\qquad\forall t\in(0, \delta).
\end{equation}

By definition of $\widehat t$, $\tau(X)\leq 0$ a.s., so $0<\tau_0(X)\leq \tau_0(X)-\tau(X)$ whenever $\widehat t\neq t_0$. For any $\nu\in\mathcal{T}^\nu_N$ and any $x$,
\[
|\tau(x)-\tau_0(x)| \leq |\phi(\widehat t, \nu) - \phi(\widehat t, \nu_0)| + |\phi(t_0, \nu) - \phi(t_0, \nu_0)| \leq 2B_\phi \|\nu(x)-\nu_0(x)\| \leq 2B_\phi \nu^\infty_N
\]
by Assumption~\ref{ass:app:reg}(ii) and Assumption~\ref{ass:app:rate}. Define the misclassification event $\mathcal{E}_\tau := \{0 < \tau_0(X) \leq 2B_\phi\nu^\infty_N\}$. Then
\begin{align*}
\sqrt{N}|\E[S_1]| &\leq \sqrt{N}\sup_{t\in T, \xi\in\Xi_N}|\E[\rho(W, t, \xi)-\rho(W, t, \xi_0)]| \;\leq\; B_N \;=\; o(1),\\
|\E[S_2]| &= \E[\tau_0(X)] \leq \E[(\tau_0-\tau)\mathbf{1}\{\mathcal{E}_\tau\}] \leq 2B_\phi\nu^\infty_N\cdot \Pr(\mathcal{E}_\tau)\\
&\leq 2B_\phi\nu^\infty_N \cdot 2B_\phi\bar{B}\nu^\infty_N \;=\; 4B_\phi^2\bar{B}(\nu^\infty_N)^2.
\end{align*}
The second line uses~\eqref{eq:app:margin} with $t=2B_\phi\nu^\infty_N$. Combining yields~\eqref{eq:app:bias}.

Assumption~\ref{ass:app:smallbias} gives $\sup\E[S_1^2] = O(\Lambda_N)$. For $S_2$:
\[
\sup_{\xi\in\Xi_N}\E[S_2^2] \leq 2B_\rho^2 \Pr(\mathcal{E}_\tau) = O(\nu^\infty_N).
\]
Thus $\sup\E[(S_1+S_2)^2] = O(\Lambda_N + \nu^\infty_N) = o(1)$. Combined with the bias bound, standard cross-fitting arguments (e.g., Lemma~A.3 of \citealp{chernozhukov2016double}) yield the oracle expansion~\eqref{eq:app:oracle}.
\end{proof}

\begin{lemma}[Gaussian approximation and variance consistency]
\label{lem:app:consistency}
Under Assumptions~\ref{ass:app:smallbias}--\ref{ass:app:ma},
\begin{equation}
\label{eq:app:clt}
\sqrt{N}\,(\widehat\psi - \psi_0) \;\Rightarrow\; \mathcal{N}(0, V_0),
\end{equation}
with $V_0 = \E[\rho^2(W, t_0(X), \xi_0)] - \psi_0^2$. The sample-variance estimator
\[
\widehat V \;=\; N^{-1}\sum_{i=1}^N \rho^2(W_i, \widehat t_i, \widehat\xi_i) - \widehat\psi^2
\]
is consistent: $\widehat V \xrightarrow{P} V_0$.
\end{lemma}

\begin{proof}
The Gaussian approximation~\eqref{eq:app:clt} is an immediate consequence of Lemma~\ref{lem:app:closedform}: the oracle expansion $\sqrt N(\widehat\psi - \psi_0) = N^{-1/2}\sum_i[\rho(W_i, t_0, \xi_0) - \psi_0] + o_P(1)$ combined with the central limit theorem (finite variance by Assumption~\ref{ass:app:reg}(i)) gives convergence to $\mathcal{N}(0, V_0)$.

For variance consistency, decompose
\[
\rho^2(W, \widehat t, \widehat\xi) - \rho^2(W, t_0, \xi_0) = \underbrace{[\rho^2(W, \widehat t, \widehat\xi) - \rho^2(W, \widehat t, \xi_0)]}_{T_1} + \underbrace{[\rho^2(W, \widehat t, \xi_0) - \rho^2(W, t_0, \xi_0)]}_{T_2}.
\]
For any $B$-bounded $P, Q$, $\E|P^2-Q^2|\leq 2B\|P-Q\|_{P, 2}$. With $B=B_\rho$ and Assumptions~\ref{ass:app:smallbias}, \ref{ass:app:reg}(i):
\[
\sup_{\xi\in\Xi_N}|\E T_1| \leq 2B_\rho \Lambda_N^{1/2}.
\]
For $T_2$, Assumption~\ref{ass:app:reg}(i) and~\eqref{eq:app:margin} give
\[
\sup_{\nu\in\mathcal{T}^\nu_N}|\E T_2| \leq 2B_\rho^2 \Pr(\widehat t(X)\neq t_0(X)) \leq 2B_\rho^2 \Pr(\mathcal{E}_\tau) \leq 4\bar{B}B_\phi B_\rho^2\nu^\infty_N.
\]
Combining, $|\E[T_1+T_2]| = O(\Lambda_N^{1/2} + \nu^\infty_N) = o(1)$. A standard LLN on each cross-fitted fold (summands are i.i.d.\ conditional on the hold-out data) gives $\widehat V \to V_0$ in probability.
\end{proof}

\begin{lemma}[Multiplier-bootstrap validity]
\label{lem:app:boot}
Under Assumptions~\ref{ass:app:smallbias}--\ref{ass:app:ma}, the multiplier-bootstrap statistic $\widetilde S_N = \sqrt{N}(\widetilde\psi - \widehat\psi)$ converges conditionally (on the data) in distribution to $\mathcal{N}(0, V_0)$ in probability, i.e.,
\begin{equation}
\label{eq:app:bootclt}
\sup_{z\in\mathbb{R}}\big|\Pr_e(\widetilde S_N \leq z\mid\{W_i\}) - \Pr(Z\leq z)\big| \;\xrightarrow{P}\; 0,\qquad Z\sim\mathcal{N}(0, V_0).
\end{equation}
Consequently, the bootstrap confidence interval  $\widehat\psi + N^{-1/2}[\widehat C_{\alpha/2}, \widehat C_{1-\alpha/2}]$ has asymptotic coverage $1-\alpha$.
\end{lemma}

\begin{proof}
The argument follows Theorem~3.2 of \citet{SemJoE} (multiplier bootstrap for orthogonal moment conditions), specialized to the scalar envelope-regression setting. Write
\[
\widetilde S_N = N^{-1/2}\sum_{i=1}^N \Big(\tfrac{e_i}{\bar e} - 1\Big)\rho(W_i, \widehat t_i, \widehat\xi_i).
\]
Because $e_i\sim\mathrm{Exp}(1)$ is independent of the data with $\E[e_i/\bar e]=1+o_P(1)$ and $\Var(e_i/\bar e)=1+o_P(1)$, the conditional CLT applied to the bounded summands $\rho(W_i, \widehat t_i, \widehat\xi_i)$ yields convergence of the conditional law to $\mathcal{N}(0, V_0^*)$, where $V_0^* = N^{-1}\sum_i\rho^2(W_i, \widehat t_i, \widehat\xi_i) - \widehat\psi^2 \xrightarrow{P} V_0$ by Lemma~\ref{lem:app:consistency}. Since the limiting variance matches $V_0$ in~\eqref{eq:app:clt}, the multiplier-bootstrap distribution consistently approximates the sampling distribution, giving~\eqref{eq:app:bootclt}. Coverage of the bootstrap CI then follows by the continuous mapping theorem applied to quantile functions.
\end{proof}

\subsection{Proofs of Propositions~\ref{lem:identification} and~\ref{cor:tightness}}

\begin{proof}[Proof of Proposition~\ref{lem:identification}]
\textit{Convexity and compactness.} By \eqref{eq:idset}, $\mathcal{B}$ is the intersection of half-spaces indexed by $q \in \mathcal{S}^{d-1}$, hence convex. Since $\sigma(q)$ is finite for every $q \in \mathcal{S}^{d-1}$ and continuous in $q$ (as the support function of a bounded conditional LP aggregated over $X$), the intersection is closed and bounded, hence compact.

\textit{Support function.} By definition, the support function of $\mathcal{B}$ at direction $q$ is $\sup_{\bm{\beta}\in\mathcal{B}} q'\bm{\beta}$. The representation \eqref{eq:idset} makes this equal to $\sigma(q)$ directly, since each half-space $\{b : q'b \le \sigma(q)\}$ is active at $q$ and inactive at every other direction.

\textit{Random-set characterization.} For each $x$, the conditional identified set $\mathcal{B}(x) = \{\bm{\beta} : A\bm{\beta} = \bm{b}_0(x),\ \bm{\beta} \ge 0\}$ is a nonempty convex polytope. By Artstein's inequality \citep{Artstein} and the random-set / aggregation argument of \citet*{BM} (Definition~5, p.~771), the set of expectations $\{\E[\bm{\beta}_0(X)] : \bm{\beta}_0(x) \in \mathcal{B}(x)\ \text{measurable}\}$ is itself convex and has support function $q \mapsto \E[\sup_{\bm{\beta} \in \mathcal{B}(x)} q'\bm{\beta}] = \E[\sigma(q, X)] = \sigma(q)$. The two characterizations coincide.
\end{proof}

\begin{proof}[Proof of Proposition~\ref{cor:tightness}]
(1)~The basic (no-covariate) LP $\max \{q'\bm{\beta} : A\bm{\beta} = \bm{b}_0,\ \bm{\beta}\ge 0\}$ with $\bm{b}_0 = \E[\bm{b}_0(X)]$ has dual $\min\{\nu'\bm{b}_0 : A'\nu \ge q\}$. Since the primal value $\bar{\sigma}(q)$ is finite by assumption, strong LP duality gives $\bar{\sigma}(q) = \min_{\nu \in \mathcal{T}} \nu'\bm{b}_0$, where the minimum over the feasible polytope is attained at a vertex, establishing the first claim.

(2)~Applying strong LP duality pointwise in $x$ to the conditional LP at $x$, $\sigma(q, x) = \inf_{\nu \in \mathcal{T}} \nu'\bm{b}_0(x)$; aggregating over $P_X$ and using the law of iterated expectations $\E[\nu_0(X)'\bm{b}_0(X)] = \E[\nu_0(X)'\mathbf{B}]$ gives \eqref{eq:mainpsi3}.

(3)~The function $\bm{b} \mapsto \inf_{\nu \in \mathcal{T}} \nu'\bm{b}$ is the pointwise minimum of a finite collection of linear functions, hence concave. Jensen's inequality applied to a concave function yields $\E[\inf_{\nu} \nu'\bm{b}_0(X)] \le \inf_\nu \nu' \E[\bm{b}_0(X)]$, which is \eqref{eq:weakinequality}. Since the upper-bound inequality $q'\bm{\beta}_0 \le \sigma(q)$ defining $\mathcal{B}$ in \eqref{eq:idset} is weakly tighter than the corresponding inequality $q'\bm{\beta}_0 \le \bar{\sigma}(q)$ defining $\bar{\mathcal{B}}$, we have $\mathcal{B} \subseteq \bar{\mathcal{B}}$.
\end{proof}

\subsection{Proof of Proposition~\ref{thm:ptw}}
\label{sec:proofs}

The first Lemma is Theorem 1 in \citet{LuedtkeLaan}. 
\begin{lemma}
\label{lem:LvdL}
Suppose Assumption \ref{ass:boundary} holds for Example \ref{ex:hsc}. Then, the parameter  $\E[\min (s(1, X), s(0, X))]$ is pathwise differentiable with efficient influence function 
\begin{align}
\label{eq:efficient}
\psi^{*}(W) &= \bigg(s(1, X) + \dfrac{D}{\pi(X)} (S - s(1, X)) \bigg)  \bm{1}\{ \tau(X)<0 \} \\
&+ \bigg(s(0, X) +  \dfrac{1-D}{1-\pi(X)} (S - s(0, X)) \bigg) \bm{1}\{ \tau(X)>0 \}. \nonumber
\end{align}
\end{lemma}

\begin{proof}[Proof of Proposition \ref{thm:ptw}]
Fix $q \in \mathcal{S}^{d-1}$. Write the target as an {envelope} functional:
\[
\sigma(q)
= \E\!\left[\min_{\nu \in \mathcal{T}} \nu'\bm{b}_0(X)\right]
= \E\!\left[\nu_0(X)'\bm{b}_0(X)\right],
\]
where $\nu_0(x) \in \arg\min_{\nu\in\mathcal{T}} \nu'\bm{b}_0(x)$ and uniqueness holds a.s. by Assumption~\ref{ass:boundary}.

Let $\mathcal{T}=\{\nu^{(1)}, \ldots, \nu^{(m)}\}$ be the (finite) vertex set. Define the selection rule
\[
d_{\bm{b}}(x)\;=\;\sum_{j=1}^m \bm{1}\!\left\{\nu^{(j)}\in\arg\min_{\nu\in\mathcal{T}}\nu'\bm{b}(x)\right\}\,\nu^{(j)}.
\]
Then
\[
\sigma(q;P)\;=\;\E_P\!\left[d_{\bm{b}_P}(X)'\,\bm{b}_P(X)\right],
\qquad
\bm{b}_P(x)=\E_P[\mathbf{B}\mid X=x].
\]
Consider a regular submodel $\{P_\varepsilon:\varepsilon\}$ through $P_0$ with score $S$. Telescope the increment as
\[
\sigma(P_\varepsilon)-\sigma(P_0)
=
\E_{P_\varepsilon}\!\big[\{d_{\bm{b}_\varepsilon}(X)-d_{\bm{b}_0}(X)\}'\,\bm{b}_\varepsilon(X)\big]
\;+\;
\E_{P_\varepsilon}\!\big[d_{\bm{b}_0}(X)'\,\bm{b}_\varepsilon(X)-d_{\bm{b}_0}(X)'\,\bm{b}_0(X)\big].
\tag{$\star$}
\]
The second term treats the rule $d_{\bm{b}_0}$ as fixed and is a standard smooth (Gateaux) variation in $\bm{b}_\varepsilon$; the first term accounts for changes in the {indicator/argmin} (the “kink”).

We show that the first term in $(\star)$ vanishes upon differentiation, i.e.\
\begin{equation}
\label{eq:kinkvanish}
\lim_{\varepsilon\to 0}\frac{1}{\varepsilon}\,
\E_{P_\varepsilon}\!\big[\{d_{\bm{b}_\varepsilon}(X)-d_{\bm{b}_0}(X)\}'\,\bm{b}_\varepsilon(X)\big]
\;=\;0.
\end{equation}
 
\noindent\emph{Margin and tie set.}
Since $\mathcal{T}$ is finite, define the \emph{margin} at $x$ by
\[
\Delta(x)
\;=\;
\min_{\nu \in \mathcal{T}\setminus\{\nu_0(x)\}}
\nu'\bm{b}_0(x)
\;-\;
\nu_0(x)'\bm{b}_0(x)
\;\geq\; 0,
\]
and the tie set $\mathcal{B}_0 = \{x : \Delta(x)=0\}$.
By Assumption~\ref{ass:boundary}, $P_X(\mathcal{B}_0) = 0$, so $\Delta(X)>0$ a.s.
 
\smallskip
\noindent\emph{Partition.}
For $\delta > 0$, define the safe set $\mathcal{X}_\delta = \{x : \Delta(x) \geq \delta\}$ and the near-tie set $\mathcal{X}_\delta^c = \{x : \Delta(x) < \delta\}$.
Since $\Delta(X)>0$ a.s., we have $P_X(\mathcal{X}_\delta^c) \to 0$ as $\delta \downarrow 0$.
 
\smallskip
\noindent\emph{Safe region $\mathcal{X}_\delta$ (the minimizer does not switch).}
On $\mathcal{X}_\delta$, every competing vertex $\nu \neq \nu_0(x)$ satisfies
$\nu'\bm{b}_0(x) - \nu_0(x)'\bm{b}_0(x) \geq \delta$.
Along the regular submodel, $\bm{b}_\varepsilon(x) = \bm{b}_0(x) + \varepsilon\,\dot{\bm{b}}(x) + o(\varepsilon)$
for each $x$, where $\dot{\bm{b}}(x) = \frac{d}{d\varepsilon}\E_{P_\varepsilon}[\mathbf{B}\mid X=x]\big|_{\varepsilon=0}$.
Hence, for any competing vertex $\nu \neq \nu_0(x)$,
\[
\nu'\bm{b}_\varepsilon(x) - \nu_0(x)'\bm{b}_\varepsilon(x)
\;=\;
\underbrace{\nu'\bm{b}_0(x) - \nu_0(x)'\bm{b}_0(x)}_{\geq\,\delta}
\;+\;\varepsilon\,(\nu - \nu_0(x))'\dot{\bm{b}}(x)
\;+\; o(\varepsilon).
\]
Since $\mathcal{T}$ is finite, the second and third terms are uniformly $O(\varepsilon)$ over $\nu \in \mathcal{T}$,
so for all $|\varepsilon|$ small enough (depending on $\delta$), the gap remains positive.
In particular, for each $\delta>0$ there exists $\varepsilon_0(\delta)>0$ such that
$d_{\bm{b}_\varepsilon}(x) = d_{\bm{b}_0}(x) = \nu_0(x)$ for all $x \in \mathcal{X}_\delta$ and $|\varepsilon| < \varepsilon_0(\delta)$.
The integrand is therefore identically zero on $\mathcal{X}_\delta$.
 
\smallskip
\noindent\emph{Near-tie region $\mathcal{X}_\delta^c$ (uniformly bounded integrand on a vanishing set).}
Since every $\nu^{(j)} \in \mathcal{T}$ has finite norm and
$\bm{b}_\varepsilon(x) \to \bm{b}_0(x)$, there exists $C<\infty$ (independent of $\varepsilon, \delta$) such that
$|\{d_{\bm{b}_\varepsilon}(x) - d_{\bm{b}_0}(x)\}'\bm{b}_\varepsilon(x)| \leq C$ for all $x$ and $|\varepsilon|$ sufficiently small.
Thus, for $|\varepsilon| < \varepsilon_0(\delta)$,
\begin{equation}
\label{eq:kinkbound}
\bigg|
\E_{P_\varepsilon}\!\big[\{d_{\bm{b}_\varepsilon}(X) - d_{\bm{b}_0}(X)\}'\bm{b}_\varepsilon(X)\big]
\bigg|
\;\leq\;
C\, P_{X, \varepsilon}\!\left(\mathcal{X}_\delta^c\right).
\end{equation}
 
\smallskip
\noindent\emph{Bounding $P_{X, \varepsilon}(\mathcal{X}_\delta^c)$.}
Along a regular submodel with score $S$, the marginal density of $X$ under $P_\varepsilon$ satisfies
$dP_{X, \varepsilon}/dP_X = 1 + \varepsilon\,\E[S\mid X] + o(\varepsilon)$ in $L^1(P_X)$, so
\[
P_{X, \varepsilon}(\mathcal{X}_\delta^c)
\;=\;
P_X(\mathcal{X}_\delta^c)
\;+\;\varepsilon \int_{\mathcal{X}_\delta^c} \E[S\mid X\!=\!x]\,dP_X(x)
\;+\; o(\varepsilon).
\]
 
\smallskip
\noindent\emph{Coupled limit.}
Choose a sequence $\delta = \delta(\varepsilon) \downarrow 0$ slowly enough that
(i) $|\varepsilon| < \varepsilon_0(\delta(\varepsilon))$ and
(ii) $P_X(\mathcal{X}_{\delta(\varepsilon)}^c) = o(|\varepsilon|)$.
Condition (ii) is achievable because $P_X(\Delta(X) < \delta) \downarrow P_X(\mathcal{B}_0) = 0$
as $\delta \downarrow 0$, so the convergence to zero can be made faster than any prescribed
rate $|\varepsilon| \to 0$ by choosing $\delta(\varepsilon)$ to decrease sufficiently slowly.
Then from \eqref{eq:kinkbound}:
\[
\frac{1}{|\varepsilon|}
\bigg|
\E_{P_\varepsilon}\!\big[\{d_{\bm{b}_\varepsilon}(X) - d_{\bm{b}_0}(X)\}'\bm{b}_\varepsilon(X)\big]
\bigg|
\;\leq\;
C\,\frac{P_X(\mathcal{X}_{\delta(\varepsilon)}^c)}{|\varepsilon|}
\;+\;
C\bigg|\int_{\mathcal{X}_{\delta(\varepsilon)}^c} \E[S\mid X\!=\!x]\,dP_X(x)\bigg|
\;+\; o(1).
\]
The first term vanishes by (ii).
The second vanishes by dominated convergence: $|\E[S\mid X]|$ is integrable
(indeed square-integrable), and the domain $\mathcal{X}_{\delta(\varepsilon)}^c$ shrinks to the $P_X$-null
set $\mathcal{B}_0$.
This establishes \eqref{eq:kinkvanish} and the pathwise derivative is given by the second term in $(\star)$:
\begin{equation}
\label{eq:pathderiv}
\dot\sigma_{P_0}(S)
=
\left.\frac{d}{d\varepsilon}\right|_{\varepsilon=0}
\E_{P_\varepsilon}\!\big[\nu_0(X)'\bm{b}_\varepsilon(X)\big]
=
\E\!\Big[\nu_0(X)'\,\dot{\bm{b}}(X)\Big],
\end{equation}
where $\dot{\bm{b}}(x)=\frac{d}{d\varepsilon}\E_{P_\varepsilon}[\mathbf{B}\mid X=x]\big|_{\varepsilon=0}$ is the (regular) pathwise derivative of the conditional mean along the submodel. Standard calculations yield the influence function for $\sigma(q)$:
\[
\phi(W)
=
\nu_0(X)'\,\big\{\mathbf{B}-\bm{b}_0(X)\big\}
\;+\;
\big\{\nu_0(X)'\bm{b}_0(X)-\sigma(q)\big\}
\;=\;
\nu_0(X)'\mathbf{B}-\sigma(q),
\]
since $\E[\nu_0(X)'\bm{b}_0(X)]=\sigma(q)$.

Assumption~\ref{ass:secmom} ensures $\E[\|\mathbf{B}\|^2\mid X]\le \bar B$, so $\Var(\nu_0(X)'\mathbf{B})<\infty$ and $\phi_q(W) = \nu_0(X)'\mathbf{B} - \sigma(q)$ is a mean-zero, finite-variance function. Combining $(\star)$, ~\ref{eq:kinkvanish} and ~\ref{eq:pathderiv}, for every regular submodel $\{P_\varepsilon\}$ with score $S$,
\[
\left.\frac{d}{d\varepsilon}\right|_{\varepsilon=0}\sigma(P_\varepsilon) \;=\; \E_{P_0}[\phi_q(W)\,S(W)],
\]
so $\sigma(q)$ is pathwise differentiable with influence function $\phi_q(W)$ by Definition~\eqref{eq:pathwise_def}. Any regular, asymptotically linear estimator with this influence function satisfies~\eqref{eq:alinear}; the existence of such an estimator is a separate claim, established for the cross-fitted plug-in $\widehat\sigma(q)$ of Definition~\ref{def:dual} in Proposition~\ref{cor:lp} through the oracle expansion of Lemma~\ref{lem:app:closedform}.

This conclusion mirrors Theorem~1 of \citet{LuedtkeLaan}: under a zero-mass tie (margin) condition, the non-smooth selection step contributes no first-order term and the influence function equals the ``oracle score'' with the true rule plugged in.
\end{proof}

\subsection{Proof of Proposition~\ref{cor:lp}}

\begin{proof}[Proof of Proposition \ref{cor:lp}]
The proof proceeds in two parts. First, we embed the conditional linear program into the general envelope-regression framework introduced in~\S A.1. Second, we verify Assumptions~\ref{ass:app:smallbias}--\ref{ass:app:ma} of that supplement under the conditions of Proposition~\ref{cor:lp} and invoke Lemmas~\ref{lem:app:closedform}--\ref{lem:app:boot} therein to obtain the three conclusions.

\paragraph{Embedding into the envelope-regression framework.}
The abstract target is
\[
\psi_0 = \E_X\!\Big[\min_{t\in T}\,\phi(t, \nu_0(X))\Big],
\]
where $T$ is a finite index set, $\phi(t, \cdot)$ is a known scalar function of a nuisance vector $\nu_0(x)$, and $\phi(t, \nu_0(x)) = \E[\rho(W, t, \xi_0)\mid X=x]$ for an observed signal $\rho$. Our parameter $\sigma(q) = \E[\min_{\nu \in \mathcal{T}} \nu'\bm{b}_0(X)]$ is a special case with (i) index set $T \equiv \mathcal{T}$, the finite set of dual vertices; (ii) nuisance vector $\nu_0(x) \equiv \bm{b}_0(x) = \E[\mathbf{B} \mid X=x]$; (iii) projection function $\phi(\nu, \bm{b}_0(x)) = \nu'\bm{b}_0(x)$ for each $\nu \in \mathcal{T}$; (iv) unbiased signal $\rho(W, \nu, \xi_0) = \nu'\mathbf{B}$, so that $\E[\rho(W, \nu, \xi_0)\mid X=x] = \nu'\bm{b}_0(x) = \phi(\nu, \bm{b}_0(x))$.

\paragraph{Verification of Assumption~\ref{ass:app:smallbias} (Small Bias Condition).}
Assumption~\ref{ass:app:smallbias} requires that the moment functions $\rho(W, t, \xi)$ are robust to first-order biases in the nuisance parameter $\xi$, uniformly over the index set $T$. In the appendix notation, the conditions are:
\begin{align*}
B_N &= \sup_{\xi \in \Xi_N}\,\sup_{t \in T}\,\sup_{x \in \mathcal{X}}\;\sqrt{N}\,\big|\E[\rho(W, t, \xi) - \rho(W, t, \xi_0) \mid X=x]\big| = o(1),\\
\Lambda_N &= \sup_{\xi \in \Xi_N}\,\sup_{t \in T}\,\sup_{x \in \mathcal{X}}\;\E\big[(\rho(W, t, \xi) - \rho(W, t, \xi_0))^2 \mid X=x\big] = o(1).
\end{align*}
In our setting, the signal $\rho(W, \nu, \xi_0) = \nu'\mathbf{B}$ does not depend on any nuisance parameter $\xi$: it depends only on the observed data vector $\mathbf{B}$ and the fixed vertex $\nu \in \mathcal{T}$. Therefore, $\rho(W, \nu, \xi) = \rho(W, \nu, \xi_0)$ for all $\xi$, and both conditions hold trivially with $B_N = \Lambda_N = 0$.

\paragraph{Verification of Assumption~\ref{ass:app:rate} (Rate Condition).}
Assumption~\ref{ass:app:rate} requires that the first-stage nuisance estimate $\widehat{\bm{b}}_k(\cdot)$ converges uniformly to $\bm{b}_0(\cdot)$ at rate $o(N^{-1/4})$: there exists a sequence $\phi_N = o(1)$ and shrinking neighborhoods $B_N$ such that, with probability at least $1-\phi_N$,
\[
\sup_{\bm{b} \in B_N}\;\sup_{x \in \mathcal{X}}\;\|\bm{b}(x) - \bm{b}_0(x)\| \;=\; o(N^{-1/4}).
\]
This is precisely Assumption~\ref{ass:firststage} of the present paper. Hence Assumption~\ref{ass:app:rate} is satisfied.

\noindent
\emph{Note.} Assumption~\ref{ass:app:rateunif} (the uniform-in-$q$ rate) is \emph{not} verified at this step because the present proof targets the pointwise-in-$q$ Proposition~\ref{cor:lp}; in the LP application the first stage $\widehat{\bm{b}}(x)$ does not depend on $q$, so when the uniform-in-$q$ Proposition~\ref{prop:app:uniformq} of~\S A.7 is invoked, the same Assumption~\ref{ass:firststage} of the main text trivially supplies Assumption~\ref{ass:app:rateunif} as well.

\paragraph{Verification of Assumption~\ref{ass:app:reg} (Regularity Conditions).}
Assumption~\ref{ass:app:reg} requires two conditions:
\begin{enumerate}
\item[(i)] \emph{Bounded moments:} $\sup_{\xi \in \Xi_N}\,\sup_{t \in T}\,\sup_{x\in\mathcal{X}}\,|\rho(W, t, \xi)| \leq B_\rho$ almost surely.
 
In our setting, $|\rho(W, \nu, \xi_0)| = |\nu'\mathbf{B}| \leq \|\nu\|\cdot\|\mathbf{B}\|$. Since $\mathcal{T}$ is a finite vertex set, $\max_{\nu\in\mathcal{T}}\|\nu\| \leq C_\nu < \infty$. By Assumption~\ref{ass:secmom}, $\E[\|\mathbf{B}\|^2\mid X=x] \leq \bar{B}$ a.s.
The condition (i) is therefore satisfied with $B_\rho = C_\nu \cdot \bar{B}^{1/2}$ when $\mathbf{B}$ is bounded a.s.\ (as in the Jobs First application, where $\mathbf{B}$ consists of indicator-based IPW signals). More generally, Assumption~\ref{ass:app:reg}(i) can be replaced by the weaker condition $\sup_{t\in T}\sup_{x\in\mathcal{X}} \E[\rho^2(W, t, \xi_0)\mid X=x] \leq B_\rho$, which follows directly from Assumption~\ref{ass:secmom}.
 
\item[(ii)] \emph{Bounded derivative:} $\sup_{\bm{b}\in B_N}\,\sup_{x\in\mathcal{X}}\,\sup_{\nu\in\mathcal{T}}\,\|\partial\phi(\nu, \bm{b}(x))/\partial \bm{b}\| \leq B_\phi$.
 
Since $\phi(\nu, \bm{b}) = \nu'\bm{b}$ is linear in $\bm{b}$, the gradient is $\partial\phi/\partial\bm{b} = \nu$. Hence $\|\partial\phi/\partial\bm{b}\| = \|\nu\| \leq C_\nu$ for all $\nu \in \mathcal{T}$, and the condition holds with $B_\phi = C_\nu = \max_{\nu\in\mathcal{T}}\|\nu\|$.
\end{enumerate}
 
\paragraph{Verification of Assumption~\ref{ass:app:ma} (Margin Condition).}
Assumption~\ref{ass:app:ma} requires that there exist finite positive constants $\bar{B}, \delta > 0$ such that for all $t \in (0, \delta)$,
\[
\sup_{(\nu_1, \nu_2)\in\mathcal{T},\;\nu_1\neq\nu_2}
\Pr\!\big(0 \leq \phi(\nu_1, \bm{b}_0(X)) - \phi(\nu_2, \bm{b}_0(X)) \leq t\big)
\;\leq\; \bar{B}\, t.
\]
Since $\phi(\nu, \bm{b}_0(x)) = \nu'\bm{b}_0(x)$, this becomes
\[
\sup_{\nu_1\neq\nu_2\in\mathcal{T}}
\Pr\!\big(0 < (\nu_1 - \nu_2)'\bm{b}_0(X) \leq t\big)
\;\leq\; \bar{B}\,t.
\]
For any $\nu_1\neq\nu_2 \in \mathcal{T}$, define $\delta_{\nu} = (\nu_1-\nu_2)/\|\nu_1-\nu_2\|$ (a unit vector) and $\underline{c} = \min_{\nu_1\neq\nu_2\in\mathcal{T}}\|\nu_1-\nu_2\| > 0$ (strictly positive since $\mathcal{T}$ is finite with distinct elements). Then
\begin{align*}
\Pr\!\big(0 < (\nu_1-\nu_2)'\bm{b}_0(X) \leq t\big)
&= \Pr\!\Big(0 < \delta_\nu'\bm{b}_0(X) \leq \frac{t}{\|\nu_1-\nu_2\|}\Big)\\
&\leq \Pr\!\Big(0 < \delta_\nu'\bm{b}_0(X) \leq \frac{t}{\underline{c}}\Big)\\
&\leq \sup_{v\in\mathbf{R}^k,\,\|v\|=1}\,\Pr\!\Big(0 < v'\bm{b}_0(X) \leq \frac{t}{\underline{c}}\Big)\\
&\leq \frac{B_f}{\underline{c}}\, t,
\end{align*}
where the last inequality uses Assumption~\ref{ass:mamain} with argument $t/\underline{c}$. Therefore, Assumption~\ref{ass:app:ma} is satisfied with $\bar{B} = B_f/\underline{c}$.

Having verified Assumptions~\ref{ass:app:smallbias}--\ref{ass:app:ma}, we invoke the self-contained Lemmas~\ref{lem:app:closedform}--\ref{lem:app:boot}.

\end{proof}

\subsection{Uniform-in-$q$ asymptotic theory}
\label{sec:uniformq}

This subsection develops a uniform-in-$q$ analog of Proposition~\ref{cor:lp}. The statement specializes the envelope-regression framework of~\S A.1 to $T=\overline{\mathcal{T}}$, $\nu_0(x)=\bm{b}_0(x)$, and $\psi_0=\sigma(q)$, and strengthens Assumption~\ref{ass:app:rate} to its uniform counterpart, Assumption~\ref{ass:app:rateunif}. All other assumptions (\ref{ass:app:smallbias}, \ref{ass:app:reg}, \ref{ass:app:ma}) are used in their existing form: as noted in~\S A.2. The result is stated for the LP setting so that it can be invoked directly as a uniform-in-$q$ strengthening of Proposition~\ref{cor:lp} of the main text.

\begin{proposition}[Uniform-in-$q$ analog of Proposition~\ref{cor:lp}]
\label{prop:app:uniformq}
Suppose Assumptions~\ref{ass:app:smallbias}, \ref{ass:app:rateunif}, \ref{ass:app:reg}, and~\ref{ass:app:ma} hold, and that Assumptions~\ref{ass:boundary} and~\ref{ass:mamain} hold uniformly over $q\in\mathcal{S}^{d-1}$. Then the cross-fitted plug-in estimator $\widehat\sigma(q)$ of Definition~\ref{def:dual} satisfies:

\begin{enumerate}
\item \emph{Uniform consistency:}
\begin{align}
\label{eq:app:urate}
\sup_{q\in\mathcal{S}^{d-1}}\big|\widehat\sigma(q)-\sigma(q)\big| \;=\; O_P(N^{-1/2}).
\end{align}

\item \emph{Uniform Gaussian approximation:} the process
\[
S_N(q)\;:=\;\sqrt{N}\big(\widehat\sigma(q)-\sigma(q)\big),\qquad q\in\mathcal{S}^{d-1},
\]
converges weakly in $\ell^\infty(\mathcal{S}^{d-1})$ to a tight, centered Gaussian process $\mathbb{G}$ with covariance kernel
\[
\Cov\!\big(\mathbb{G}(q), \mathbb{G}(q')\big) \;=\; \E\!\big[(\nu_0(X, q)'\mathbf{B}-\sigma(q))\,(\nu_0(X, q')'\mathbf{B}-\sigma(q'))\big].
\]

\item \emph{Uniform bootstrap validity:} the multiplier-bootstrap process
\[
\widetilde S_N(q)\;:=\;\sqrt{N}\big(\widetilde\sigma(q)-\widehat\sigma(q)\big),\qquad q\in\mathcal{S}^{d-1},
\]
converges conditionally (in probability) weakly in $\ell^\infty(\mathcal{S}^{d-1})$ to the same Gaussian process $\mathbb{G}$.
\end{enumerate}
\end{proposition}

\begin{proof}
The proof extends the three-step argument underlying Lemmas~\ref{lem:app:closedform}--\ref{lem:app:boot} by verifying that the oracle expansion, the CLT, and the bootstrap remainder can each be controlled uniformly in $q\in\mathcal{S}^{d-1}$.

\paragraph{Uniform oracle expansion.}
For each fixed $q$, Lemma~\ref{lem:app:closedform} gives the oracle expansion
\[
\sqrt{N}\,\big(\widehat\sigma(q)-\sigma(q)\big)
\;=\; N^{-1/2}\sum_{i=1}^N\!\big(\nu_0(X_i, q)'\mathbf{B}_i - \sigma(q)\big) + R_N(q),\qquad R_N(q)=o_P(1).
\]
Inspecting the proof of Lemma~\ref{lem:app:closedform}, $R_N(q)$ decomposes as a bias term and a second-order stochastic term:
\[
R_N^{\mathrm{bias}}(q) \;=\; O\big(B_N + \sqrt{N}(\nu^\infty_N)^2\big),
\qquad R_N^{\mathrm{var}}(q) \;=\; O_P\big(\sqrt{\Lambda_N + \nu^\infty_N}\big).
\]
Under Assumption~\ref{ass:app:rateunif} (uniform rate), the rate $\nu^\infty_N$ is $q$-free; under Assumption~\ref{ass:app:smallbias} (stated uniformly over $t\in T=\overline{\mathcal{T}}$), $B_N$ and $\Lambda_N$ are likewise $q$-free. The margin condition Assumption~\ref{ass:app:ma} holds uniformly over $T=\overline{\mathcal{T}}$ and therefore controls the event $\{\widehat\nu_0(X_i, q)\neq\nu_0(X_i, q)\}$ uniformly in $q$. The bias term is then automatically uniform in $q$ because it depends on $q$ only through these $q$-free quantities; the second-order stochastic term inherits uniform-in-$q$ control from the fact that $\nu_0(\cdot, q)$ takes values in the finite set $\overline{\mathcal{T}}$, so $R_N^{\mathrm{var}}(\cdot)$ is bounded above by a sample-independent envelope $\bar R_N := \max_{\nu\in\overline{\mathcal{T}}} R_N^{\mathrm{var}}\bigl|_{\nu_0\equiv\nu}$ that itself satisfies $\bar R_N = O_P\bigl(\sqrt{\Lambda_N + \nu^\infty_N}\bigr) = o_P(1)$. Hence
\[
\sup_{q\in\mathcal{S}^{d-1}}\big|R_N(q)\big| \;=\; o_P(1).
\]

\paragraph{Donsker property of the leading empirical process.}
Define the indexing class
\[
\mathcal{F} \;=\; \big\{\, f_q(W) \;=\; \nu_0(X, q)'\mathbf{B} - \sigma(q) \;:\; q\in\mathcal{S}^{d-1}\,\big\}.
\]
For each fixed $x$, the map $q\mapsto\nu_0(x, q)=\arg\min_{\nu\in\mathcal{T}(q)}\nu'\bm{b}_0(x)$ is piecewise constant on $\mathcal{S}^{d-1}$ and takes values in the finite set $\overline{\mathcal{T}}$; the $q$-partition on which it is constant is determined by finitely many hyperplane crossings in $q$-space. Hence
\[
f_q(W) \;=\; \sum_{\nu\in\overline{\mathcal{T}}} \big(\nu'\mathbf{B} - \sigma(q)\big)\,\mathbf{1}\!\big\{\nu_0(X, q)=\nu\big\},
\]
and the class $\mathcal{F}$ is contained in a finite union of products of indicator classes (one per vertex $\nu\in\overline{\mathcal{T}}$) and affine functions of $q$. Each indicator class $\{\mathbf{1}\{\nu_0(\cdot, q)=\nu\}:q\in\mathcal{S}^{d-1}\}$ is a VC subgraph class with VC dimension bounded by a constant depending only on the dimensions $(d, k)$ of the LP. The envelope $|f_q|\le \max_{\nu\in\overline{\mathcal{T}}}\|\nu\|\cdot\|\mathbf{B}\| + |\sigma(q)|$ has finite second moment under Assumption~\ref{ass:secmom}. By Donsker theorem for VC subgraph classes combined with  stability under finite linear operations and finite unions, $\mathcal{F}$ is $P$-Donsker, so
\begin{equation}
\label{eq:gaussapprox}
\mathbb{G}_N \;=\; N^{-1/2}\sum_{i=1}^N\!\big(f_q(W_i) - \E f_q(W)\big) \;\rightsquigarrow\; \mathbb{G}\ \ \text{in}\ \ \ell^\infty(\mathcal{S}^{d-1}),
\end{equation}
with the covariance kernel stated in 2.

\paragraph{Uniform consistency.}
immediately follows from uniform Gaussian approximation: weak convergence to a tight Gaussian process implies $\sup_{q\in\mathcal{S}^{d-1}}|S_N(q)|=O_P(1)$, and hence $\sup_q|\widehat\sigma(q)-\sigma(q)| = O_P(N^{-1/2})$.

\paragraph{Uniform multiplier-bootstrap validity.}
Because $\mathcal{F}$ is a $P$-Donsker class with a square-integrable envelope, the multiplier CLT for VC classes delivers, conditionally on the data in probability,
\[
\widetilde{\mathbb{G}}_N \;=\; N^{-1/2}\sum_{i=1}^N \Big(\tfrac{e_i}{\bar e}-1\Big) f_q(W_i) \;\rightsquigarrow\; \mathbb{G} \ \ \text{in}\ \ \ell^\infty(\mathcal{S}^{d-1}),
\]
with the same Gaussian limit $\mathbb{G}$ as in~\ref{eq:gaussapprox}. It remains to show that replacing the oracle $f_q(W_i)$ by its cross-fitted plug-in $\widehat\nu(X_i, q)'\mathbf{B}_i$ leaves the bootstrap process unchanged in the limit. The discrepancy is
\[
\Delta_N(q) \;=\; N^{-1/2}\sum_{i=1}^N\!\Big(\tfrac{e_i}{\bar e}-1\Big)\,\bigl[\widehat\nu(X_i, q) - \nu_0(X_i, q)\bigr]'\mathbf{B}_i.
\]
Conditional on the data, $\Delta_N(q)$ is mean zero (since $\E[e_i/\bar e - 1\,|\,\{W_j\}] = 0$), and its conditional variance satisfies
\[
\Var\!\bigl(\Delta_N(q)\,\big|\,\{W_j\}\bigr)
\;\le\; \tfrac{1}{N}\sum_{i=1}^N \bigl|[\widehat\nu(X_i, q) - \nu_0(X_i, q)]'\mathbf{B}_i\bigr|^2 \cdot \widehat{\Var}\!\bigl(\tfrac{e_i}{\bar e}\bigr).
\]
The factor $\widehat{\Var}(e_i/\bar e) = 1 + o_P(1)$ by the LLN. The summand is non-zero only on the misclassification event $\{\widehat\nu(X_i, q)\neq\nu_0(X_i, q)\}$, which is contained in the margin event $\mathcal{E}_\tau = \{0<\tau_0(X_i, q)\le 2 B_\phi \nu_N^\infty\}$ defined in the proof of Lemma~\ref{lem:app:closedform}; on $\mathcal{E}_\tau$, the integrand is bounded by $4(\max_{\nu\in\overline{\mathcal{T}}}\|\nu\|)^2\,\|\mathbf{B}_i\|^2$, an envelope with finite second moment under Assumption~\ref{ass:secmom}. By Assumption~\ref{ass:app:ma}, $\Pr(\mathcal{E}_\tau) = O(\nu_N^\infty)$ uniformly over the finite vertex set $\overline{\mathcal{T}}$; hence the empirical analog satisfies $N^{-1}\sum_i \mathbf{1}\{\mathcal{E}_\tau\} = O_P(\nu_N^\infty)$ uniformly in $q$. Combining these bounds,
\[
\sup_{q\in\mathcal{S}^{d-1}}\Var\bigl(\Delta_N(q)\,\big|\,\{W_j\}\bigr) \;=\; O_P\bigl(\nu_N^\infty\bigr) \;=\; o_P(1),
\]

and a conditional Markov inequality on the finite-cell envelope yields $\sup_{q\in\mathcal{S}^{d-1}}|\Delta_N(q)| = o_P(1)$ in probability over the data. The plug-in bootstrap process therefore shares the same weak limit $\mathbb{G}$ as the oracle bootstrap process, giving 3.
\end{proof}

\renewcommand{\theequation}{B.\arabic{equation}}
\renewcommand{\thelemma}{B.\arabic{lemma}}
\renewcommand{\theassumption}{B.\arabic{assumption}}
\renewcommand{\theremark}{B.\arabic{remark}}
\renewcommand{\thesection}{B}
\setcounter{equation}{0}
\setcounter{assumption}{0}
\setcounter{lemma}{0}
\setcounter{remark}{0}
\setcounter{figure}{0}
\renewcommand{\thefigure}{B.\arabic{figure}}

\setcounter{table}{0}
\renewcommand{\thetable}{B.\arabic{table}}

\section[Online Supplement B: Additional Figures and Empirical Results]{Additional Figures and Empirical Results}
\label{sec:figsrobust}

\subsection{Simulation Exercise}
\label{subsec:simu}
We build a synthetic dataset similar to the Connecticut Jobs First experiment, consisting of baseline covariates, a randomly assigned treatment and a categorical observable state defined by the woman's earnings bin and on/off-welfare participation for each observation unit. The covariate vector \(X=(X_1,\ldots,X_{28})\) includes demographic indicators, eight pre-RA earnings quarters, eight pre-RA AFDC payment quarters, and applicant-status/years-employed. The treatment is assigned independently of covariates and potential outcomes, with $D \sim \mathrm{Bernoulli}(0.5)$, mirroring the balanced randomization of the Jobs First experiment. More details on the data generating process can be found in \href{https://github.com/gev26/clpbounds}{Github Repository}. The population set size is $N = 2500$ with $T=5$ per-person quarters.

Table~\ref{tab:simulation_bounds} reports the simulation results for relevant parameters, where the first-stage conditional moment $\bm b_0(x)$ is estimated with three estimators and base set of covariates. Inferential coverage of 95\% confidence intervals is computed via the person-clustered \(\operatorname{Exp}(1)\) multiplier bootstrap with 200 draws. Across all configurations the set-identified bounds include the true parameter and are tighter than the analytical bound, which we compute following the closed-form expressions of \citet{KT}.

\begin{table}[htbp]

\caption{\textsc{Set-identified response probabilities: simulation exercise}}
\label{tab:simulation_bounds}

\begin{minipage}{0.98\linewidth}
\setlength{\tabcolsep}{5.5pt}
\renewcommand{\arraystretch}{1.15}

\resizebox{\linewidth}{!}{%
\begin{tabular}{@{}cccccccc@{}}
\toprule
\toprule
\multicolumn{2}{c}{State occupied under}
& Symbol
& True
& Analytical
& OLS
& Ridge
& LASSO \\
\cmidrule(lr){1-2}
AFDC
& JF
&
& $\pi_0$
& Bounds
& Estimate
& Estimate
& Estimate \\
\midrule

\multicolumn{8}{l}{\emph{GroupKFold cross-fitted first-stage estimation}} \\
\addlinespace[0.2em]

$0n$ & $1r$ & $\pi_{0n,1r}$ & $0.511$ & $\{0.000, 1.000\}$ & $\{0.137, 1.000\}$ & $\{0.144, 1.000\}$ & $\{0.309, 1.000\}$ \\

$2n$ & $1r$ & $\pi_{2n,1r}$ & $0.559$ & $\{0.209, 1.000\}$ & $\{0.259, 1.000\}$ & $\{0.254, 1.000\}$ & $\{0.249, 1.000\}$ \\

$1n$ & $1r$ & $\pi_{1n,1r}$ & $0.542$ & $\{0.310, 1.000\}$ & $\{0.363, 1.000\}$ & $\{0.378, 1.000\}$ & $\{0.421, 1.000\}$ \\

\bottomrule
\end{tabular}%
}

\vspace{0.45em}

\noindent
\parbox{\linewidth}{%
\footnotesize
\emph{Notes:} Number of state refers to earnings level, with \(0\) indicating no earnings, \(1\) indicating earnings below three times the monthly FPL, and \(2\) indicating earnings above three times the monthly FPL. 
\(n\) indicates welfare nonparticipation, \(r\) indicates welfare participation with truthful reporting of earnings. Numbers in braces are estimated lower and upper bounds. 
Analytical bounds column reports sharp bounds derived via closed-form expressions in \citet{KT}.
}

\end{minipage}
\end{table}

\subsection{Block Expansion Illustration}

Figure~\ref{fig:design-matrices} shows sequential partitioning of the coarse design matrix. Observable states are pairs (earnings bin, welfare regime). Earnings bins follow \citet{KT}: $0$ (zero), $1$ (positive, $\le \text{FPL}$),
$2$ ($>\text{FPL}$). The regime suffix is $n$ off welfare and $p$
on welfare; $r$ signifies truthful reporting, while $u$ underreporting. $0p$ is the on-welfare zero earners
and $2u$ flags above-FPL cells on welfare, which can only be reached by underreporting. State $1r$ is dropped from the constraint system as redundant. $1r$ coincides with $1p$ under JF, but not under ADFC where agents have underreporting incentives below FPL and hence $1p$ includes both latent types $1r$ and $1u$. The granular figures refine bin~$1$ into deciles $b_1,\dots, b_5$ and bin~$2$ into $b_6, b_7, b_8$, yielding the $13$-row set $\{0n, b_in, b_kn, 0p, b_kp\}$ for $i\!=\!1,\dots,5$ and $k\!=\!6, 7, 8$.
The nine groups $G_1,\dots,G_9$ index the coarse transitions of the \citet{KT} model and are preserved under refinement:
\begin{align*}
G_1 &\!:0n\!\to\!1r, & G_2 &\!:0r\!\to\!0n, & G_3 &\!:2n\!\to\!1r, \\
G_4 &\!:0r\!\to\!2n, & G_5 &\!:0r\!\to\!1r, & G_6 &\!:0r\!\to\!1n, \\
G_7 &\!:1n\!\to\!1r, & G_8 &\!:0r\!\to\!2u, & G_9 &\!:2u\!\to\!1r.
\end{align*}
Substantively, $\{G_1,G_3,G_7\}$ are take-up margins,
$\{G_2,G_4,G_6\}$ exit margins from the hub~$0p$, and $\{G_8,G_9\}$ the underreporting margin opened by Jobs First. Each column carries $+1$ at the destination row (blue), $-1$ at the source row (red), or both. Figure~\ref{fig:design-matrices} Panel (a) displays the coarse coefficient matrix $A$
($5\!\times\!9$) of~\eqref{eq:Acoarse}, with rows indexed by
$\{0n, 1n, 2n, 0p, 2p\}$. Panel (b) splits each refined source state into its sub-bins to obtain the 
$13\!\times\!25$ design; a further split of
the $1r$-destined columns of $G_1$ and $G_5$ yields the $13\!\times\!33$ design. Panel (d) is
$13\!\times\!45$ design that splits every $1r$ destination.

\begin{figure}[htbp]
\centering
\includegraphics[width=0.65\textwidth]{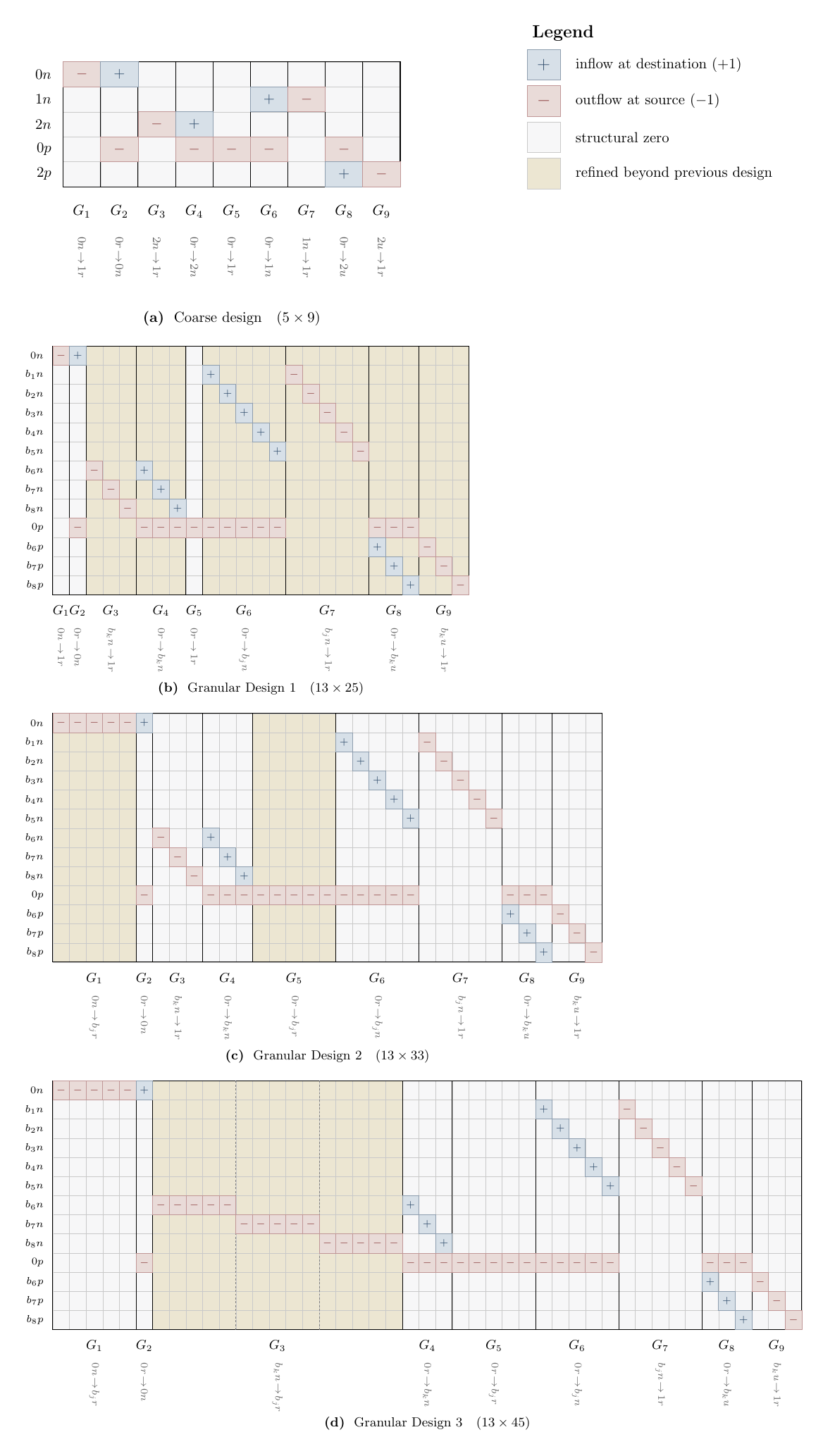}
\caption{\textsc{Possible block expansions of the coarse design matrix.}}
\label{fig:design-matrices}
\end{figure}

\subsection{Additional Results for Coarse Design}

We report supplementary results that extend the analysis of Table~\ref{tab:main} in Section~\ref{sec:emp}. Table~\ref{tab:kt_lasso_appendix} reports results with LASSO first-stage, using two feature sets (baseline and extended set) and two cross-fitting schemes (observation-level and person-level). Since the data possess a panel structure with multiple person-quarter observations per individual under observation-level partitioning, observations from the same individual may be assigned to both the training and held-out folds of a given split, inducing a form of leakage that compromises the orthogonality between the first-stage estimator and the held-out fold. Person-level partitioning assigns the full sequence of quarterly observations for each individual to a single fold to eliminate the leakage but reduces the effective sample size within each fold.

\begin{table}[htbp]
\centering
\caption{\textsc{Set-identified response probabilities: Jobs First case study, $5 \times 9$ coarse design}}
\label{tab:kt_lasso_appendix}

\scriptsize
\setlength{\tabcolsep}{2.8pt}
\renewcommand{\arraystretch}{1.08}

\begin{threeparttable}

\begin{adjustbox}{max width=\textwidth}
\begin{minipage}{\textwidth}

\begin{tabular*}{\textwidth}{@{\extracolsep{\fill}}llcccccc@{}}
\toprule
\toprule
& & \multicolumn{2}{c}{State occupied under}
& & & \multicolumn{2}{c}{Our estimate} \\
\cmidrule(lr){3-4}
\cmidrule(lr){7-8}
\makecell[l]{Response\\type}
& 
& AFDC
& JF
& Symbol
& \makecell{Kline \& Tartari\\reported bounds}
& \makecell{LASSO\\GroupKFold}
& \makecell{LASSO\\KFold} \\
\midrule

\multicolumn{8}{l}{\emph{Panel A. Base covariate regime}} \\

\multirow{9}{*}{Detailed}
& \multirow{9}{*}{$\left\{\vphantom{\begin{array}{c}0\\0\\0\\0\\0\\0\\0\\0\\0\end{array}}\right.$}
& $0n$ & $1r$ & $\pi_{0n,1r}$ & $\{0.055,0.620\}$ & $\{0.000,0.589\}$ & $\{0.047,0.593\}$ \\
& & $1n$ & $1r$ & $\pi_{1n,1r}$ & $\{0.382,0.987\}$ & $\{0.377,1.000\}$ & $\{0.377,0.961\}$ \\
& & $2n$ & $1r$ & $\pi_{2n,1r}$ & $\{0.280,1.000\}$ & $\{0.277,1.000\}$ & $\{0.270,1.000\}$ \\
& & $0r$ & $0n$ & $\pi_{0r,0n}$ & $\{0.000,0.170\}$ & $\{0.000,0.163\}$ & $\{0.000,0.164\}$ \\
& & $''$ & $1n$ & $\pi_{0r,1n}$ & $\{0.000,0.170\}$ & $\{0.000,0.179\}$ & $\{0.000,0.164\}$ \\
& & $''$ & $2n$ & $\pi_{0r,2n}$ & $\{0.000,0.154\}$ & $\{0.001,0.180\}$ & $\{0.000,0.163\}$ \\
& & $''$ & $1r$ & $\pi_{0r,1r}$ & $\{0.000,0.170\}$ & $\{0.000,0.179\}$ & $\{0.000,0.164\}$ \\
& & $''$ & $2u$ & $\pi_{0r,2u}$ & $\{0.031,0.051\}$ & $\{0.031,0.210\}$ & $\{0.031,0.196\}$ \\
& & $2u$ & $1r$ & $\pi_{2u,1r}$ & $\{0.000,1.000\}$ & $\{0.000,1.000\}$ & $\{0.000,1.000\}$ \\

\addlinespace[0.4em]

\multirow{3}{*}{Composite}
& \multirow{3}{*}{$\left\{\vphantom{\begin{array}{c}0\\0\\0\end{array}}\right.$}
& \makecell{Not\\working} 
& Working 
& $\pi_{0,1+}$ 
& $0.167$ 
& $0.160$ 
& $0.160$ \\

& & \makecell{Off\\welfare} 
& \makecell{On\\welfare} 
& $\pi_{n,p}$ 
& $\{0.231,0.445\}$ 
& $\{0.205,0.428\}$ 
& $\{0.223,0.428\}$ \\

& & \makecell{On welfare,\\not working} 
& \makecell{Off\\welfare} 
& $\pi_{0r,n}$ 
& $\{0.000,0.170\}$ 
& $\{0.000,0.163\}$ 
& $\{0.000,0.163\}$ \\

\addlinespace[0.7em]

\multicolumn{8}{l}{\emph{Panel B. Extended covariate regime}} \\

\multirow{9}{*}{Detailed}
& \multirow{9}{*}{$\left\{\vphantom{\begin{array}{c}0\\0\\0\\0\\0\\0\\0\\0\\0\end{array}}\right.$}
& $0n$ & $1r$ & $\pi_{0n,1r}$ & $\{0.055,0.620\}$ & $\{0.000,0.590\}$ & $\{0.047,0.588\}$ \\
& & $1n$ & $1r$ & $\pi_{1n,1r}$ & $\{0.382,0.987\}$ & $\{0.377,1.000\}$ & $\{0.378,0.957\}$ \\
& & $2n$ & $1r$ & $\pi_{2n,1r}$ & $\{0.280,1.000\}$ & $\{0.274,1.000\}$ & $\{0.274,1.000\}$ \\
& & $0r$ & $0n$ & $\pi_{0r,0n}$ & $\{0.000,0.170\}$ & $\{0.000,0.164\}$ & $\{0.000,0.163\}$ \\
& & $''$ & $1n$ & $\pi_{0r,1n}$ & $\{0.000,0.170\}$ & $\{0.000,0.180\}$ & $\{0.000,0.163\}$ \\
& & $''$ & $2n$ & $\pi_{0r,2n}$ & $\{0.000,0.154\}$ & $\{0.000,0.180\}$ & $\{0.000,0.163\}$ \\
& & $''$ & $1r$ & $\pi_{0r,1r}$ & $\{0.000,0.170\}$ & $\{0.000,0.180\}$ & $\{0.000,0.163\}$ \\
& & $''$ & $2u$ & $\pi_{0r,2u}$ & $\{0.031,0.051\}$ & $\{0.031,0.211\}$ & $\{0.031,0.194\}$ \\
& & $2u$ & $1r$ & $\pi_{2u,1r}$ & $\{0.000,1.000\}$ & $\{0.000,1.000\}$ & $\{0.000,1.000\}$ \\

\addlinespace[0.4em]

\multirow{3}{*}{Composite}
& \multirow{3}{*}{$\left\{\vphantom{\begin{array}{c}0\\0\\0\end{array}}\right.$}
& \makecell{Not\\working} 
& Working 
& $\pi_{0,1+}$ 
& $0.167$ 
& $0.160$ 
& $0.160$ \\

& & \makecell{Off\\welfare} 
& \makecell{On\\welfare} 
& $\pi_{n,p}$ 
& $\{0.231,0.445\}$ 
& $\{0.204,0.428\}$ 
& $\{0.224,0.428\}$ \\

& & \makecell{On welfare,\\not working} 
& \makecell{Off\\welfare} 
& $\pi_{0r,n}$ 
& $\{0.000,0.170\}$ 
& $\{0.000,0.163\}$ 
& $\{0.000,0.163\}$ \\

\bottomrule
\end{tabular*}

\begin{tablenotes}[flushleft]
\scriptsize
\item \emph{Notes:} Number of state refers to earnings level, with \(0\) indicating no earnings, \(1\) indicating earnings below three times the monthly FPL, and \(2\) indicating earnings above three times the monthly FPL. 
\(n\) indicates welfare nonparticipation, \(r\) indicates welfare participation with truthful reporting of earnings, \(u\) indicates welfare participation with underreporting of earnings, and \(p\) indicates welfare participation irrespective of reporting. 
Numbers in braces are estimated lower and upper bounds. 
The Kline \& Tartari reported bounds column reproduces the corresponding estimates from their Table 5. 
The LASSO columns report CLP bounds obtained from LASSO first-stage estimation under GroupKFold and KFold cross-fitting. 
Panel A uses the base covariate regime; Panel B uses the extended covariate regime.
\end{tablenotes}

\end{minipage}
\end{adjustbox}

\end{threeparttable}
\end{table}

\subsection{Additional Results for Granular Designs}
\label{subsec:specs}

Among the configurations considered above, LASSO delivers the strongest first-stage performance under the granular partition. We therefore fix LASSO as the first-stage estimator in the remaining robustness exercises and vary instead the partition. Below we summarize the structure of specifications under consideration, with Figure~\ref{fig:design1} showing design matrices for each case. Table~\ref{tab:pi_bounds_modeD} reports the set identified intervals for the parameters of interest under each specification, together with the composite bounds on $\pi_{2n, 1r}$ and $\pi_{1n, 1r}$ evaluated under the partitioned regime but with composite $q$ (See Table~\ref{tab:main} and~\ref{tab:granular-lasso-bounds} \textit{Notes}). 

\begin{enumerate}
    \item \textbf{Spec.~1} ($5\times13$). Source states $1n$ and $2n$ are not subdivided; the destination state $1r$ is split into the fine bins $b_{1}r-b_{5}r$ for the $1n\!\to\!1r$ flow.\\ 
    \item \textbf{Spec.~2} ($7\times11$). Source states not subdivided; the destination state $1r$ is split into restricted bins $l_{1}r-l_{3}r$ for the $2n\!\to\!1r$ flow. In this design, we will make an extra assumption. Recall that $1r$ type has both truthful reporters and underreporters as latent types under AFDC. We will assume that no underreporters exist in the two lowest income bins $b_1n$ and $b_2n$. This will allow to add row constraints for these states. Because of this extra assumption, the lower bound on composite $2n \to 1r$ is significantly higher at $56.4\%$. \\
    \item \textbf{Spec.~3} ($7\times25$). The $2n\!\to\!1r$ flow is subdivided on both the source side ($b_{6}n-b_{8}n$) and the destination side ($b_{1}r-b_{5}r$).\\
    \item \textbf{Spec.~4} ($9\times25$). Source-side $2n$ split into $b_{6}n-b_{8}n$ and $1n$ into $l_{1}n-l_{3}n$. \\  
    \item \textbf{Spec.~5} ($11\times29$). Source-side $2n$ split into $b_{6}n-b_{8}n$ and $1n$ into $b_{1}n-b_{5}n$. \\
    \item \textbf{Spec.~6} ($11\times29$). Full source-side split of $1n$ into $l_{1}n-l_{3}n$ and of $2n$ into $b_{6}n-b_{8}n$.  \\
    \item \textbf{Spec.~7} ($11\times53$). Most granular design. Source-side fine bins $b_{1}n-b_{8}n$ combined with destination-side fine bins $b_{1}r-b_{5}r$ of the state $1r$.  \\
    \item \textbf{Spec.~8} ($13\times33$). Full source-side fine binning $1n$: $b_{1}n-b_{5}n$; $2n$: $b_{6}n-b_{8}n$.

\end{enumerate}

\begin{figure}[htbp]
\centering
\includegraphics[width=1.1\textwidth]{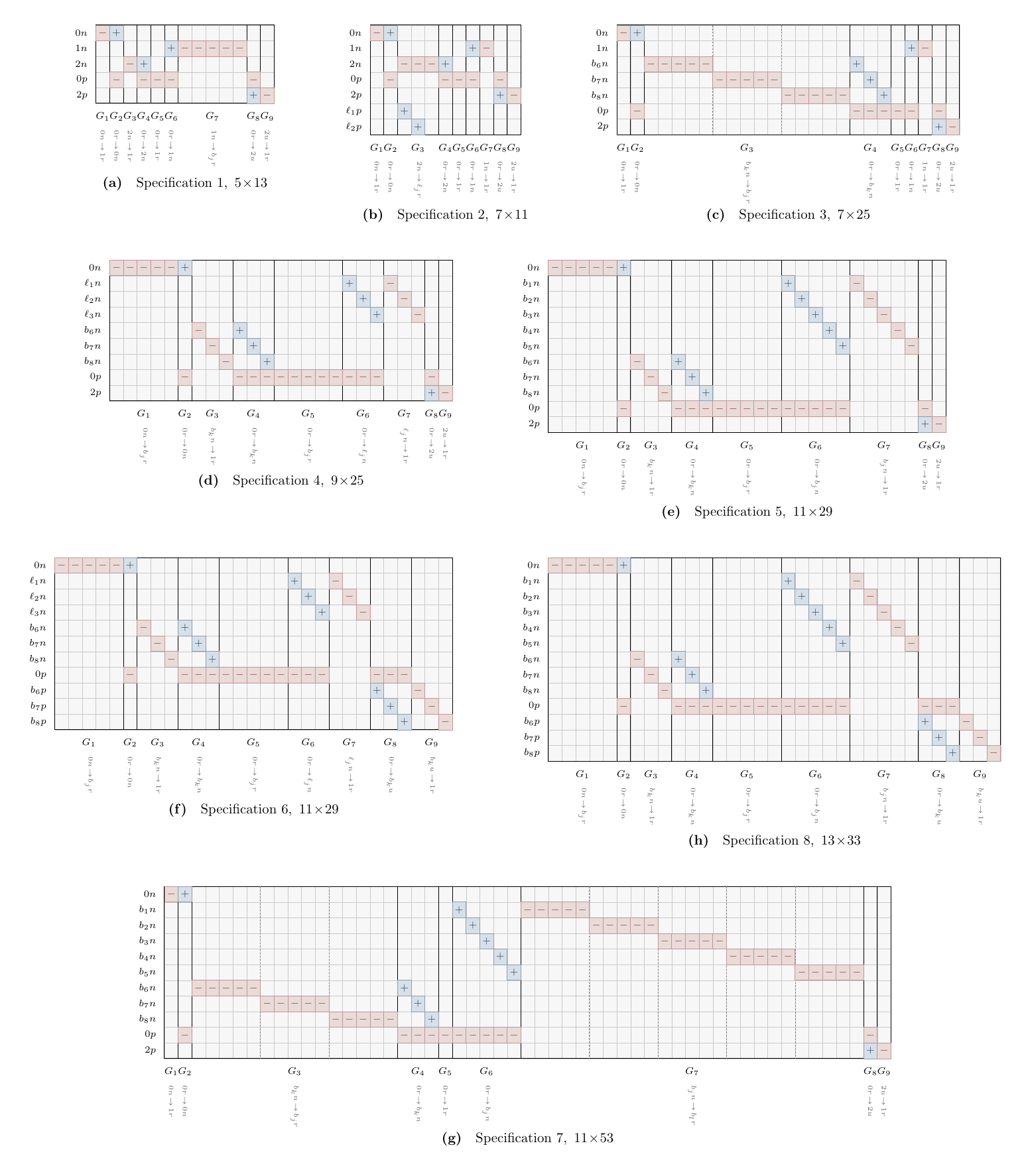}
\caption{\textsc{Different granular specification designs corresponding to specifications in Table~\ref{tab:pi_bounds_modeD}.}}
\label{fig:design1}
\end{figure}

\captionsetup[table]{labelfont=bf,font=normalsize,justification=centering}
\newsavebox{\pitbl}
\newsavebox{\pitblL}
\newsavebox{\pitblR}

\begin{table}[p]
\centering
\caption{\textsc{Set-identified response probabilities: Jobs First case study, all specifications}}
\label{tab:pi_bounds_modeD}
\scriptsize
\setlength{\tabcolsep}{4pt}
\renewcommand{\arraystretch}{0.92}
\sbox{\pitblL}{\begin{tabular}[t]{@{}l c c c c@{}}
\toprule
 & \multicolumn{2}{c}{State occupied under} & & \\
\cmidrule(lr){2-3}
Flow & AFDC & JF & Symbol & $\pi$-bounds \\
\midrule
\multicolumn{5}{@{}l}{\textit{Panel A. Specification~1} (design matrix $5\times13$).} \\
\ldelim\{{5}{14mm}[$\,1n\!\to\!1r$] & $1n$ & $b_{1}r$ & $\pi_{1n,\,b_{1}r}$ & $\{0.000,\;1.000\}$ \\
 & $1n$ & $b_{2}r$ & $\pi_{1n,\,b_{2}r}$ & $\{0.000,\;1.000\}$ \\
 & $1n$ & $b_{3}r$ & $\pi_{1n,\,b_{3}r}$ & $\{0.000,\;1.000\}$ \\
 & $1n$ & $b_{4}r$ & $\pi_{1n,\,b_{4}r}$ & $\{0.000,\;1.000\}$ \\
 & $1n$ & $b_{5}r$ & $\pi_{1n,\,b_{5}r}$ & $\{0.000,\;1.000\}$ \\
\cmidrule(l){2-5}
 & $2n$ & $1r$ & $\pi_{2n,\,1r}$ & $\{0.272,\;1.000\}$ \\
 & $1n$ & $1r$ & $\pi_{1n,\,1r}$ & $\{0.375,\;1.000\}$ \\
\midrule
\multicolumn{5}{@{}l}{\textit{Panel B. Specification~2} (design matrix $7\times11$).} \\
\ldelim\{{3}{14mm}[$\,2n\!\to\!1r$] & $2n$ & $l_{1}r$ & $\pi_{2n,\,l_{1}r}$ & $\{0.248,\;0.248\}$ \\
 & $2n$ & $l_{2}r$ & $\pi_{2n,\,l_{2}r}$ & $\{0.322,\;0.322\}$ \\
 & $2n$ & $l_{3}r$ & $\pi_{2n,\,l_{3}r}$ & $\{0.000,\;0.501\}$ \\
\cmidrule(l){2-5}
 & $2n$ & $1r$ & $\pi_{2n,\,1r}$ & $\{0.564,\;1.000\}$ \\
 & $1n$ & $1r$ & $\pi_{1n,\,1r}$ & $\{0.374,\;0.760\}$ \\
\midrule
\multicolumn{5}{@{}l}{\textit{Panel C. Specification~3} (design matrix $7\times25$).} \\
\ldelim\{{8}{14mm}[$\,2n\!\to\!1r$] & $b_{6}n$ & $1r$ & $\pi_{b_{6}n,\,1r}$ & $\{0.294,\;1.000\}$ \\
 & $b_{7}n$ & $1r$ & $\pi_{b_{7}n,\,1r}$ & $\{0.326,\;1.000\}$ \\
 & $b_{8}n$ & $1r$ & $\pi_{b_{8}n,\,1r}$ & $\{0.192,\;1.000\}$ \\
 & $2n$ & $b_{1}r$ & $\pi_{2n,\,b_{1}r}$ & $\{0.000,\;1.000\}$ \\
 & $2n$ & $b_{2}r$ & $\pi_{2n,\,b_{2}r}$ & $\{0.000,\;1.000\}$ \\
 & $2n$ & $b_{3}r$ & $\pi_{2n,\,b_{3}r}$ & $\{0.000,\;1.000\}$ \\
 & $2n$ & $b_{4}r$ & $\pi_{2n,\,b_{4}r}$ & $\{0.000,\;1.000\}$ \\
 & $2n$ & $b_{5}r$ & $\pi_{2n,\,b_{5}r}$ & $\{0.000,\;1.000\}$ \\
\cmidrule(l){2-5}
 & $2n$ & $1r$ & $\pi_{2n,\,1r}$ & $\{0.257,\;1.000\}$ \\
 & $1n$ & $1r$ & $\pi_{1n,\,1r}$ & $\{0.375,\;1.000\}$ \\
\midrule
\multicolumn{5}{@{}l}{\textit{Panel D. Specification~4} (design matrix $9\times25$).} \\
\ldelim\{{3}{14mm}[$\,2n\!\to\!1r$] & $b_{6}n$ & $1r$ & $\pi_{b_{6}n,\,1r}$ & $\{0.294,\;1.000\}$ \\
 & $b_{7}n$ & $1r$ & $\pi_{b_{7}n,\,1r}$ & $\{0.326,\;1.000\}$ \\
 & $b_{8}n$ & $1r$ & $\pi_{b_{8}n,\,1r}$ & $\{0.192,\;1.000\}$ \\
\ldelim\{{3}{14mm}[$\,1n\!\to\!1r$] & $l_{1}n$ & $1r$ & $\pi_{l_{1}n,\,1r}$ & $\{0.000,\;1.000\}$ \\
 & $l_{2}n$ & $1r$ & $\pi_{l_{2}n,\,1r}$ & $\{0.235,\;1.000\}$ \\
 & $l_{3}n$ & $1r$ & $\pi_{l_{3}n,\,1r}$ & $\{0.464,\;1.000\}$ \\
\cmidrule(l){2-5}
 & $2n$ & $1r$ & $\pi_{2n,\,1r}$ & $\{0.257,\;1.000\}$ \\
 & $1n$ & $1r$ & $\pi_{1n,\,1r}$ & $\{0.350,\;1.000\}$ \\
\midrule
\multicolumn{5}{@{}l}{\textit{Panel E. Specification~5} (design matrix $11\times29$).} \\
\ldelim\{{3}{14mm}[$\,2n\!\to\!1r$] & $b_{6}n$ & $1r$ & $\pi_{b_{6}n,\,1r}$ & $\{0.294,\;1.000\}$ \\
 & $b_{7}n$ & $1r$ & $\pi_{b_{7}n,\,1r}$ & $\{0.326,\;1.000\}$ \\
 & $b_{8}n$ & $1r$ & $\pi_{b_{8}n,\,1r}$ & $\{0.192,\;1.000\}$ \\
\ldelim\{{5}{14mm}[$\,1n\!\to\!1r$] & $b_{1}n$ & $1r$ & $\pi_{b_{1}n,\,1r}$ & $\{0.000,\;1.000\}$ \\
 & $b_{2}n$ & $1r$ & $\pi_{b_{2}n,\,1r}$ & $\{0.235,\;1.000\}$ \\
 & $b_{3}n$ & $1r$ & $\pi_{b_{3}n,\,1r}$ & $\{0.550,\;1.000\}$ \\
 & $b_{4}n$ & $1r$ & $\pi_{b_{4}n,\,1r}$ & $\{0.474,\;1.000\}$ \\
 & $b_{5}n$ & $1r$ & $\pi_{b_{5}n,\,1r}$ & $\{0.360,\;1.000\}$ \\
\cmidrule(l){2-5}
 & $2n$ & $1r$ & $\pi_{2n,\,1r}$ & $\{0.257,\;1.000\}$ \\
 & $1n$ & $1r$ & $\pi_{1n,\,1r}$ & $\{0.341,\;1.000\}$ \\
\bottomrule
\end{tabular}}%
\sbox{\pitblR}{\begin{tabular}[t]{@{}l c c c c@{}}
\toprule
 & \multicolumn{2}{c}{State occupied under} & & \\
\cmidrule(lr){2-3}
Flow & AFDC & JF & Symbol & $\pi$-bounds \\
\midrule
\multicolumn{5}{@{}l}{\textit{Panel F. Specification~6} (design matrix $11\times29$).} \\
\ldelim\{{3}{14mm}[$\,2n\!\to\!1r$] & $b_{6}n$ & $1r$ & $\pi_{b_{6}n,\,1r}$ & $\{0.300,\;1.000\}$ \\
 & $b_{7}n$ & $1r$ & $\pi_{b_{7}n,\,1r}$ & $\{0.328,\;1.000\}$ \\
 & $b_{8}n$ & $1r$ & $\pi_{b_{8}n,\,1r}$ & $\{0.204,\;1.000\}$ \\
\ldelim\{{3}{14mm}[$\,1n\!\to\!1r$] & $l_{1}n$ & $1r$ & $\pi_{l_{1}n,\,1r}$ & $\{0.000,\;1.000\}$ \\
 & $l_{2}n$ & $1r$ & $\pi_{l_{2}n,\,1r}$ & $\{0.235,\;1.000\}$ \\
 & $l_{3}n$ & $1r$ & $\pi_{l_{3}n,\,1r}$ & $\{0.465,\;1.000\}$ \\
\cmidrule(l){2-5}
 & $2n$ & $1r$ & $\pi_{2n,\,1r}$ & $\{0.265,\;1.000\}$ \\
 & $1n$ & $1r$ & $\pi_{1n,\,1r}$ & $\{0.351,\;1.000\}$ \\
 \midrule
\multicolumn{5}{@{}l}{\textit{Panel G. Specification~7} (design matrix $11\times53$).} \\
\ldelim\{{8}{14mm}[$\,2n\!\to\!1r$] & $b_{6}n$ & $1r$ & $\pi_{b_{6}n,\,1r}$ & $\{0.294,\;1.000\}$ \\
 & $b_{7}n$ & $1r$ & $\pi_{b_{7}n,\,1r}$ & $\{0.326,\;1.000\}$ \\
 & $b_{8}n$ & $1r$ & $\pi_{b_{8}n,\,1r}$ & $\{0.192,\;1.000\}$ \\
 & $2n$ & $b_{1}r$ & $\pi_{2n,\,b_{1}r}$ & $\{0.000,\;1.000\}$ \\
 & $2n$ & $b_{2}r$ & $\pi_{2n,\,b_{2}r}$ & $\{0.000,\;1.000\}$ \\
 & $2n$ & $b_{3}r$ & $\pi_{2n,\,b_{3}r}$ & $\{0.000,\;1.000\}$ \\
 & $2n$ & $b_{4}r$ & $\pi_{2n,\,b_{4}r}$ & $\{0.000,\;1.000\}$ \\
 & $2n$ & $b_{5}r$ & $\pi_{2n,\,b_{5}r}$ & $\{0.000,\;1.000\}$ \\
\ldelim\{{10}{14mm}[$\,1n\!\to\!1r$] & $b_{1}n$ & $1r$ & $\pi_{b_{1}n,\,1r}$ & $\{0.000,\;1.000\}$ \\
 & $b_{2}n$ & $1r$ & $\pi_{b_{2}n,\,1r}$ & $\{0.235,\;1.000\}$ \\
 & $b_{3}n$ & $1r$ & $\pi_{b_{3}n,\,1r}$ & $\{0.550,\;1.000\}$ \\
 & $b_{4}n$ & $1r$ & $\pi_{b_{4}n,\,1r}$ & $\{0.474,\;1.000\}$ \\
 & $b_{5}n$ & $1r$ & $\pi_{b_{5}n,\,1r}$ & $\{0.360,\;1.000\}$ \\
 & $1n$ & $b_{1}r$ & $\pi_{1n,\,b_{1}r}$ & $\{0.000,\;1.000\}$ \\
 & $1n$ & $b_{2}r$ & $\pi_{1n,\,b_{2}r}$ & $\{0.000,\;1.000\}$ \\
 & $1n$ & $b_{3}r$ & $\pi_{1n,\,b_{3}r}$ & $\{0.000,\;1.000\}$ \\
 & $1n$ & $b_{4}r$ & $\pi_{1n,\,b_{4}r}$ & $\{0.000,\;1.000\}$ \\
 & $1n$ & $b_{5}r$ & $\pi_{1n,\,b_{5}r}$ & $\{0.000,\;1.000\}$ \\
\cmidrule(l){2-5}
 & $2n$ & $1r$ & $\pi_{2n,\,1r}$ & $\{0.257,\;1.000\}$ \\
 & $1n$ & $1r$ & $\pi_{1n,\,1r}$ & $\{0.341,\;1.000\}$ \\
\midrule
\multicolumn{5}{@{}l}{\textit{Panel H. Specification~8} (design matrix $13\times33$).} \\
\ldelim\{{3}{14mm}[$\,2n\!\to\!1r$] & $b_{6}n$ & $1r$ & $\pi_{b_{6}n,\,1r}$ & $\{0.300,\;1.000\}$ \\
 & $b_{7}n$ & $1r$ & $\pi_{b_{7}n,\,1r}$ & $\{0.328,\;1.000\}$ \\
 & $b_{8}n$ & $1r$ & $\pi_{b_{8}n,\,1r}$ & $\{0.204,\;1.000\}$ \\
\ldelim\{{5}{14mm}[$\,1n\!\to\!1r$] & $b_{1}n$ & $1r$ & $\pi_{b_{1}n,\,1r}$ & $\{0.000,\;1.000\}$ \\
 & $b_{2}n$ & $1r$ & $\pi_{b_{2}n,\,1r}$ & $\{0.235,\;1.000\}$ \\
 & $b_{3}n$ & $1r$ & $\pi_{b_{3}n,\,1r}$ & $\{0.550,\;1.000\}$ \\
 & $b_{4}n$ & $1r$ & $\pi_{b_{4}n,\,1r}$ & $\{0.474,\;1.000\}$ \\
 & $b_{5}n$ & $1r$ & $\pi_{b_{5}n,\,1r}$ & $\{0.365,\;1.000\}$ \\
\cmidrule(l){2-5}
 & $2n$ & $1r$ & $\pi_{2n,\,1r}$ & $\{0.265,\;1.000\}$ \\
 & $1n$ & $1r$ & $\pi_{1n,\,1r}$ & $\{0.343,\;1.000\}$ \\
\addlinespace[21.55pt]
\bottomrule
\end{tabular}}%
\typeout{MEASURE LEFT=\the\dimexpr\ht\pitblL+\dp\pitblL\relax RIGHT=\the\dimexpr\ht\pitblR+\dp\pitblR\relax}%
\sbox{\pitbl}{%
\begin{tabular}{@{}c@{\hspace{6pt}}!{\vrule width 0.7pt}@{\hspace{6pt}}c@{}}
\usebox{\pitblL} & \usebox{\pitblR}
\end{tabular}}%
\usebox{\pitbl}\par
\vspace{5pt}
\begin{minipage}{\wd\pitbl}
\fontsize{7}{8.4}\selectfont
\textit{Notes.} Each panel corresponds to one specification, ordered from least to most granular by the dimension of the LASSO design matrix (rows~$\times$~columns), reported in each panel heading. The eight specifications are summarized in \S B.4 and Figure~\ref{fig:design1}. State labels follow \citet{KT}. Within a coarse state, $b_{k}n$ and $b_{k}r$ denote fine sub-bins, and $l_{k}n$, $l_{k}r$ denote restricted (lower) fine sub-bins (see spec2 in \S B.4). Sub-bin rows are bracketed to indicate the coarse flow to which they aggregate (by source bin, $b_{k}n\!\to\!1r$, and/or by destination bin, $2n/1n\!\to\!b_{k}r$); the two rows below each bracket, $\pi_{2n,1r}$ and $\pi_{1n,1r}$, report the bounds for the aggregated coarse bin under the same specification. Each composite bound is obtained by setting the vector $q$ to one
on the coordinates of the sub-bins that compose the coarse bin and to zero elsewhere. Numbers in braces are estimated lower and upper bounds and clipped to $[0,1]$ when
needed. In the coarse model the dual feasible region has few enough vertices that they can be enumerated, and each bound is obtained by checking over this finite set. In granular models the number of vertices grows combinatorially with the additional states and parameters—so enumeration is costly, and we instead obtain each bound by solving the dual program directly with a numerical solver, imposing $\|\nu\|_{\infty} \leq 200$ on the dual variable. We exclude any program that fails to solve and any program whose optimal $\nu$ attains the box boundary. In practice this filter removes less than $2\%$ of observations.\\[3pt]
\end{minipage}
\end{table}

\subsection{Additional Results for Welfare Bounds}

Under the most granular design reported in the main text, the bound on the net welfare gain from the $2n \to 1r$ transition is inconclusive in sign, with an identified interval that contains zero. We treat this granular result, rather than its coarse counterpart, as the primary one. The ambiguous sign is not an artifact of weak identification and reflects genuine heterogeneity in the welfare consequences of the transition across narrow earnings cells.
To see the mechanism, consider two women who both opt into welfare under Jobs First but originate in different cells. A woman in bin $b_6n$, earning just above the federal poverty line, who moves to $b_5r$, just below it, sacrifices little earned income while gaining the transfer; her welfare change is positive and potentially large. A woman in bin $b_8n$, earning well above the poverty line, who moves to $b_2r$ forgoes substantial earnings that the transfer might not offset. Both transitions are subsumed within the single coarse flow $2n \to 1r$ in coarser design cases that do not differentiate between those two cells, applying a common welfare weight to them and resulting in an average of those effects. More granular designs represent them as distinct cells and on aggregation the welfare-gain bound straddles zero precisely because the underlying cell-level effects do.

\vspace{1cm}

\setlength{\LTcapwidth}{\textwidth}
\small
\renewcommand{\arraystretch}{0.9}
\begin{longtable}{lll|lll}
\caption{\textsc{Set-identified welfare bounds: Jobs First case study, some specifications}}
\label{tab:welfare_bounds}\\
\toprule
Flow & Transition & Welfare $\$$/m  & Flow & Transition & Welfare $\$$/m\\
\midrule
\endfirsthead
\multicolumn{6}{l}{\textit{(Table~\ref{tab:welfare_bounds} continued)}}\\
\toprule
Flow & Transition & Welfare $\$$/m & Flow & Transition & Welfare $\$$/m\\
\midrule
\endhead
\bottomrule
\endfoot
\multicolumn{3}{l|}{\textit{Panel A. Specification 1 ($5\times13$)}} & \multicolumn{3}{l}{\textit{Panel D. Specification 7 ($11\times53$)}}\\
$2n\to1r$ & $2n\!\to\!1r$ & $\{-95.95,\ -23.70\}$ & $2n\to1r$ & $b_{6n}\!\to\!b_{1r}$ & $\{-74.58,\ 0.00\}$\\
 & Coarse $2n\to1r$ & $\{-95.95,\ -23.70\}$ &  & $b_{6n}\!\to\!b_{2r}$ & $\{-49.53,\ 0.00\}$\\
\cmidrule(lr){1-3}
$1n\to1r$ & $1n\!\to\!1r$ & $\{11.37,\ 30.56\}$ &  & $b_{6n}\!\to\!b_{3r}$ & $\{-26.10,\ 0.00\}$\\
 & Coarse $1n\to1r$ & $\{11.37,\ 30.56\}$ &  & $b_{6n}\!\to\!b_{4r}$ & $\{-3.16,\ 0.00\}$\\
\addlinespace[5pt]
\multicolumn{3}{l|}{\textit{Panel B. Specification 2 ($7\times11$)}} &  & $b_{6n}\!\to\!b_{5r}$ & $\{0.00,\ 20.72\}$\\
$2n\to1r$ & $2n\!\to\!\ell_{1r}$ & $\{-30.48,\ -30.48\}$ &  & $b_{7n}\!\to\!b_{1r}$ & $\{-97.49,\ 0.00\}$\\
 & $2n\!\to\!\ell_{2r}$ & $\{-31.29,\ -31.29\}$ &  & $b_{7n}\!\to\!b_{2r}$ & $\{-72.41,\ 0.00\}$\\
 & $2n\!\to\!\ell_{3r}$ & $\{-28.78,\ 0.30\}$ &  & $b_{7n}\!\to\!b_{3r}$ & $\{-48.96,\ 0.00\}$\\
 & Coarse $2n\to1r$ & $\{-90.56,\ -61.47\}$ &  & $b_{7n}\!\to\!b_{4r}$ & $\{-25.99,\ 0.00\}$\\
\cmidrule(lr){1-3}
$1n\to1r$ & $1n\!\to\!1r$ & $\{11.35,\ 23.07\}$ &  & $b_{7n}\!\to\!b_{5r}$ & $\{-2.09,\ 0.00\}$\\
 & Coarse $1n\to1r$ & $\{11.35,\ 23.07\}$ &  & $b_{8n}\!\to\!b_{1r}$ & $\{-163.33,\ 0.00\}$\\
\addlinespace[5pt]
\multicolumn{3}{l|}{\textit{Panel C. Specification 8 ($13\times33$)}} &  & $b_{8n}\!\to\!b_{2r}$ & $\{-138.32,\ 0.00\}$\\
$2n\to1r$ & $b_{6n}\!\to\!1r$ & $\{-39.74,\ -3.58\}$ &  & $b_{8n}\!\to\!b_{3r}$ & $\{-114.95,\ 0.00\}$\\
 & $b_{7n}\!\to\!1r$ & $\{-62.59,\ -5.62\}$ &  & $b_{8n}\!\to\!b_{4r}$ & $\{-92.05,\ 0.00\}$\\
 & $b_{8n}\!\to\!1r$ & $\{-129.28,\ -11.85\}$ &  & $b_{8n}\!\to\!b_{5r}$ & $\{-68.22,\ 0.00\}$\\
 & Coarse $2n\to1r$ & $\{-138.49,\ -21.05\}$ &  & Coarse $2n\to1r$ & $\{-178.64,\ 14.63\}$\\
\cmidrule(lr){1-3}\cmidrule(lr){4-6}
$1n\to1r$ & $b_{1n}\!\to\!1r$ & $\{-2.49,\ 71.70\}$ & $1n\to1r$ & $b_{1n}\!\to\!b_{1r}$ & $\{0.00,\ 40.83\}$\\
 & $b_{2n}\!\to\!1r$ & $\{2.47,\ 53.64\}$ &  & $b_{1n}\!\to\!b_{2r}$ & $\{0.00,\ 62.90\}$\\
 & $b_{3n}\!\to\!1r$ & $\{4.72,\ 33.65\}$ &  & $b_{1n}\!\to\!b_{3r}$ & $\{0.00,\ 83.54\}$\\
 & $b_{4n}\!\to\!1r$ & $\{1.42,\ 9.30\}$ &  & $b_{1n}\!\to\!b_{4r}$ & $\{0.00,\ 103.75\}$\\
 & $b_{5n}\!\to\!1r$ & $\{-16.27,\ -2.09\}$ &  & $b_{1n}\!\to\!b_{5r}$ & $\{0.00,\ 124.79\}$\\
 & Coarse $1n\to1r$ & $\{-10.15,\ 78.22\}$ &  & $b_{2n}\!\to\!b_{1r}$ & $\{0.00,\ 20.19\}$\\
 &  &  &  & $b_{2n}\!\to\!b_{2r}$ & $\{0.00,\ 44.14\}$\\
 &  &  &  & $b_{2n}\!\to\!b_{3r}$ & $\{0.00,\ 66.52\}$\\
 &  &  &  & $b_{2n}\!\to\!b_{4r}$ & $\{0.00,\ 88.45\}$\\
 &  &  &  & $b_{2n}\!\to\!b_{5r}$ & $\{0.00,\ 111.27\}$\\
 &  &  &  & $b_{3n}\!\to\!b_{1r}$ & $\{-3.45,\ 0.00\}$\\
 &  &  &  & $b_{3n}\!\to\!b_{2r}$ & $\{0.00,\ 23.12\}$\\
 &  &  &  & $b_{3n}\!\to\!b_{3r}$ & $\{0.00,\ 47.97\}$\\
 &  &  &  & $b_{3n}\!\to\!b_{4r}$ & $\{0.00,\ 72.30\}$\\
 &  &  &  & $b_{3n}\!\to\!b_{5r}$ & $\{0.00,\ 97.63\}$\\
 &  &  &  & $b_{4n}\!\to\!b_{1r}$ & $\{-28.32,\ 0.00\}$\\
 &  &  &  & $b_{4n}\!\to\!b_{2r}$ & $\{-1.36,\ 0.00\}$\\
 &  &  &  & $b_{4n}\!\to\!b_{3r}$ & $\{0.00,\ 23.85\}$\\
 &  &  &  & $b_{4n}\!\to\!b_{4r}$ & $\{0.00,\ 48.54\}$\\
 &  &  &  & $b_{4n}\!\to\!b_{5r}$ & $\{0.00,\ 74.24\}$\\
 &  &  &  & $b_{5n}\!\to\!b_{1r}$ & $\{-52.73,\ 0.00\}$\\
 &  &  &  & $b_{5n}\!\to\!b_{2r}$ & $\{-26.57,\ 0.00\}$\\
 &  &  &  & $b_{5n}\!\to\!b_{3r}$ & $\{-2.10,\ 0.00\}$\\
 &  &  &  & $b_{5n}\!\to\!b_{4r}$ & $\{0.00,\ 21.86\}$\\
 &  &  &  & $b_{5n}\!\to\!b_{5r}$ & $\{0.00,\ 46.80\}$\\
 &  &  &  & Coarse $1n\to1r$ & $\{-58.04,\ 160.94\}$\\
\end{longtable}
 
\vspace{2pt}
\begin{tablenotes}[flushleft]
\scriptsize
\item \emph{Notes:} Each panel corresponds to one specification, ordered from least to most granular (summarized in \S B.4 and Figure~\ref{fig:design1}) by the dimension of the LASSO design matrix (rows~$\times$~columns), reported in each panel heading. State labels follow \citet{KT}. Within a coarse state, $b_{k}n$ and $b_{k}r$ denote fine sub-bins, and $l_{k}n$, $l_{k}r$ denote restricted fine sub-bins (see spec2 in \S B.4). Sub-bin rows are bracketed to indicate the coarse flow to which they aggregate. Each composite bound is obtained by setting the vector $q$ to one
on the coordinates of the sub-bins that compose the coarse bin and to zero elsewhere. Numbers in braces are estimated lower and upper bounds and clipped to $[0,1]$ when
needed. Welfare values are in monthly dollars. In the coarse model the dual feasible region has few enough vertices that they can be enumerated, and each bound is obtained by checking over this finite set. In granular models the number of vertices grows combinatorially with the additional states and parameters—so enumeration is costly, and we instead obtain each bound by solving the dual program directly with a numerical solver, imposing $\|\nu\|_{\infty} \leq 200$ on the dual variable. We exclude any program that fails to solve and any program whose optimal $\nu$ attains the box boundary. In practice this filter removes less than $2\%$ of observations.\\[3pt]
\end{tablenotes}

\end{document}